\documentclass{LMCS}

\def\doi{8 (4:1) 2012}
\lmcsheading%
{\doi}
{1--45}
{}
{}
{Dec.~21, 2006}
{Oct.~\phantom03, 2012}
{}

\usepackage{enumerate}
\usepackage{hyperref}

\renewcommand{\paragraph}{\subsection}

\newif\iflongversion

\longversiontrue

\usepackage{ifpdf}
\ifpdf
\else
\fi
\usepackage{ifthen}
\usepackage{enumerate}
\usepackage{latexsym,amssymb}
\usepackage{amsthm}
\usepackage{stmaryrd}
\usepackage[numbers,sort&compress]{natbib}
\usepackage{alltt}
\usepackage{verbatim} 
\usepackage{proof}
\usepackage{xspace}
\usepackage{url}
\usepackage{bussproofs}

\newcounter{newrule}
\usepackage{amsmath,amssymb,amsthm,amsfonts} 
\usepackage{amsmath,amssymb} 
\newdimen\proofrulebreadth \proofrulebreadth=.05em
\newdimen\proofdotseparation \proofdotseparation=1.25ex
\newdimen\proofrulebaseline \proofrulebaseline=2ex
\newcount\proofdotnumber \proofdotnumber=3
\let\then\relax
\def\hfi{\hskip0pt plus.0001fil}
\mathchardef\squigto="3A3B
\newif\ifinsideprooftree\insideprooftreefalse
\newif\ifonleftofproofrule\onleftofproofrulefalse
\newif\ifproofdots\proofdotsfalse
\newif\ifdoubleproof\doubleprooffalse
\let\wereinproofbit\relax
\newdimen\shortenproofleft
\newdimen\shortenproofright
\newdimen\proofbelowshift
\newbox\proofabove
\newbox\proofbelow
\newbox\proofrulename
\def\shiftproofbelow{\let\next\relax\afterassignment\setshiftproofbelow\dimen0 }
\def\shiftproofbelowneg{\def\next{\multiply\dimen0 by-1 }%
\afterassignment\setshiftproofbelow\dimen0 }
\def\setshiftproofbelow{\next\proofbelowshift=\dimen0 }
\def\setproofrulebreadth{\proofrulebreadth}

\def\prooftree{
\ifnum  \lastpenalty=1
\then   \unpenalty
\else   \onleftofproofrulefalse
\fi
\ifonleftofproofrule
\else   \ifinsideprooftree
        \then   \hskip.5em plus1fil
        \fi
\fi
\bgroup
\setbox\proofbelow=\hbox{}\setbox\proofrulename=\hbox{}%
\let\justifies\proofover\let\leadsto\proofoverdots\let\Justifies\proofoverdbl
\let\using\proofusing\let\[\prooftree
\ifinsideprooftree\let\]\endprooftree\fi
\proofdotsfalse\doubleprooffalse
\let\thickness\setproofrulebreadth
\let\shiftright\shiftproofbelow \let\shift\shiftproofbelow
\let\shiftleft\shiftproofbelowneg
\let\ifwasinsideprooftree\ifinsideprooftree
\insideprooftreetrue
\setbox\proofabove=\hbox\bgroup$\displaystyle 
\let\wereinproofbit\prooftree
\shortenproofleft=0pt \shortenproofright=0pt \proofbelowshift=0pt
\onleftofproofruletrue\penalty1
}

\def\eproofbit{
\ifx    \wereinproofbit\prooftree
\then   \ifcase \lastpenalty
        \then   \shortenproofright=0pt  
        \or     \unpenalty\hfil         
        \or     \unpenalty\unskip       
        \else   \shortenproofright=0pt  
        \fi
\fi
\global\dimen0=\shortenproofleft
\global\dimen1=\shortenproofright
\global\dimen2=\proofrulebreadth
\global\dimen3=\proofbelowshift
\global\dimen4=\proofdotseparation
\global\count255=\proofdotnumber
$\egroup  
\shortenproofleft=\dimen0
\shortenproofright=\dimen1
\proofrulebreadth=\dimen2
\proofbelowshift=\dimen3
\proofdotseparation=\dimen4
\proofdotnumber=\count255
}

\def\proofover{
\eproofbit 
\setbox\proofbelow=\hbox\bgroup 
\let\wereinproofbit\proofover
$\displaystyle
}%
\def\proofoverdbl{
\eproofbit 
\doubleprooftrue
\setbox\proofbelow=\hbox\bgroup 
\let\wereinproofbit\proofoverdbl
$\displaystyle
}%
\def\proofoverdots{
\eproofbit 
\proofdotstrue
\setbox\proofbelow=\hbox\bgroup 
\let\wereinproofbit\proofoverdots
$\displaystyle
}%
\def\proofusing{
\eproofbit 
\setbox\proofrulename=\hbox\bgroup 
\let\wereinproofbit\proofusing
\kern0.3em$
}

\def\endprooftree{
\eproofbit 
  \dimen5 =0pt
\dimen0=\wd\proofabove \advance\dimen0-\shortenproofleft
\advance\dimen0-\shortenproofright
\dimen1=.5\dimen0 \advance\dimen1-.5\wd\proofbelow
\dimen4=\dimen1
\advance\dimen1\proofbelowshift \advance\dimen4-\proofbelowshift
\ifdim  \dimen1<0pt
\then   \advance\shortenproofleft\dimen1
        \advance\dimen0-\dimen1
        \dimen1=0pt
        \ifdim  \shortenproofleft<0pt
        \then   \setbox\proofabove=\hbox{%
                        \kern-\shortenproofleft\unhbox\proofabove}%
                \shortenproofleft=0pt
        \fi
\fi
\ifdim  \dimen4<0pt
\then   \advance\shortenproofright\dimen4
        \advance\dimen0-\dimen4
        \dimen4=0pt
\fi
\ifdim  \shortenproofright<\wd\proofrulename
\then   \shortenproofright=\wd\proofrulename
\fi
\dimen2=\shortenproofleft \advance\dimen2 by\dimen1
\dimen3=\shortenproofright\advance\dimen3 by\dimen4
\ifproofdots
\then
        \dimen6=\shortenproofleft \advance\dimen6 .5\dimen0
        \setbox1=\vbox to\proofdotseparation{\vss\hbox{$\cdot$}\vss}%
        \setbox0=\hbox{%
                \advance\dimen6-.5\wd1
                \kern\dimen6
                $\vcenter to\proofdotnumber\proofdotseparation
                        {\leaders\box1\vfill}$%
                \unhbox\proofrulename}%
\else   \dimen6=\fontdimen22\the\textfont2 
        \dimen7=\dimen6
        \advance\dimen6by.5\proofrulebreadth
        \advance\dimen7by-.5\proofrulebreadth
        \setbox0=\hbox{%
                \kern\shortenproofleft
                \ifdoubleproof
                \then   \hbox to\dimen0{%
                        $\mathsurround0pt\mathord=\mkern-6mu%
                        \cleaders\hbox{$\mkern-2mu=\mkern-2mu$}\hfill
                        \mkern-6mu\mathord=$}%
                \else   \vrule height\dimen6 depth-\dimen7 width\dimen0
                \fi
                \unhbox\proofrulename}%
        \ht0=\dimen6 \dp0=-\dimen7
\fi
\let\doll\relax
\ifwasinsideprooftree
\then   \let\VBOX\vbox
\else   \ifmmode\else$\let\doll=$\fi
        \let\VBOX\vcenter
\fi
\VBOX   {\baselineskip\proofrulebaseline \lineskip.2ex
        \expandafter\lineskiplimit\ifproofdots0ex\else-0.6ex\fi
        \hbox   spread\dimen5   {\hfi\unhbox\proofabove\hfi}%
        \hbox{\box0}%
        \hbox   {\kern\dimen2 \box\proofbelow}}\doll%
\global\dimen2=\dimen2
\global\dimen3=\dimen3
\egroup 
\ifonleftofproofrule
\then   \shortenproofleft=\dimen2
\fi
\shortenproofright=\dimen3
\onleftofproofrulefalse
\ifinsideprooftree
\then   \hskip.5em plus 1fil \penalty2
\fi
}

\usepackage{diagrams}
\diagramstyle[size=2em]

\newcommand{\co}{\colon}
\newcommand{\ld}{.}
\newcommand{\ts}{\vdash}

\newcommand{\fv}[1]{\text{FV}(#1)}

\newcommand{\In}[2]{#1 \colon  \! #2}
\newcommand{\lam}[3]{\lambda \In{#1}{#2}.\: #3}

\newcommand{\codefont}[1]{\mathrm{#1}}

\newcommand{\Ctx}{\codefont{Ctx}}
\newcommand{\Type}{\codefont{Type}}
\newcommand{\Prop}{\codefont{Prop}}

\newcommand{\dsum}[3]{\Sum_{#1 \co \! #2}#3}

\DeclareMathOperator{\Prod}{\textstyle{\prod}}
\DeclareMathOperator{\Coprod}{\textstyle{\coprod}}
\DeclareMathOperator{\Sum}{\textstyle{\sum}}
\newcommand{\opcat}[1]{{{#1}^{\mathrm{op}}}}
\newcommand{\dom}[1]{\mathrm{dom}(#1)}

\newcommand{\from}{\gets}

\newcommand{\Fam}[1]{\mathrm{Fam}(#1)}
\newcommand{\catcat}{\mathbf{Cat}}
\newcommand{\Sets}{\mathbf{Set}}
\newcommand{\CBUlt}{\ensuremath{\mathit{CBUlt}}}
\newcommand{\BiUlt}{\ensuremath{\mathit{BiUlt}}}

\newcommand{\BiCBUlt}{\ensuremath{\mathit{BiCBUlt}}}

\newcommand{\catfont}{\mathbb}
\newcommand{\bcat}{{\catfont{B}}}
\newcommand{\ccat}{{\catfont{C}}}
\newcommand{\dcat}{{\catfont{D}}}

\newcommand{\etopos}{\ensuremath{\mathcal{E}}\xspace}
\newcommand{\stopos}{\ensuremath{\mathcal{S}}\xspace}

\newcommand{\Hom}[3]{\mathrm{Hom}_{#1}({#2},{#3})}
\newcommand{\Sub}[1]{\mathrm{Sub}(#1)}

\newcommand{\inv}[1]{{#1}^{-1}}

\newcommand{\id}{\mathrm{id}}

\newcommand{\charmap}[1]{\chi_{#1}}
\newcommand{\nameof}[1]{\ulcorner #1 \urcorner}

\newcommand{\slice}[2]{#1 / #2}
\newcommand{\eslice}[1]{\etopos / #1}
\newcommand{\p}[1]{p_{#1}} 

\newcommand{\successor}{\mathrm{succ}}
\newcommand{\lradj}[2]{#1 \dashv #2}

\newcommand{\Powset}[1]{\mathcal{P}(#1)}
\newcommand{\finpowset}[1]{\mathcal{P}_{\mathrm{fin}}(#1)}

\newcommand{\denlb}{[\![}
\newcommand{\denrb}{]\!]}
\newcommand{\den}[1]{\denlb{#1}\denrb}

\renewcommand{\iff}{\ensuremath{\mathrel{\Leftrightarrow}}}

\newcommand{\imp}{\to}
\newcommand{\bimp}{\leftrightarrow}
\newcommand{\meet}{\wedge}

\newcommand{\existsunique}{\exists!}
\newcommand{\transClosureAlt}[1]{#1^\omega}

\newcommand{\biimpdefn}{\stackrel{\mathrm{def}}{\Longleftrightarrow}}
\newcommand{\alphaseq}{\overline{\alpha}}

\newcommand{\etal}{\emph{et\ al.}\xspace}

\newcommand{\GAP}{\hspace*{.5cm}}

\newenvironment{arraytl}{\begin{array}[t]{@{}l@{}}}{\end{array}}

\newcommand{\pair}[2]{\langle #1 , #2 \rangle}

\newcommand{\iso}{\cong}

\newcommand{\SItopos}{\mathcal{S}}
\newcommand{\SIslice}[1]{\slice{\SItopos}{#1}}
\newcommand{\later}{\ensuremath{\mathop{\blacktriangleright}}}
\newcommand{\laterSlice}[1]{\mathop{{\later}_{#1}}}
\newcommand{\laterPred}{\ensuremath{\mathop{\rhd}}}

\newcommand{\earlier}{\ensuremath{\mathop{\blacktriangleleft}}}
\newcommand{\earlierSlice}[1]{\earlier_{#1}}
\newcommand{\next}{\mathrm{next}}
\newcommand{\previous}{\mathrm{prev}}
\newcommand{\fix}{\mathrm{fix}}
\newcommand{\Fix}{\mathrm{Fix} \,} 

\newcommand{\omegabar}{\overline{\omega}}
\newcommand{\contractivePred}{\mathrm{Contr}}
\newcommand{\complete}{contractively complete}

\renewcommand{\succ}{\mathrm{succ}}

\newcommand{\forces}[2]{#1 \models #2} 

\newcommand{\PA}{X}
\newcommand{\PB}{Y}
\newcommand{\PC}{Z}
\newcommand{\PAn}[1]{\PSindex{\PA}{#1}}
\newcommand{\PBn}[1]{\PSindex{\PB}{#1}}
\newcommand{\PSindex}[2]{#1(#2)}

\newcommand{\SubPA}{A}

\newcommand{\SubPAn}[1]{\PSindex{\SubPA}{#1}}

\newcommand{\restrmap}[1]{r_{#1}}
\newcommand{\restrict}[2]{#1|_{#2}}

\newcommand{\mlfont}[1]{\mathsf{#1}}

\newcommand{\Fmuref}{\ensuremath{\mlfont{F}_{\mu,\mlfont{ref}}}}

\newcommand{\tofin}{\mathrel{\to_{\mathrm{fin}}}}
\newcommand{\tomon}{\mathrel{\to_{\mathrm{mon}}}}

\newcommand{\W}{\mathcal{W}}
\newcommand{\ST}{\widehat\STalt}
\newcommand{\STalt}{\mathcal{T}}
\newcommand{\STaltC}{\mathcal{T}^{\mathrm{c}}}

\newcommand{\App}{\mathrm{app}}
\newcommand{\Lam}{\mathrm{lam}}

\newcommand{\statesOp}{\operatorname{\mathit{states}}}
\newcommand{\compOp}{\operatorname{\mathit{comp}}}

\newcommand{\ap}{\hskip2pt}

\newcommand{\tp}{\tau}

\newcommand{\comptp}[1]{{#1}}
\newcommand{\arrowtp}[2]{{#1} \mathop\to {#2}}
\newcommand{\reftp}[1]{\operatorname{\mlfont{ref}}{#1}}

\newcommand{\tm}{t}

\newcommand{\unitcst}{()}
\newcommand{\pairtm}[2]{({#1},\, {#2})}
\newcommand{\fst}[1]{\mlfont{fst}\ap {#1}}
\newcommand{\snd}[1]{\mlfont{snd}\ap {#1}}
\newcommand{\void}[1]{\mlfont{void}\ap {#1}}
\newcommand{\inl}[1]{\mlfont{inl}\ap {#1}}
\newcommand{\inr}[1]{\mlfont{inr}\ap {#1}}
\newcommand{\case}[5]{\mlfont{case}\, {#1} \; {#2}.{#3} \;
  {#4}. {#5}}
\newcommand{\fold}[1]{\mlfont{fold}\ap {#1}}
\newcommand{\unfold}[1]{\mlfont{unfold}\ap {#1}}

\newcommand{\reftm}[1]{\mlfont{ref}\ap {#1}}
\newcommand{\fixtm}[3]{\mlfont{fix}\ap {#1}. \lambda{#2}.{#3}}
\newcommand{\lookup}[1]{\mlfont{!}\hskip1pt{#1}}
\newcommand{\assign}[2]{{#1} \mathbin{\mlfont{:=}} {#2}}
\newcommand{\lamtm}[2]{\overline{\lambda} {#1}.\hskip1pt {#2}}
\newcommand{\tplam}[2]{\Lambda {#1}.\hskip1pt {#2}}
\newcommand{\tpapp}[2]{{#1} \ap [{#2}]}

\newcommand{\tpjudge}[4]{{#1} \mid {#2} \vdash {#3} : {#4}}

\newcommand{\mngTp}[3]{\den{#1}{#3}}  
\newcommand{\mngTpC}[3]{\den{{#1}}^\mathrm{c}{#3}}  

\newcommand{\semrel}[4]{{#1} \mid {#2} \models {#3} : {#4}}
\newcommand{\funrel}[3]{{#1} \mathrel{\multimap_{#3}} {#2}}

\newcommand{\Term}{\mathrm{Term}}
\newcommand{\Value}{\mathrm{Value}}
\newcommand{\Store}{\mathrm{Store}}
\newcommand{\Config}{\mathrm{Config}}

\newcommand{\step}{\mathrm{step}}
\newcommand{\eval}{\mathrm{eval}}

\newenvironment{proofof}[1]
  {\vspace{7pt}\par\noindent\textbf{Proof of {#1}.}}
  {\qed\vspace{7pt}\par }

 \renewenvironment{thebibliography}[1]{%
   \begin{oldthebibliography}{#1}%
     \setlength{\parskip}{0ex}%
     \setlength{\itemsep}{0ex}%
 }%
 {%
   \end{oldthebibliography}%
 }

\begin{document}
 
\title[First Steps in Synthetic Guarded Domain Theory]{First steps in
  synthetic guarded domain theory: \\ step-indexing in the topos of
  trees\rsuper*}

\author[L. Birkedal]{Lars Birkedal\rsuper a}
\address{{\lsuper{a,b}}IT University of Copenhagen}
\email{\{birkedal, mogel\}@itu.dk}

\author[R.E. M{\o}gelberg]{Rasmus Ejlers M{\o}gelberg\rsuper b}
\address{\vskip-6 pt}
\email{mogel@itu.dk}

\author[J. Schwinghammer]{Jan Schwinghammer\rsuper c}
\address{{\lsuper c}Saarland University}
\email{jan@ps.uni-saarland.de}

\author[K. St{\o}vring]{Kristian St{\o}vring\rsuper d}
\address{{\lsuper d}DIKU, University of Copenhagen}
\email{stovring@diku.dk}

\titlecomment{{\lsuper*}This is an expanded and revised version of a paper that
  appeared at LICS'11~\cite{BirkedalL:sgdt}.  In addition to
  containing more examples and proofs, it also contains a more general
  treatment of models of synthetic guarded domain theory.}


\begin{abstract}
  We present the topos $\stopos$ of trees as a model of guarded
  recursion.  We study the internal dependently-typed higher-order
  logic of $\stopos$ and show that $\stopos$ models two modal
  operators, on predicates and types, which serve as guards in
  recursive definitions of terms, predicates, and types.  In
  particular, we show how to solve recursive type equations involving
  \emph{dependent} types.  We propose that the internal logic of
  $\stopos$ provides the right setting for the synthetic construction
  of abstract versions of step-indexed models of programming languages
  and program logics. As an example, we show how to construct a model
  of a programming language with higher-order store and recursive
  types entirely inside the internal logic of $\stopos$.  Moreover, we
  give an axiomatic categorical treatment of models of synthetic
  guarded domain theory and prove that, for any complete Heyting
  algebra $A$ with a well-founded basis, the topos of sheaves over $A$
  forms a model of synthetic guarded domain theory, generalizing the
  results for $\stopos$.
\end{abstract}

\maketitle

\section{Introduction}
\label{sec:introduction}


Recursive definitions are ubiquitous in computer science. In
particular, in semantics of programming languages and program logics
we often use recursively defined functions and
relations, and also recursively defined types (domains).
For example, in recent years there has been extensive work on giving
semantics of type systems for programming languages with dynamically
allocated higher-order store, such as general ML-like references. 
Models have been expressed as Kripke models over a recursively
defined set of worlds (an example of a recursively defined domain)
and have involved recursively defined relations to interpret
the recursive types of the programming
language; see~\cite{BirkedalL:essence} and the references therein.

In this paper we study a topos $\stopos$, which we show models guarded
recursion in the sense that it allows for guarded recursive
definitions of both recursive functions and relations as well as
recursive types.  The topos $\stopos$ is known as the topos of trees
(or forests); what is new here is our application of this topos to
model guarded recursion.

The internal logic of $\stopos$ is a standard many-sorted 
higher-order logic extended with modal operators on both types
and terms. (Recall that terms in higher-order logic include both
functions and relations, as the latter are simply $\Prop$-valued functions.)
This internal logic can then be used as a language to describe semantic models 
of programming languages with the features mentioned above. 
As an example which uses both recursively defined types and
recursively defined relations in the $\stopos$-logic, we present 
a model of \Fmuref, a call-by-value programming language with 
impredicative polymorphism, recursive types, and general ML-like
references. 

To situate our work in relation to earlier work, we now give a
quick overview of the technical development of the present paper 
followed by a comparison to related work.  We end the introduction
with a summary of our contributions.

\subsection{Overview of technical development}

The topos $\stopos$ is the category of presheaves on
$\omega$, the first infinite ordinal.  This topos is known as the
topos of trees, and is one of the most basic examples
of presheaf categories.  

There are several ways to think intuitively about this topos. 
Let us recall one intuitive description, which can serve to understand
why it models guarded recursion. An object $X$ of $\stopos$ is a
contravariant functor from $\omega$ (viewed as a preorder) to $\Sets$.
We think of $X$ as a variable set, i.e., a family of sets $X(n)$, indexed over natural
numbers $n$, and with restriction maps $X(n+1)\to X(n)$. 
Morphisms $f: X \to Y$ are natural transformations from $X$ to $Y$.
The variable sets include the ordinary sets as so-called constant
sets:  for an ordinary set $S$, there is an object $\Delta(S)$ in
$\stopos$ with $\Delta(S)(n)=S$  for
all $n$.  Since $\stopos$ is a category of presheaves, it is a topos,
in particular it is cartesian closed category and has a subobject classifier $\Omega$
(a type 
of propositions).  The internal logic of $\stopos$ is an extension of 
standard Kripke semantics: for constant sets, the truth value of a
predicate is just the set of worlds (downwards closed subsets of $\omega$) for
which the predicate holds.  This observation suggests that there is
a modal ``later'' operator $\laterPred$ on predicates $\Omega^{\Delta(S)}$ on
constant sets, similar to what has been studied
earlier~\cite{dreyer+:lics09,Appel:Mellies:Richards:Vouillon:07}. 
Intuitively, for a predicate $\varphi: \Omega^{\Delta(S)}$ on constant set
$\Delta(S)$, $\laterPred(\varphi)$ contains $n+1$ if $\varphi$ contains $n$. (A future
world is a smaller number, hence the name ``later'' for this operator.)
A recursively specified predicate $\mu r. \varphi(r)$ is well-defined if every occurrence of 
the recursion variable $r$ in $\varphi$ is guarded by a $\laterPred$ modality:  by definition of $\laterPred$, to know
whether $n+1$ is in the predicate it suffices to know whether 
$n$ is in the predicate.  
There is also an associated L{\"o}b rule for induction, $(\laterPred \varphi\to \varphi)\to \varphi$, as
in~\cite{Appel:Mellies:Richards:Vouillon:07}.

Here we show that in fact there is a later operator not only on
predicates on constant sets, but also on predicates on general variable
sets, with associated L{\"o}b rule, and well-defined guarded recursive definitions
of predicates.

Moreover, there is also a later operator $\later$ 
(a functor) on the variable sets themselves: $\later(X)$ is given by
$\later(X)(1) = \{ \star \}$ and $\later(X)(n+1)=X(n)$.  
We can show the well-definedness of recursive 
  variable sets $\mu X. F(X)$ in which the recursion variable
$X$ is guarded by this operator $\later$.  Intuitively, such a recursively specified variable
set  is well-defined since by definition of $\later$, 
to know what $\mu X. F(X)$ is at level $n+1$ it suffices to know
what it is at level $n$. 

In the technical sections of the paper, we make the above precise.  In
particular, we detail the internal logic and the use of later on
functions / predicates and on types.  We explain how one can solve
mixed-variance recursive type equations, for a wide collection of
types.  We show how to use the internal logic of $\stopos$ to give a
model of \Fmuref.  The model, including the operational semantics of
the programming language, is defined completely inside the internal
logic; we discuss the connection between the resulting model and earlier models  by relating
internal definitions in the internal logic to standard (external)
definitions.
Since $\stopos$ is a topos, $\stopos$ also models
dependent types.  We give technical semantic results as needed for
using later on dependent types and for recursive type-equations
involving dependent types.  We think of this as a first step towards
a formalized dependent type theory with a later operator; here we
focus on the foundational semantic issues.

To explain the relationship to some of the related work, we point out
that $\stopos$ is equivalent to the category of sheaves on
$\omegabar$, where $\omegabar$ is the complete Heyting algebra of natural numbers with the usual ordering and extended with a top element $\infty$. 
Moreover, this sheaf category, and hence also $\stopos$, is equivalent to the
topos obtained by the tripos-to-topos
construction~\cite{Hyland:Johnstone:Pitts:80}
applied to the tripos $\Sets(\_,\omegabar)$.  The logic of constant
sets in $\stopos$ is exactly the logic of 
this tripos.\footnote{Recall that the tripos $\Sets(\_,\omegabar)$ is a
  model of logic in which types and terms are interpreted as sets and
  functions, and predicates are interpreted as $\omegabar$-valued
  functions.}

In the first part of this paper we work with the presentation of $\stopos$ as
presheaves since it is the most concrete, but in fact many
of our results generalize to sheaf categories over other complete
well-founded Heyting algebras. Indeed, we include a more general 
axiomatic treatment of models of synthetic guarded domain theory
and prove that, for any complete Heyting algebra with a well-founded
basis, the topos of sheaves over the Heyting algebra yields a model
of synthetic guarded domain theory.  We present this
generalization after the more concrete treatment of $\stopos$, 
since the concrete treatment of $\stopos$ is perhaps more accessible.

\subsection{Related work}

Nakano presented a simple type theory with guarded recursive types
\cite{Nakano:00} which can be modelled using complete bounded
ultrametric spaces \cite{Birkedal:Schwinghammer:Stovring:10:Nakano}.
We show in Section~\ref{sec:cbult} that the category $\BiCBUlt$ of
bisected, complete bounded ultrametric spaces is a co-reflective
subcategory of~$\stopos$. Thus, our present work can be seen as an
extension of the work of Nakano to include the full internal language
of a topos, in particular dependent types, and an associated
higher-order logic.  Pottier \cite{Pottier:FORK} presents an extension
of System F with recursive kinds based on Nakano's calculus; hence
\stopos also models the kind language of his system.

Di Gianantonio and Miculan~\cite{DiGianantonio:Miculan:04} studied guarded recursive 
definitions of functions in certain sheaf toposes over well-founded
complete Heyting algebras, thus including $\stopos$. 
Our work extends the work of Di Gianantonio and Miculan by also 
including guarded recursive definitions of \emph{types}, by emphasizing
the use of the internal logic (this was suggested as future work
in~\cite{DiGianantonio:Miculan:04}),
and by including an extensive example application.  
Moreover, our general treatment of sheaf models includes sheaves
over any well-founded complete Heyting algebra, whereas 
Di Gianantonio and Miculan restrict attention to those
Heyting algebras that arise as the opens of a topological space.

Earlier work has 
advocated the use of complete bounded
ultrametric spaces for solving recursive type and relation equations
that come up when modelling programming languages with higher-order
store~\cite{BirkedalL:essence,Schwinghammer:Birkedal:Stovring:11}.  
As mentioned above, $\BiCBUlt$ is a
subcategory of $\stopos$, and thence our present work can be seen as
an improvement of this earlier work: it is an
improvement since $\stopos$ supports full higher-order
logic. In the earlier work, one had to show that the
functions defined in the interpretation of the programming language
types were non-expansive.  Here we take the synthetic approach 
(cf.~\cite{Hyland:90}) and place ourselves in the internal logic of the topos when defining the
interpretation of the programming language, see Section~\ref{sec:application}. 
This means that there is no need to prove properties like
non-expansiveness since, intuitively, all functions in the topos are
suitably non-expansive. 

Dreyer \etal~\cite{dreyer+:lics09} proposed a logic, called LSLR, for defining step-indexed
interpretations of programming languages with recursive
types, building on earlier work by Appel \etal~\cite{Appel:Mellies:Richards:Vouillon:07} who proposed the use of a
later modality on predicates.  The point of LSLR is that it provides
for more abstract ways of constructing and reasoning with step-indexed
models, thus avoiding tedious calculations with step indices.  The
core logic of LSLR is the logic of the tripos $\Sets(\_,\omegabar)$
mentioned above,\footnote{Dreyer \etal~\cite{dreyer+:lics09} presented
  the semantics of their second-order logic in more concrete terms,
  avoiding the use of triposes, but it is indeed a fragment of the
  internal logic of the mentioned tripos.}  which allows for
recursively defined predicates
following~\cite{Appel:Mellies:Richards:Vouillon:07}, but not
recursively defined types.  One point of passing from this tripos to
the topos $\stopos$ is that it gives us a wider collection of types
(variable sets rather than only constant sets), which makes it
possible also to have mixed-variance recursively defined types.\footnote{The
  terminology can be slightly confusing:
  in~\cite{Appel:Mellies:Richards:Vouillon:07}, our notion of
  recursive relations were called recursive types, probably because
  the authors of \emph{loc.cit.} used such to interpret recursive
  types of a programming language. Recursive types in our sense were
  not considered in~\cite{Appel:Mellies:Richards:Vouillon:07}.}

Dreyer \etal~developed an extension of LSLR called LADR for
reasoning about step-indexed models of the programming language
\Fmuref\ with higher-order store~\cite{dreyer+:popl10}.  LADR is a
specialized logic where much of the world structure used for reasoning
efficiently about local state is hidden by the model of the logic;
here we are proposing a general logic that can be used to construct
many step-indexed models, including the one used to model LADR.
In particular, in our example application in Section~\ref{sec:application}, we \emph{define} a set of
worlds inside the $\stopos$ logic, using recursively defined types. 

As part of our analysis of recursive dependent types, we define a
class of types, called \emph{functorial types}.  We show that functorial types are closed
under nested recursive types, a result which is akin to results on
nested inductive types~\cite{abbott:icalp04,dybjer:97}.  The
difference is that we allow for general mixed-variance recursive
types, but on the other hand we require that all occurrences of
recursion variables must be guarded. 

\subsection{Summary of contributions}

We show how the topos $\stopos$, and, more generally, any topos of
sheaves over a complete Heyting algebra with a well-founded basis,
provides a simple but powerful model of guarded recursion, allowing
for guarded recursive definitions of both terms and types in the
internal dependently-typed higher-order logic.  In particular, we
\begin{iteMize}{$\bullet$}
\item show that the two later modalities are well-behaved on slices;  
\item give existence theorems for fixed points of guarded recursive terms 
  and guarded nested dependent mixed-variance recursive types;
\item detail the relation of \stopos to the category of complete bounded ultrametric spaces; 
\item present, as an example application, a synthetic model of
  $\Fmuref$ constructed internally in $\stopos$;
\item give an axiomatic treatment of a general class of models
  of guarded recursion.
\end{iteMize}
Our general existence theorems for recursive types in
Section~\ref{sec:class:of:models} are phrased in terms of
$\textit{Sh}(A)$-categories, i.e., categories enriched in sheaves over a
complete Heyting algebra $A$ with a well-founded basis, and generalize
earlier work on recursive types for categories enriched in complete
bounded ultrametric spaces~\cite{BirkedalL:metric-enriched-journal}.

\section{The $\mathcal{S}$ Topos}
\label{sec:s}
\label{sec:model}


The category $\stopos$ is that of presheaves on $\omega$, the preorder of natural numbers 
starting from 1 and ordered by inclusion.
Explicitly, the objects of $\SItopos = \Sets^{\opcat{\omega}}$ are
families of sets indexed by natural numbers together with restriction
maps $\restrmap{n} \co \PAn{n+1} \to \PAn{n}$. Morphisms are families
$(f_{n})_{n}$ of maps commuting with the restriction maps as indicated
in the diagram
\begin{diagram}[size=2em]
\PAn{1} & \lTo & \PAn{2} & \lTo & \PAn{3} & \lTo & \dots \\
\dTo^{f_{1}} && \dTo^{f_{2}} && \dTo^{f_{3}} & \\
\PBn{1} & \lTo & \PBn{2} & \lTo & \PBn{3} & \lTo & \dots
\end{diagram}
If $x \in \PAn{m}$ and $n \leq m$ we write $\restrict{x}{n}$ for
$\restrmap{n} \circ \dots \circ \restrmap{m-1}(x)$.

As all presheaf categories, $\SItopos$ is a topos, in particular it is
cartesian closed and has a subobject classifier. Moreover, it
is complete and cocomplete, and limits and
colimits are computed pointwise. The $n$'th component of the exponential
$\PSindex{\PB^{\PA}}{n}$ is the set of tuples $(f_{1}, \dots, f_{n})$
commuting with the restriction maps, and the restriction maps of
$\PB^{\PA}$ are given by projection. We sometimes use the notation $\PA \to \PB$ for $\PB^\PA$.

A subobject $\SubPA$ of $\PA$ is a family of subsets $\SubPAn{n}
\subseteq \PAn{n}$ such that $\restrmap{n}(\SubPAn{n+1}) \subseteq
\SubPAn{n}$. The subobject classifier has $\PSindex{\Omega}{n} = \{0,
\dots, n\}$ and restriction maps $\restrmap{n}(x) = \min(n,x)$. The
characteristic morphism $\charmap{\SubPA}\co \PA \to \Omega$ maps $x
\in \PAn{n}$ to the maximal $m$ such that $\restrict{x}{m} \in
\SubPAn{m}$ if such an $m$ exists and $0$ otherwise.

The natural numbers object $N$ in $\stopos$ is the constant set of
natural numbers.

Intuitively, we can think of the set $\PA(n)$ as what the type $\PA$ 
looks like, if one has at most $n$ time steps to reason about it. The
restriction maps $\restrmap{n} \co \PA(n+1) \to \PA(n)$ describe what
happens to the data when one time step passes.
This intuition is illustrated by the following example.

\newcommand{\Stream}{\mathit{Str}}
\newcommand{\Succ}{\mathit{succ}}

\begin{exa}\label{ex:streams:i}
  We can define the object $\Stream\in\stopos$ of \emph{(variable) streams} of natural numbers
  as follows:
  \begin{diagram}[size=2em]
    N^1 & \lTo & N^2 & \lTo & N^3 & \lTo & \dots \\
  \end{diagram}
  where the restriction maps $\restrmap{m} : N^{m+1} \to N^m$ map
  $(n_1,\ldots,n_m, n_{m+1})$ to $(n_1,\ldots,n_m)$.  Intuitively, this is the type
  of streams where the head is immediately available, but the tail is
  only available after one time step. If we have $n$ time steps to reason about this
  type we can access the $n$ first elements, hence $\Stream(n) =
  N^{n}$.
  
  The successor function $\Succ$ on streams, which adds one to every element in a stream,
  can be defined in the model by 
  \begin{displaymath}
    \Succ_m : N^{m} \to N^{m} = (n_1,\ldots,n_m) \mapsto (n_1 + 1, \ldots, n_m + 1).
  \end{displaymath}
  Clearly $\Succ$ is a natural transformation from $\Stream$ to $\Stream$; hence it
  is a well-defined map in $\stopos$.
  Observe that $\Succ_m$ can also be defined by induction as
  $\Succ_1(n) = n+1$ and
  $\Succ_{m+1}(n_1,n_2,\ldots, n_{m+1}) = (n_1 + 1, \Succ_{m}(n_2,\ldots,n_{m+1}))$.

  The subobject $A\subseteq\Stream$ of increasing streams can be defined by
  letting $A_m \subseteq N^{m}$ be the set of tuples $(n_1,\ldots,n_m)$ that
  are increasing (i.e., $n_i > n_j$, for $i>j$). Note that $A$ is trivially closed
  under the restriction maps, and thus it is a well-defined subobject of $\Stream$.
\end{exa}

\subsection{The $\later$ endofunctor} 
Define the functor $\later \co \SItopos \to \SItopos$ by
$\PSindex{\later \PA}{1} = \{ \star \}$ and $\PSindex{\later \PA}{n+1}
= \PAn{n}$. This functor, called \emph{later}, has a left adjoint (so
$\later$ preserves all limits) given by $\PSindex{\earlier \PA}{n} =
\PAn{n+1}$. Since limits are computed pointwise, $\earlier$ preserves
them, and so the adjunction $\lradj{\earlier}{\later}$ defines a
geometric morphism, in fact an embedding. However, we shall not make use of
this fact in the present paper (because $\earlier$ is not a fibred
endo-functor on the codomain fibration, hence is not a useful operator
in the dependent type theory; see Section~\ref{sec:dtt}).

There is a natural transformation $\next_{\PA} \co \PA \to \later \PA$
whose $1$st component is the unique map into $\{ \star \}$ and whose
$(n+1)$th component is $\restrmap{n}$. Although $\next$ looks like a unit
$\later$ is not a monad: there are no natural transformations $\later \later \to \later$. 

Since $\later$ preserves finite limits, there is always a morphism
\begin{equation}
  \label{eq:J}
  J: \later(X\to Y) \to (\later X \to \later Y).  
\end{equation}

\subsection{An operator on predicates} 
There is a morphism
$\laterPred \co \Omega \to \Omega$ mapping $n \in \PSindex{\Omega}{m}$
to $\min(m, n+1)$.  By setting $\charmap{\laterPred\!\SubPA} =
\laterPred\circ \charmap{\SubPA}$ there is an induced operation on
subobjects, again denoted $\laterPred$.
This operation,
which we also call ``later'', is connected to the $\later$ functor, since
there is a pullback diagram
\begin{diagram}[size=2em]
\laterPred m  \SEpbk & \rTo & \later \SubPA \\
\dTo && \dTo_{\later m} \\
X & \rTo^{\next_X} & \later X
\end{diagram}
for any subobject $m \co \SubPA \to X$. 

\subsection{Recursive morphisms}
We introduce a notion of contractive morphism  and show that these   have unique fixed points.

\begin{defi}
  \label{def:contractive}
  A morphism $f \co \PA \to \PB$ is \emph{contractive} if there exists
  a morphism $g\co \later \PA \to \PB$ such that $f = g \circ
  \next_{\PA}$. A morphism $f \co \PA\times \PB \to \PC$ is
  contractive in the first variable if there exists  $g$ such that $f
  = g \circ (\next_{\PA} \times \id_{\PB})$.
\end{defi}
For instance, contractiveness
 of $\laterPred$ on $\Omega$ is witnessed by $\succ\co \later\Omega\to\Omega$ with \label{page:succ}
$\succ_n(k)=k\mathop+1$. 

\begin{lem} \label{lem:contractive:comp}\hfill
\begin{enumerate}[\em(1)]
\item \label{item:contractive:comp} If $f \co X \to Y$ and $g\co Y \to Z$ and either $f$ or $g$ are contractive also $gf$ is contractive.
\item \label{item:contractive:prod} If $f\co X \to Y$ and $g \co X' \to Y'$ are contractive, so is $f \times g \co X \times X' \to Y \times Y'$. 
\item \label{item:contractive:curry} A morphism $h \co X\times Y \to Z$ is contractive in the first variable iff $\hat h \co X \to Z^{Y}$ is  contractive. 
\end{enumerate}
\end{lem}

\noindent If $f \co \PA \to \PB$ is contractive as witnessed by $g$, the value of $f_{n+1}(x)$
can be computed from $\restrmap{n}(x)$ and moreover, $f_1$ must be
constant. If $\PA = \PB$ we can define a fixed point $x \co 1 \to \PA$
by defining $x_1 = g_1(\star)$ and $x_{n+1} = g_{n+1}(x_n)$. This
construction can be generalized to include fixed points of morphisms
with parameters as follows.

\begin{thm} \label{thm:fp:op}
  There exists a natural family of morphisms $\fix_{\PA} \co ({\later
    \PA} \to \PA) \to \PA$, indexed by the collection of all objects
  $\PA$, which computes unique fixed points in the sense that if $f
  \co \PA\times \PB \to \PA$ is contractive in the first variable as
  witnessed by $g$, i.e., $f = g \circ (\next_{\PA} \times \id_{\PB})$,
  then $\fix_{\PA} \circ \hat{g}$ is the unique $h\co Y\to X$ such that $f \circ
  \pair{h}{\id_{\PB}} = h$ (here $\hat{g}$ denotes the exponential transpose of $g$).
\end{thm}

\subsection{Internal logic}
\label{sec:internal-logic}

We start by calling to mind parts of the Kripke-Joyal forcing semantics for $\stopos$.
For $X_1, \dots, X_m$ in $\stopos$,\, $\varphi:X_1\times\cdots\times X_m\to\Omega$,\,  $n\in\omega$, and
$\alpha_1\in X_1(n), \ldots,\alpha_m\in X_m(n)$, we define
$\forces{n}{\varphi(\alpha_1,\ldots,\alpha_m)}$ iff
$\varphi_n(\alpha_1,\ldots,\alpha_m) = n.$

The standard clauses for the forcing relation are as
follows~\cite[Example~9.5]{Lambek:Scott:86}
(we write $\alphaseq$ for a sequence $\alpha_1,\ldots,\alpha_m$):
\begin{align*}
  \forces{n}{(s=t)}{\alphaseq}  
     &  \iff  \den{s}_n(\alphaseq) = \den{t}_n(\alphaseq)   \\
  \forces{n}{R(t_1, \ldots,t_k)}{\alphaseq}  
    &  \iff  \den{R}_n(\den{t_1}_n(\alphaseq),\ldots,\den{t_k}_n(\alphaseq)) = n  \\
  \forces{n}{(\varphi\land\psi)(\alphaseq)}
    & \iff   \forces{n}{\varphi (\alphaseq)} \land \forces{n}{\psi(\alphaseq)} \\
  \forces{n}{(\varphi \lor\psi)(\alphaseq)}
   & \iff   \forces{n}{\varphi(\alphaseq)} \lor \forces{n}{\psi(\alphaseq)} \\
  \forces{n}{(\varphi\imp\psi)(\alphaseq)}
    & \iff  \forall k\leq n.\, \forces{k}{\varphi(\restrict{\alphaseq}{k})} \imp \forces{k}{\psi(\restrict{\alphaseq}{k})} \\
  \forces{n}{(\exists x{:}X.\varphi)(\alphaseq)} 
    & \iff \exists \alpha\,{\in}\,\den{X}(n).\, \forces{n}{\varphi(\alphaseq,\alpha)} \\
  \forces{n}{(\forall x{:}X.\varphi)(\alphaseq)} 
    & \iff \forall k\leq n, \alpha \,{\in}\,\den{X}(k). \,
    \forces{k}{\varphi(\restrict{\alphaseq}{k},\alpha)} 
\end{align*}

\begin{prop}
  \label{prop:laterPred}
  $\laterPred$ is the unique morphism on $\Omega$ satisfying
    $\forces{1}{\laterPred{\varphi(\alpha)}}$  and  $\forces{n\mathop+1}{\laterPred\varphi(\alpha)} \iff \forces{n}{\varphi(\restrict{\alpha}{n}})$. 
    Moreover, 
    $\forall x,y\co X. \laterPred(x\mathop=y)\, \bimp\, \next_X(x) \mathop= \next_X(y)$ 
 holds in $\stopos$. 
\end{prop}
The following definition will be useful for presenting facts about the
internal logic of $\stopos$.
\begin{defi}
  An object $X$ in $\stopos$ is \emph{total}
  if all the restriction maps $\restrmap{n}$ are surjective. 
\end{defi}
Hence all constant objects $\Delta(S)$ are total, but the total
objects also include many non-constant objects, e.g., the subobject
classifier.  The above definition is phrased in terms of the
model; the internal logic can be used to give a simple
characterization of when $X$ is total and inhabited by a global
element\footnote{$X$ is inhabited by a global element if there 
exists a morphism $x\co 1 \to X$}
: that is the case iff   $\next_X$ is internally surjective
in $\stopos$, i.e., iff $\forall y:\later X.\exists x:X.\next_X(x)=y$
holds in $\stopos$.
The following proposition can be proved using the forcing semantics;
note that the distribution rules below for $\laterPred$ generalize the
ones for constant sets described in~\cite{dreyer+:lics09} (since
constant sets are total).

\begin{thm}
\label{prop:internal-logic-rules}
  In the internal logic of $\stopos$ we have:
  \begin{enumerate}[\em(1)]
  \item (Monotonicity). $\forall p:\Omega.\, p\imp \laterPred p$.
  \item (L{\"o}b rule).
    $\forall p:\Omega.\, (\laterPred p \imp p) \imp p$.
  \item $\laterPred$ commutes with the logical connectives 
    $\top$, $\land$, $\imp$, $\lor$, 
    but does not preserve $\bot$. 
  \item For all $X$, $Y$, and $\varphi$, we have
    $\exists y:Y.\laterPred \varphi(x,y) \imp \laterPred(\exists y:Y.\,\varphi(x,y))$.
    The implication in the opposite direction holds if $Y$ is total and inhabited.
  \item For all $X$, $Y$, and $\varphi$, we have
    $\laterPred(\forall y:Y.\,\varphi(x,y)) \imp \forall y:Y.\laterPred\varphi(x,y)$.
    The implication in the opposite direction holds if $Y$ is total.
  \end{enumerate}
\end{thm}

\noindent We now define an internal notion of contractiveness in the logic of
$\stopos$ which implies (in the logic) the existence of a unique fixed
point for inhabited types.

\begin{defi}\label{defn:internally-contractive}
  The predicate $\contractivePred$ on $Y^X$ is defined in the internal logic
  by
  \begin{displaymath}
    \contractivePred(f) \biimpdefn \forall x,x': X. \laterPred(x=x') \imp f(x) = f(x').
  \end{displaymath}
\end{defi}

For a morphism $f: X\to Y$, corresponding to a global element of
$Y^X$,  we have that if $f$ is contractive (in the external sense of
Definition~\ref{def:contractive}), then $\contractivePred(f)$ holds in 
the logic of $\stopos$.  
The converse is true if $X$ is total and inhabited, but not in general.  
We use both notions of contractiveness: the external notion provides
for a simple algebraic theory of fixed points for not only
morphisms but also functors (see Section~\ref{sec:recdom}), whereas
the internal notion is useful when working in the internal logic.

The internal notion of contractiveness generalizes the usual metric notion of
contractiveness for functions between complete bounded ultrametric
spaces; see Section~\ref{sec:cbult}.

\begin{thm}[Internal Banach Fixed-Point Theorem]
  \label{thm:internal-banach}
  The following holds in $\stopos$:
  \begin{displaymath}
    (\exists x:X.\top) \land \contractivePred(f) \imp \existsunique x: X.\, f(x) = x.
  \end{displaymath}
\end{thm}

The above theorem (the Internal Banach Fixed-Point Theorem) is proved
in the internal logic using the following lemma, which expresses a
non-classical property. The lemma can be proved in the internal logic
using the L{\"o}b rule (and using that $N$ is a total object) ---
below we give a semantic proof using the Kripke-Joyal semantics.
\begin{lem}
  The following holds in $\stopos$:
  \begin{displaymath}
    \contractivePred(f) \imp \exists n: N. \forall x,x': X.\, f^n(x) = f^n(x').
  \end{displaymath}
\end{lem}

\begin{proof}
We must show that any $m$ forces the predicate. Unfolding the definition of the forcing relation, we see that it suffices to show that for all
$m$ and all $f \in X^X(m)$ there exists an $n$ such that 
\[
\forces{m}{\contractivePred(f)} \imp \forces{m}{\forall x,x': X.\, f^n(x) = f^n(x')}
\]
The element $f$ is a family $(f_i \co X(i) \to X(i))_{i \leq m}$ and the condition $\forces{m}{\contractivePred(f)}$ implies that $f_i^i(x) = f_i^i(y)$ for all $i\leq m$ and all $x,y \in X(i)$. In particular $f_i^m$ is constant. Therefore choosing $n=m$ makes $\forces{m}{\forall x,x': X.\, f^n(x) = f^n(x')}$ true.
\end{proof}

\subsection{Recursive relations} 
As an example application of Theorem~\ref{thm:internal-banach}, we consider the 
definition of recursive predicates.
Let $\varphi(r):\Omega^X$ be a predicate on $X$ in the internal logic of $\stopos$ as
presented above (over non-dependent types, but possibly using
$\laterPred$) with free variable $r$, also of type $\Omega^X$.
Note that $\Omega^X$ is inhabited by a global
element.  If $r$ only occurs under a $\laterPred$ in $\varphi$, then
$\varphi$ defines an internally contractive map
$\varphi\co\Omega^X\to\Omega^X$ (proved by external induction on $\varphi$).
Therefore, by Theorem~\ref{thm:internal-banach}, $\existsunique\,
r\co\Omega^X. \varphi(r) = r$ holds in $\stopos$.  By description (aka
axiom of unique choice), which holds in any
topos~\cite{Lambek:Scott:86}, there is then a morphism $R: 1 \to
\Omega^X$ such that $\varphi(R)=R$ in $\stopos$, and since internal and
external equality coincides, also $\varphi(R) = R$ externally as
morphisms $1\to \Omega^X$. 
In summa, we have shown the well-definedness of recursive predicates
$r = \varphi(r)$ where $r$ only occurs guarded by  $\laterPred$ in $\varphi$.

Note that we have \emph{proved} the existence of recursive
guarded relations (and thus do not have to add them to the language with
special syntax) since we are working with a higher-order logic.

\begin{exa} \label{ex:rewrite}
  Suppose $R\subseteq X \times X$ is some relation on a set $X$. We
  can include it into $\stopos$ by using the functor $\Delta \co
  \Sets \to \SItopos$, obtaining $\Delta R \subseteq \Delta X \times
  \Delta X$. Consider the recursive relation
\[
\transClosureAlt{R} (x,y) \biimpdefn (x = y) \vee \exists z \ld (\Delta R(x,z) \meet \laterPred\transClosureAlt{R}(z,y)) \, .
\]
Now, $\forces{n+1}{\transClosureAlt{R} (x,y)}$ iff $(x,y) \in \cup_{0 \leq i \leq n} R^i$ or there exists $z$ such that $R^{n+1}(x,z)$.
If $R$ is a rewrite relation then $\forces{n+1}{\transClosureAlt{R} (x,y)}$ states the extent to which we can determine if $x$ rewrites to~$y$ by inspecting all rewrite sequences of length at most~$n$. 
\end{exa}

A variant of Example~\ref{ex:rewrite} is used in Section~\ref{sec:application}. 

\newcommand{\strengthof}[3]{#1_{#2, #3}}
\newcommand{\bang}{!}
\newcommand{\comp}{\textit{comp}}

\subsection{Recursive domain equations}
\label{sec:recdom}

In this section we present a simplified version of our results on solutions to recursive domain equations in $\SItopos$ sufficient for the example of Section~\ref{sec:application}. The full results on recursive domain equations can be found in Section~\ref{sec:dtt}.

Denote by $\nameof{f} \co 1 \to Y^X$ the curried version of $f \co X
\to Y$. Following Kock \cite{Kock:72} we say that an endofunctor $F\co
\SItopos \to \SItopos$ is \emph{strong} if, for all $X,Y$, there
exists a morphism $\strengthof{F}{X}{Y} \co Y^X \to FY^{FX}$ such that
$\strengthof{F}{X}{Y} \circ \nameof{f} = \nameof{Ff}$ for all $f$.

\begin{defi}
\label{def:locally:contractive}
A strong endofunctor on $\SItopos$ is \emph{locally contractive}
if each $\strengthof{F}{X}{Y}$ is contractive, i.e., there exists a family $\strengthof{G}{X}{Y}$ such that $\strengthof{G}{X}{Y} \circ \next_{X^Y} = \strengthof{F}{X}{Y}$  and moreover $G$ respects identity and composition, that is the following diagrams commute
\begin{diagram}
\later (Y^X) \times \later (Z^Y) & \rTo^{\iso} & \later (Y^X \times Z^Y) & \rTo^{\later (\comp)} & \later (Z^X)  &\GAP & 1 & \rTo^{\later \nameof{\id}} & \later (X^X)\\
\dTo^{\strengthof{G}{X}{Y} \times \strengthof{G}{Y}{Z}} &&&& \dTo_{\strengthof{G}{X}{Z}} &&& \rdTo_{\nameof{\id}} & \dTo_{\strengthof{G}{X}{X}} \\
FY^{FX} \times FZ^{FY} & & \rTo^{\comp} && FZ^{FX} &&&& X^X
\end{diagram}
\end{defi}
This notion readily generalizes to
mixed-variance endofunctors on $\SItopos$.
\begin{rem}
  Definition~\ref{def:locally:contractive} is slightly less general
  than the one given in the conference version of this paper~\cite{BirkedalL:sgdt}
  where local contractiveness simply required $\strengthof{F}{X}{Y}$
  to be contractive. The definition given here greatly simplifies the
  proof of existence of solutions to recursive domain equations,
  especially in the general case as presented in
  Section~\ref{sec:class:of:models}, and at the same time, the extra
  requirements used here do not rule out any examples we know of. In
  particular, the syntactic conditions for well-definedness of
  recursive types remain unchanged.
\end{rem}

The requirement of $G$ commuting with composition and identity can be
rephrased as $G$ defining an enriched functor. In
Section~\ref{sec:axioms} we use this observation to generalise the notion of
locally contractive functor. 

For example, $\later$ is locally contractive (as witnessed by $J$ (\ref{eq:J})), and one can show that
the composition of a strong functor and a locally contractive functor
(in either order) is locally contractive (see Lemma~\ref{lem:loc:contractive:comp} for a
generalized statement). As a result, one can show
that any type expression $A(X, Y)$ constructed from type variables
$X,Y$ using $\later$ and simple type constructors in which $X$ occurs
only negatively and $Y$ only positively and both only under $\later$
gives rise to a locally contractive functor. Indeed, in Section~\ref{sec:dtt}
we present such syntactic conditions ensuring that a type expression in 
dependent type theory induces a locally contractive functor.

\begin{thm} \label{thm:rec:dom:eq}
 Let $F \co \opcat{\SItopos} \times \SItopos \to \SItopos$ be a locally contractive functor. Then there exists a unique $X$ (up to isomorphism) such that $F(X,X) \iso X$. 
\end{thm}

Section~\ref{sec:class:of:models} gives a detailed proof of a generalised version of this theorem. Here we just sketch a proof. We consider first the covariant case.

\begin{lem}  \label{lem:lcf:n:iso}
Let $F \co \SItopos \to \SItopos$ be locally contractive and say that $f\co X \to Y$ is an $n$-isomorphism if $f_{i}$ is an isomorphism for all $i\leq n$. Then $F$ maps $n$-isomorphisms to $n+1$-isomorphisms for all $n$. 
\end{lem}

Since any morphism $f\co X \to Y$ is a $0$-isomorphism $F^n(f)\co F^nX \to F^nY$ is an $n$-isomorphism. Consider the sequence
\begin{diagram}[LaTeXeqno] \label{eq:fixed:point:sequence}
F1 & \lTo^{F\bang} & F^{2}1 & \lTo^{F^2(\bang)}  & F^{3}1 & \lTo^{F^3(\bang)} & F^{4}1 & \dots
\end{diagram}
The sequence above is a sequence of morphisms and objects in $\SItopos$ and 
so represents a diagram of sets and functions as in
\begin{diagram}[size=2em,LaTeXeqno] \label{eq:fixed:point:square}
F(1)(1) & \lTo^{F(\bang)_1} & F^{2}(1)(1) & \lTo^{F^2(\bang)_1}  & F^{3}(1)(1) & \lTo^{F^3(\bang)_1} & F^{4}(1)(1) & \lTo & \dots \\
\uTo && \uTo && \uTo && \uTo \\
F(1)(2) & \lTo^{F(\bang)_2} & F^{2}(1)(2) & \lTo^{F^2(\bang)_2}  & F^{3}(1)(2) & \lTo^{F^3(\bang)_2} & F^{4}(1)(2) & \lTo &  \dots \\
\uTo && \uTo && \uTo && \uTo \\
F(1)(3) & \lTo^{F(\bang)_3} & F^{2}(1)(3) & \lTo^{F^2(\bang)_3}  & F^{3}(1)(3) & \lTo^{F^3(\bang)_3} & F^{4}(1)(3) & \lTo & \dots \\
\vdots && \vdots && \vdots && \vdots
\end{diagram}
By the above observation, $F^n(\bang)_k$ is an isomorphism for $k \leq
n$, in other words, after $k$ iterations of $F$ the first $k$
components are fixed by further iterations of $F$. Intuitively, we can
therefore form a fixed point for $F$ by taking the diagonal of
(\ref{eq:fixed:point:square}), i.e, the object whose $k$'th component
is $F^k(1)(k)$. Indeed, in Section~\ref{sec:class:of:models} we
construct this object as the limit of (\ref{eq:fixed:point:sequence}).

Any fixed point for such an $F$ must be at the same time an initial algebra and a final coalgebra: given any fixed point $f \co FX \iso X$ and algebra $g \co FY \to Y$ a morphism $h\co X \to Y$ is a homomorphism iff $\nameof{h}$ is a fixed point of $\xi = \lam{k}{X \to Y}{g \circ Fk \circ \inv{f}}$. Since $F$ is locally contractive, $\xi$ is contractive and so must have a unique fixed point. The case of final coalgebras is similar.

Thus, $\SItopos$ is algebraically compact in the sense of Freyd
\cite{Freyd:90,Freyd:90a,Freyd:91} with respect to locally contractive functors. The
solutions to general recursive domain equations can then be
established using Freyd's constructions.

\begin{exa}\label{ex:streams:ii}
  Recall the type $\Stream$ of streams defined concretely in the model
  in Example~\ref{ex:streams:i}.  It can be defined in the internal
  language using Theorem~\ref{thm:rec:dom:eq}, namely as the type satisfying
  the recursive domain equation
  \begin{displaymath}
    \Stream \iso N \times \later \Stream.
  \end{displaymath}
  Write $i: N\times\later\Stream\to \Stream$ for the isomorphism. (Observe that
  $i_m$ is nothing but the identity function.)
  
  Now, we can define the successor function in the internal language as the fixed point of
  the following contractive function $F: (\Stream\to\Stream)\to (\Stream\to\Stream)$:
    \begin{align*}
    F(f) = \lambda s. &  \mathsf{let}\; (n,t) \Leftarrow i^{-1}(s) \\
                               &  \mathsf{in} \; i(n+1, J(\next f)(t))
    \end{align*}
    Note that $F$ is clearly contractive (in the external sense) since
    the argument $f$ is only used under $\next$ in $F(f)$.
    Hence $F$ has a fixed point, which is  indeed the successor
    function from Example~\ref{ex:streams:i}, i.e., $\Succ = \fix_{\Stream\to\Stream}{F}$.
\end{exa}

\section{Application to Step-Indexing}
\label{sec:application}

As an example, we now construct a model of a programming language with
higher-order store and recursive types entirely inside the internal
logic of $\stopos$.
There are two points we wish to make here.  First, although the
programming language is quite expressive, the internal model
looks---almost---like a naive, set-theoretic model.  The exception is
that guarded recursion is used in a few, select places, such as
defining the meaning of recursive types, where the naive approach
would fail.  Second, when viewed externally, we recover a standard,
step-indexed model.  This example therefore illustrates that the
topos of trees gives rise to simple, synthetic accounts of
step-indexed models.

All definitions and results in Sections~\ref{subsec:application:language}~to~\ref{subsec:application:interpretation}
are in the internal logic of $\stopos$.  In Section~\ref{subsec:application:external} we
investigate what these results mean externally.

\subsection{Language}
\label{subsec:application:language}

The types and terms of \Fmuref\ are as follows:
\begin{align*}
 \tp & ::= 1 \mid \tp_1
    \mathop\times \tp_2  
    \mid \mu \alpha. \tp
    \mid \forall \alpha. \tp \mid \alpha \mid \arrowtp{\tp_1}{\tp_2} \mid \reftp{\tp}
\\[2mm]
 \tm & ::=  x \mid l \mid \unitcst \mid
    \pairtm{\tm_1}{\tm_2} 
    \mid \fst\tm \mid \snd\tm
 \mid \fold\tm   \mid \unfold\tm \mid
    \\&{}\qquad
    \tplam \alpha \tm \mid \tpapp \tm \tp \mid
    \lamtm x \tm \mid \tm_1 \ap \tm_2
    \mid \reftm\tm \mid \lookup\tm \mid \assign{\tm_1}{\tm_2}
\end{align*}
  (The full term language also includes sum types, and can be found in
\iflongversion
Appendix~\ref{app:step-indexing}.)
\fi
  Here $l$ ranges over location constants, which are encoded as
  natural numbers.  

  More explicitly, the sets $\mathrm{OType}$ and
  $\mathrm{OTerm}$ of possibly open types and terms are defined by
  induction according to the grammars above (using that toposes model
  $W$-types~\cite{palmgren-moerdijk}), and then by quotienting
  with respect to  $\alpha$-equivalence.

The set $\mathrm{OValue}$ of syntactic values is an inductively defined
subset of $\mathrm{OTerm}$:
  \begin{align*}
    v & ::= x \mid l \mid \unitcst \mid \pairtm{v_1}{v_2} \mid 
    \fold{v} \mid \tplam{\alpha}{\tm} \mid \lamtm{x}{t}
  \end{align*}

Let $\Term$ and $\Value$ be the subsets of closed terms and closed
values, respectively.  Let $\Store$ be the set of finite maps from
natural numbers to closed values; this is encoded as the set of those
finite lists of pairs of natural numbers and closed values that contain
no number twice.  Finally, let $\mathrm{Config} = \Term \times \Store$.

  The typing judgements have the form\/
  $\tpjudge{\Xi}{\Gamma}{\tm}{\tp}$ where $\Xi$ is a context of type
  variables and $\Gamma$ is a context of term variables.  The typing
  rules are standard and can be found in
\iflongversion
 Appendix~\ref{app:step-indexing}.
\fi
Notice, however, that there is no
  context of location variables and no typing judgement for stores: we
  only need to type-check terms that can occur in programs.

\subsection{Operational semantics}
\label{subsec:application:operational_semantics}

We assume a standard one-step relation $\step \co \Powset{\Config \times
\Config}$ on configurations by induction, following the usual
presentation of such relations by means of inference rules (see, e.g., the
online appendix to~\cite{dreyer+:popl10}).  For
simplicity, allocation is deterministic: when allocating a new
reference cell, we choose the smallest location not already in the
store.  Notice that the $\step$ relation is defined on untyped
configurations.  Erroneous configurations are ``stuck.''

So far, we have defined the language and operational semantics exactly
as we would in standard set theory.  Next comes the crucial
difference.  We use Theorem~\ref{thm:internal-banach} to define the predicate 
$\eval \co \Powset{\Term \times \Store \times \Powset{\Value \times \Store}}$, 
\[
\begin{array}{l}
\eval(t,s,Q)\\
\quad \biimpdefn \begin{arraytl}(t \in \Value
~\land~ Q(t,s)) ~\lor\\
(\exists t_1 \co \Term, s_1 \co \Store.\,\\
\quad \mathrm{step}((t,s),(t_1,s_1)) ~\land~ \laterPred \mathrm{eval}(t_1,s_1,Q))\end{arraytl}
\end{array}
\]
Intuitively, the predicate $Q$ is a post-condition, and $\eval(t,s,Q)$
is a \emph{partial} correctness specification, in the sense of Hoare
logic, meaning the following: (1) The configuration $(t,s)$ is safe, i.e., 
it does not lead to an error.  (2) \emph{If} the configuration $(t,s)$
evaluates to some pair $(v,s')$, \emph{then} at that point in time
$(v,s')$ satisfies $Q$.  We shall justify this intuition in
Section~\ref{subsec:application:external} below.  The use of $\laterPred$ ensures that the
predicate is well-defined; in effect, we postulate that one evaluation
step in the programming language actually takes one unit of time in
the sense of the internal logic.  As we shall see below, this
``temporal'' semantics is essential in the proof of the fundamental
theorem of logical relations.

Notice how guarded recursion is used to give a simple,
coinduction-style definition of partial correctness.  The L\"ob rule
can then be conveniently used for reasoning about this definition.
For example, the rule gives a very easy proof that if $(t,s)$ is a
configuration that reduces to itself in the sense that
$\step((t,s),(t,s))$ holds, then $\eval(t,s,Q)$ holds for any $Q$.
The L\"ob rule also proves the following results, which are used
to show the fundamental theorem below.

\begin{prop}
\label{prop:rule-of-consequence}
  Let $Q, Q' \in \Powset{\Value \times \Store}$ such that $Q
  \subseteq Q'$.  Then for all $t$ and $s$ we have that $\mathrm{eval}(t,s,
  Q)$ implies $\mathrm{eval}(t,s, Q')$.
\end{prop}

\begin{prop}
\label{prop:eval-move-ectx}
For all stores $s$, all terms $t$, all evaluation contexts $E$ such that
$E[t]$ is closed, and all predicates $Q \in \Powset{\Value \times \Store}$, we have
that $\mathrm{eval}(E[t],s, Q)$ holds iff
$\mathrm{eval}(t,s, \, \lambda (v_1,s_1). \, \mathrm{eval}(E[v_1],s_1, Q))$ holds.
\end{prop}

\subsection{Definition of Kripke worlds}
\label{subsec:application:worlds}

The main idea behind our interpretation of types is
as in~\cite{BirkedalL:essence,BirkedalST:10}:  Since \Fmuref\
includes reference types, we use a Kripke model of types, where
a semantic type is defined to be a world-indexed
family of sets of syntactic values. A world is a map from locations to
semantic types.  This introduces a circularity between
semantic types $\mathcal{T}$ and worlds~$\mathcal{W}$, which can be expressed 
as a pair of domain equations: $\W = N \tofin \STalt$ and $\STalt = \W \tomon \Powset{\Value}$. 

Rather than solving the above stated domain equations exactly,
we solve a guarded variant. 
More precisely, we define the set
\[
\label{eqn:worldfunctor}
\ST = \mu X.\, \later((N \tofin X)  \tomon \Powset{\Value}) \ .
\]
Here $N \tofin X$ is the set
$\dsum{A}{\finpowset{N}}{X^A}$ where $\finpowset{N} = \{A \subset N
\mid \exists m \forall n \in A \ld\, n<m \}$. The set $\dsum{A}{\finpowset{N}}{X^A}$
is ordered by graph inclusion
and $\tomon$ is the set of monotonic functions realized as a subset
type on the function space. 

The type $\ST$ can be seen to be well-defined as a consequence of 
the theory of Section~\ref{sec:dtt}, in particular Proposition~\ref{prop:functorial:types}.
Alternatively,
observe that the corresponding functor is of the form $F = \later
\mathop\circ G$.  Here $G$ is strong because its action on
morphisms can be defined as a term $Y^X \to GY^{GX}$ in the internal
logic.  Now, since $\later$ is locally contractive so is $F$.  Hence
by Theorem~\ref{thm:rec:dom:eq}, $F$ has a unique fixed point
$\ST$, with an isomorphism $i\co \ST\to F(\ST)$. 
We define
\begin{align*}
  \W &= N \tofin \ST\ ,
  &
   \STalt &= \W \tomon \Powset{\Value}\ ,
\end{align*}
and $  \STalt^{\mathrm{c}} = \W \to \Powset{\Term}$. 
Notice that $\ST$ is isomorphic to $\later\STalt$.
We now define $\App\co
\ST \to \STalt$ and $\Lam \co \STalt \to \ST$ as follows.  
First, $\App$ is the isomorphism $i$ composed with the operator $d \co \later\STalt
\to \STalt$ given by
\begin{displaymath}
  d (f) = \lambda w.\lambda v. \succ( J(J(f)(\next \,w)) (\next \,v)),
\end{displaymath}
where $J$ is the map in~(\ref{eq:J})
and $\succ\co \later \Omega \to \Omega$ is as defined on page~\pageref{page:succ}.
(This is a general way of
lifting algebras for $\later$ to function spaces.)
Here one needs to check that $d$ is well-defined, 
i.e., preserves monotonicity.
Second,
$\Lam\co  \STalt \to \ST$ is defined by $\Lam = i^{-1} \circ \next_\STalt$.

Define $\laterPred : \STalt \to \STalt$ as the
pointwise extension of $\laterPred : \Omega \to \Omega$, i.e., for $\nu \in
\STalt$, $w \in \W$ and $v \in \Value$, we have that $(\laterPred \nu)(w)(v)$ holds iff
$\laterPred(\nu(w)(v))$ holds.  
\begin{lem}
\label{lem:app-lam}
  $\mathrm{app} \circ \mathrm{lam} = \laterPred : \mathcal{T} \to \mathcal{T}$.
\end{lem}

\subsection{Interpretation of types}
 \label{subsec:application:interpretation}

Let $\mathrm{TVar}$ be the set of type variables, and for $\tp \in
\mathrm{OType}$, let $\mathrm{TEnv}(\tp) = \{\,\varphi \in \mathrm{TVar}
  \tofin \STalt \mid \fv{\tp} \subseteq \dom{\varphi}\,\}$.  The
  interpretation of programming-language types is defined by
  induction, as a function
\[
\den{\cdot} : \prod_{\tp \in \mathrm{OType}} \mathrm{TEnv(\tp)} \to
\STalt \, .
\]
We show some cases of the definition here; the complete definition can
be found in
\iflongversion
Appendix~\ref{app:interpretation-types}.
\fi
  \begin{align*}
    \mngTp{\alpha}{\Xi}{\varphi} &= \varphi(\alpha)\\
    \mngTp{\tp_1 \times \tp_2}{\Xi}{\varphi} &= \lambda w.\,
\{ \pairtm{v_1}{v_2} \mid
v_1 \in \mngTp{\tp_1}{\Xi}{\varphi}(w) \mathrel{\land}
v_2 \in \mngTp{\tp_2}{\Xi}{\varphi}(w) \}\\
     \mngTp{\reftp \tp}{\Xi}{\varphi} &=
\lambda w.\, \{\, l \mid\begin{arraytl} l \in \dom{w} \mathrel{\land}
\forall w_1 \geq w .\, 
\App (w(l)) (w_1) =
\laterPred(\mngTp{\tp}{\Xi}{\varphi})(w_1) \, \}\end{arraytl}\\
    \mngTp{\forall \alpha . \tp}{\Xi}{\varphi} &= \lambda w.\,\{\,
    \tplam{\alpha}{t}
    \mid\begin{arraytl} \forall \nu \in \STalt.\, \forall w_1 \geq w.\, 
    t \in \compOp(\mngTp{\tp}{\Xi,\alpha}{\varphi[\alpha \mapsto
      \nu]})(w_1) \, \}\end{arraytl}\\
    \mngTp{\mu \alpha . \tp}{\Xi}{\varphi} &= \mathit{fix}\left(\lambda
      \nu. \, \lambda w. \,
    \{\, \fold{v} \mid \laterPred(v \in
    \mngTp{\tp}{\Xi,\alpha}{\varphi[\alpha \mapsto \nu]} \ap (w))
    \}\right)\\
    \mngTp{\arrowtp{\tp_1}{\tp_2}}{\Xi}{\varphi} &= 
\lambda w.\, \{\, \lamtm{x}{t} \mid
\begin{arraytl}\forall w_1 \geq w.\,
\forall v \in \mngTp{\tp_1}{\Xi}{\varphi}(w_1).\, 
t[v/x] \in \compOp(\mngTp{\tp_2}{\Xi}{\varphi})(w_1) \, \}
\end{arraytl}
\end{align*}
Here the operations $\compOp : \STalt \to \STalt^\mathrm{c}$ and
$\statesOp : \W \to \Powset{\Store}$ are given by
\begin{align*}
\compOp(\nu)(w) &= \{\, t \mid \begin{arraytl}\forall s \in \statesOp(w) .\,
\mathit{eval}(t,s,\, 
\lambda (v_1,s_1).\, 
\exists w_1 \geq w.\,\\ \quad v_1 \in \nu(w_1) ~\land~ s_1 \in \statesOp(w_1))
\,\}\end{arraytl}\\
\statesOp(w) &= \{\, s \mid
\begin{arraytl}\dom{s} = \dom{w} ~\mathrel\land~\\
\forall l \in \dom{w}.\, s(l) \in \App (w(l))(w) \, \}.\end{arraytl}
\end{align*}

Notice that this definition is almost as simple as an attempt at a
naive, set-theoretic definition, except for the two explicit uses
of~$\laterPred$.  In the definition of $\den{\mu\alpha.\tp}$, the use
of~$\laterPred$ ensures that the fixed point is well-defined according
to Theorem~\ref{thm:internal-banach}.  As for the definition of
$\den{\reftp \tp}$, the~$\laterPred$ is needed because we have
$\laterPred$ instead of the identity in
Lemma~\ref{lem:app-lam}.
In both
cases, the intuition is the usual one from step-indexing: since an
evaluation step takes a unit of time, it suffices that a certain
formula only holds later.

\begin{prop}[Fundamental theorem]
  \label{prop:fundamental}
  If\/ $\vdash t : \tp$, then for all $w \in \mathcal{W}$ we have $t \in \compOp(\mngTp{\tp}{\emptyset}{\emptyset})(w)$.
\end{prop}
\begin{proof}
To show this, one first generalizes to open types and open terms in
the standard way, and then one shows semantic counterparts of all the
typing rules of the language.
\iflongversion
See Appendix~\ref{app:step-indexing-soundness}.
\fi
To illustrate the use of $\laterPred$, we outline the case of
reference lookup: $\vdash \lookup{\tm} :  \tp$.  Here the
essential proof obligation is that $v \in
\mngTp{\reftp \tp}{\emptyset}{\emptyset}(w)$ implies $\lookup v \in
\compOp(\mngTp{\tp}{\emptyset}{\emptyset})(w)$.  To show this,
we unfold the definition of $\compOp$.  Let $s \in \statesOp(w)$ be
given; we must show
\begin{equation}
\label{eq:fundamental-theorem-lookup:maintext}
\mathit{eval}(\lookup v,s,\,\lambda (v_1,s_1).\, 
\exists w_1 \geq w.\,
v_1 \in \mngTp{\tp}{\emptyset}{\emptyset}(w_1)  \land  s_1 \in \statesOp(w_1))\,.
\end{equation}
By the assumption that $v \in \mngTp{\reftp
  \tp}{\emptyset}{\emptyset}(w)$, we know that $v = l$ for some
location $l$ such that $l \in \dom{w}$ and $\App (w(l)) (w_1) =
\laterPred(\mngTp{\tp}{\emptyset}{\emptyset})(w_1)$ for all $w_1 \geq w$.
Since $s \in \statesOp(w)$, we know that $l \in \dom{s} = \dom{w}$ and
$s(l) \in \App(w(l))(w)$.  We therefore have $\step((\lookup v, s),
(s(l), s))$.  Hence, by unfolding the definition of $\eval$ in
\eqref{eq:fundamental-theorem-lookup:maintext} and using the rules from
Proposition~\ref{prop:internal-logic-rules}, it
remains to show that
$\exists w_1 \geq w.\,
\laterPred(s(l) \in \mngTp{\tp}{\emptyset}{\emptyset}(w_1)) ~\land~ \laterPred(s \in \statesOp(w_1))$.
We choose $w_1 = w$.  First, $s \in \statesOp(w)$ and hence
$\laterPred(s \in \statesOp(w))$.  Second, $s(l) \in \App(w(l))(w) =
\laterPred(\mngTp{\tp}{\emptyset}{\emptyset})(w)$, which means exactly that
$\laterPred(s(l) \in \mngTp{\tp}{\emptyset}{\emptyset}(w))$.
\end{proof}

\subsection{The view from the outside}
\label{subsec:application:external}

We now return to the standard universe of sets and give external
interpretations of the internal results above.  One basic ingredient
is the fact that the constant-presheaf functor $\Delta : \Sets \to
\stopos$ commutes with formation of $W$-types.  This fact can be shown
by inspection of the concrete construction of $W$-types for presheaf categories
given in~\cite{palmgren-moerdijk}.

Let $\mathrm{OType'}$ and $\mathrm{OTerm'}$ be the sets of possibly
open types and terms, respectively, as defined by the grammars above.
Similarly, let $\Value'$, $\Store'$, $\Config'$, and $\step'$
be the external counterparts of the definitions
from the previous sections.

\begin{prop}
  $\mathrm{OType} \cong \Delta(\mathrm{OType'})$, and similarly for
  $\mathrm{OTerm}$, $\Value$, $\Store$, and $\Config$.
  Moreover, under these isomorphisms $\step$ corresponds to
  $\Delta\step'$ as a subobject of $\Config \times \Config$.
\end{prop}

This result essentially says that the external interpretation of the
$\step$ relation is world-independent, and has the expected meaning:
for all $n$ we have that $\forces{n}{\step((t',s'),(t',s'))}$ holds iff
$(t,s)$ actually steps to $(t',s')$ in the standard operational
semantics.  We next consider the $\eval$ predicate:

\begin{prop}
  $\forces{n}{\eval(t,s,Q)}$ iff the following property holds: for all
  $m < n$, if $(t,s)$ reduces to $(v,s')$ in $m$ steps, then
  $\forces{(n-m)}{Q(v,s')}$. 
\end{prop}

Using this property and the forcing semantics from
Section~\ref{sec:internal-logic}, one
obtains that the external meaning of the interpretation of types is a
step-indexed model in the standard sense.  
In particular, note that an element of
$\Powset{\Value}(n)$ can be viewed as a set of pairs $(m,v)$ of
natural numbers $m \leq n$ and values which is downwards closed in the
first component. 

\subsection{Discussion}

For simplicity, we have just considered a unary model in this extended
example; we believe the approach scales well both to relational models
and to more sophisticated models for reasoning about local
state~\cite{Ahmed-al:POPL2009,dreyer+:icfp2010,Schwinghammer:Birkedal:Stovring:11}.
In particular, we have experimented with an internal-logic formulation of parts
of~\cite{Schwinghammer:Birkedal:Stovring:11}, which involve recursively
defined relations on recursively defined types. 

As mentioned above, the operational semantics of this example was for
simplicity chosen to be deterministic.  We expect that one can easily
adapt the approach presented here
to non-deterministic languages.  For that, the evaluation predicate
must be changed to quantify universally (rather than existentially)
over computation steps, and errors must explicitly be ruled
out, as in: 
\[
\begin{array}{l}
\eval'(t,s,Q)\\
\quad \stackrel{\mathrm{def}}\iff \begin{arraytl}(t \in \Value
\to Q(t,s)) ~\land~ \lnot\mathrm{error}(t,s) ~\land~\\
(\forall t_1 : \Term, s_1 : \Store.\,\\
\quad \mathrm{step}((t,s),(t_1,s_1)) \to \laterPred \eval'(t_1,s_1,Q)).\end{arraytl}
\end{array}
\]

As mentioned in the Introduction, in~\cite{BirkedalL:essence} the recursive equation for $\STalt$
was solved in the category $\CBUlt$ of
ultrametric spaces. Using the space $\STalt$ the model was then defined in the
usual universe of sets in the standard, explicit
step-indexed style.
Here instead we observe that the relevant part of $\CBUlt$ is a full
subcategory of $\stopos$ (Section~\ref{sec:cbult}), solve the
recursive equation in $\stopos$, and then \emph{stay} within $\stopos$
to give a simpler model that does not refer to step indices.
In particular, the proof of the fundamental theorem is much simpler
when done in $\stopos$.

\section{Dependent Types}
\label{sec:dtt}

\newcommand{\HODTTj}[4]{{#1} \vdash {#3} : {#4}} 
\newcommand{\HODTTej}[5]{{#1} \mid {#2} \vdash {#3} = {#4} : {#5}} 
\newcommand{\HODTTtj}[3]{{#1} \vdash {#3} : \Type} 
\newcommand{\HODTTcj}[1]{{#1} : \Ctx} 
\newcommand{\rect}[2]{\mu{#1}\ld {#2}} 
\newcommand{\cardinality}[1]{\mid{#1}\mid}
\newcommand{\st}{\mathrm{st}}
\newcommand{\functorialcontractiveness}{functorial contractiveness}
\newcommand{\Functorialcontractiveness}{Functorial contractiveness}

\newcommand{\contractivelyfunctorial}{contractively functorial}

\newcommand{\nextSlice}[1]{\next^{#1}}
\newcommand{\symmetric}[1]{\tilde{#1}}
\newcommand{\widesymmetric}[1]{\widetilde{#1}}

\newcommand{\ecat}{\catfont{E}}
\newcommand{\hodtt}{HoDTT}
\newcommand{\elements}[1]{\int #1}
\newcommand{\Psh}[1]{\BigPsh{#1}}
\newcommand{\BigPsh}[1]{\widehat{#1}}

Since $\stopos$ is a topos it models not only higher-order logic over
simple type theory, but also over dependent type theory.  The aim of
this section is to provide the semantic foundation for extending the
dependent type theory with type constructors corresponding to $\later$
and guarded recursive types, although we postpone a detailed syntactic
formulation of such a type theory to a later paper.

Recall that dependent types in context are interpreted in slice
categories,\footnote{For now we follow the practise of ignoring coherence
  issues related to the interpretation of substitution in codomain
  fibrations since there are various ways to avoid these issues,
  e.g.~\cite{Hofmann:94}. See the end of the section for more on this issue.} 
in particular a type $\Gamma \ts A$ is
interpreted as an object of $\SIslice{\den{\Gamma}}$. To extend the
interpretation of dependent type theory with a type constructor
corresponding to $\later$, we must therefore extend the definition of
$\later$ to slice categories.

\paragraph{Slice categories concretely}

Before defining $\laterSlice{I} \co \slice{\stopos}{I} \to
\slice{\stopos}{I}$ we give a concrete description of the slice categories $\SIslice{I}$.

We first recall the construction of the category of elements for presheaves
over partial orders. For $B$ a partial order, we write $\BigPsh{B}$
for the category of presheaves over $B$, i.e., category of functors and natural 
transformations from $\opcat{B}$ to $\Sets$. 

\begin{defi} \label{def:elements} Let $B$ be a partially ordered
  set and let $X$ be a presheaf over $B$. Define the partially ordered set of
  \emph{elements} of $X$ as $\elements{X} = \{ (b,x) \mid b\in B \meet
  x \in X(b) \}$ with order defined as $(b,x) \leq (c,y)$ iff $b \leq
  c$ and $\restrict{y}{b} = x$.
\end{defi}

Note that if one applies this construction to an object $X$ of
$\SItopos$ one gets a forest $\elements X$: the roots are the elements
of $X(1)$ the children of the roots are the elements of $X(2)$ and so
on. Indeed any forest is of the form $\elements X$ for some $X$ in $\SItopos$.

\begin{prop} \label{prop:slices:concretely} Let $B$ be a partially ordered
  set and let $I$ be a presheaf over $B$. Then $\slice{\Psh{B}}{I} \simeq \Psh{\elements{I}}$.
\end{prop}

\begin{proof}
  This is a standard theorem of sheaf
  theory~\cite[Ex.~III.8]{MacLane:Moerdijk:92}, and we just recall one
  direction of the equivalence. An object $\p X \co X \to I$ of the
  slice category $\slice{\Psh{B}}{I}$ corresponds to the presheaf that maps
  $(b, i) \in \elements{I}$ to $\inv{(\p X)_b}(i)$.
\end{proof}

Thus we conclude that the slices of $\SItopos$ are of the form
presheaves over a forest.

\subsection{Generalising $\later$ to slices} \label{sec:dtt:later:slice}
There is a simple generalisation of the $\later$ functor from
$\SItopos$ to presheaves over any forest $\elements I$: if $X$ is a presheaf
over $\elements I$ then
\[
\laterSlice{I} X(n,i) = \left\{ 
\begin{array}{ll}
 1 & \text{if $n=1$}  \\
 X(n-1, \restrict{i}{n-1}) & \text{if $n>1$}
\end{array}
\right.
\]
In Section~\ref{sec:class:of:models} we shall see how to generalise this even further. 

The map $\next_X: X\to\laterSlice{I}{X}$ is represented by the 
following natural transformation in $\BigPsh{\elements{I}}$:
\begin{align*}
\next_{(1,i)}(x) & \,= * \\
\next_{(n+1,i)}(x) & \, = \restrict{x}{(n, \restrict{i}{n})}
\end{align*}

The fixed point combinator also generalizes to slices. Indeed, if $f : X \to X$
in $\BigPsh{\elements{I}}$ is contractive, in the sense that there exists a $g: \laterSlice{I} {X}\to X$ such that $f= g \circ \next$,
then we can construct a fixed point of $f$ (i.e., a natural transformation $1\to X$)
by:
\begin{displaymath}
  \begin{array}{rcl}
  x_{(1,i)} & = & g_{(1,i)}(*) \\    
  x_{(n+1,i)} & = & g_{(n+1,i)} (x_{(n,\restrict{i}{n})}).
  \end{array}
\end{displaymath}
This construction generalises to a fixed point combinator $\fix_X \co (\laterSlice{I} X \to X) \to X$ satisfying the properties of the global fixed point operator described in Theorem~\ref{thm:fp:op}. 

\begin{prop} \label{prop:later:pb}
 Let $\p{Y} \co Y \to I$ be an object of $\SIslice{I}$. There is a map $\laterSlice{I} Y \to \later Y$  making the diagram below a pullback.
\begin{diagram}[size=2em]
\laterSlice{I} Y \SEpbk & \rTo & \later Y \\
\dTo^{\p{\laterSlice{I} Y}} && \dTo_{\later \p{Y}} \\
I & \rTo^{\next} & \later I
\end{diagram}
\end{prop}
One could have also taken the pullback diagram of Proposition~\ref{prop:later:pb} as a definition of $\laterSlice{I}$, and indeed we do so in our axiomatic treatment of models of guarded recursion in Section~\ref{sec:axioms}.

The definition above allows us to consider $\later$ as a type
constructor on dependent types, interpreting $\den{\Gamma \ts \later
  A} = \laterSlice{\den{\Gamma}}(\den{\Gamma \ts A})$. 
The following proposition expresses that this interpretation of
$\later$ behaves well wrt.\ substitution. 

\begin{prop} 
  For every $u \co J \to
  I$ in $\SItopos$ there is a natural isomorphism $u^* \circ
  \laterSlice{I} \iso \laterSlice{J} \circ u^*$.  As a consequence,
  the collection of functors $(\laterSlice{I})_{I \in \stopos}$ define
  a fibred endofunctor on the codomain fibration. Moreover, $\next$
  defines a fibred natural transformation from the fibred identity on
  the codomain fibration to $\later$.
\end{prop}

We remark that each $\laterSlice{I}$ has a left adjoint, but in Section~\ref{sec:left:adjoint} we prove that this
family of left adjoints does not commute with reindexing. As a consequence, it does not define
a well-behaved dependent type constructor. 

\subsection{Recursive dependent types} \label{sec:rec:dep:types}
Since the slices of $\stopos$ are cartesian closed, the notions of strong
functors and locally contractive functors from
Definition~\ref{def:locally:contractive} also make sense in slices. 
Thus we can formulate a version of Theorem~\ref{thm:rec:dom:eq} generalised to all slices of $\SItopos$. 
The next theorem does that, and further generalises to parametrized domain equations, 
a step necessary for modelling nested recursive types. 

For the statement of the theorem recall the symmetrization $\symmetric{F} \co
(\opcat\ccat \times \ccat)^{n} \to \opcat\ccat \times \ccat$ of a
functor $F \co (\opcat\ccat \times \ccat)^{n} \to \ccat$ defined as
$\symmetric{F}(\vec X, \vec Y) = \pair{F(\vec Y, \vec X)}{F(\vec X,
  \vec Y)}$.

\begin{thm} \label{exist:unique:rec:dom:eq} Let $F\co
  (\opcat{(\SIslice{I})} \times \SIslice{I})^{n+1} \to \SIslice{I}$ be strong and 
  locally contractive in the $(n+1)th$ variable pair. Then there
  exists a unique (up to isomorphism) \[\Fix F \co
  (\opcat{(\SIslice{I})} \times \SIslice{I})^{n} \to \SIslice{I}\] such
  that $F \circ \pair{\id}{\symmetric{\Fix F}} \iso \Fix F$.
Moreover, if $F$ is locally contractive in all variables, so is $\Fix F$. 
\end{thm}

We postpone the proof of this theorem to Section~\ref{sec:class:of:models},
where we prove the existence of solutions to recursive domain equations
for a wider class of categories and functors. 

One can prove that the fixed points obtained by
Theorem~\ref{exist:unique:rec:dom:eq} are initial dialgebras in the
sense of Freyd~\cite{Freyd:90,Freyd:90a,Freyd:91}. This universal property generalises initial algebras
and final coalgebras to mixed-variance functors, and can be used to
prove mixed induction / coinduction principles~\cite{Pitts:96}.

The formation of recursive types is well-behaved wrt. substitution:

\begin{prop} \label{prop:rec:types:subst}
If 
\begin{diagram}[size=2em]
 (\opcat{(\SIslice{I})} \times \SIslice{I})^{n+1} & \rTo^{F} & \SIslice{I} \\
 \dTo^{u^{*}} && \dTo_{u^{*}}\\
 (\opcat{(\SIslice{J})} \times \SIslice{J})^{n+1} & \rTo^{G} & \SIslice{J} 
\end{diagram}
commutes up to isomorphism, so does
\begin{diagram}[size=2em]
 (\opcat{(\SIslice{I})} \times \SIslice{I})^{n} & \rTo^{\Fix F} & \SIslice{I} \\
 \dTo^{u^{*}} && \dTo_{u^{*}}\\
 (\opcat{(\SIslice{J})} \times \SIslice{J})^{n} & \rTo^{\Fix G} & \SIslice{J} 
\end{diagram}
\end{prop}

For the moment, our proof of Proposition~\ref{prop:rec:types:subst} is conditional on the existence of unique fixed points, i.e., we prove that if $\Fix F$ and $\Fix G$ exist, then they make the required diagram commute up to isomorphism.

\begin{proof}
  Note that $\Fix G \circ u^{*}$ is the unique $H$ up to isomorphism such
  that
  \[G(u^{*}(\vec X, \vec Y), \widesymmetric{H}(\vec X, \vec
  Y))) \iso H(\vec X, \vec Y). \]
 Now,
  \begin{align*}
    G(u^{*}(\vec X, \vec Y), (\widesymmetric{(u^{*}\circ \Fix F)}(\vec X, \vec Y))) & 
          \iso G(u^{*}(\vec X, \vec Y), u^{*}(\widesymmetric{\Fix F}(\vec X, \vec Y))) \\
    & 
          \iso u^{*}(F(\vec X, \vec Y, \widesymmetric{\Fix F}(\vec X, \vec Y))) \\
    & \iso u^{*} \widesymmetric{\Fix F}(\vec X, \vec Y))
  \end{align*}
  and so we conclude $u^{*} \Fix F(\vec X, \vec Y) \iso \Fix G
  (u^{*}(\vec X, \vec Y))$
\end{proof}

\paragraph{A higher order dependent type theory with guarded recursion}
In this section we sketch a type theory for guarded recursive types in
combination with dependent types and explain how it can be interpreted soundly in
$\SItopos$.  Since the type theory is an extension of standard higher-order 
dependent type theory, which can be interpreted in any
topos, we focus on the extension to guarded recursion, and refer to~\cite{Jacobs:99} for
details on dependent higher-order type theory and its interpretation
in a topos.   This section is meant to illustrate how the semantic results above can be 
understood type theoretically; we leave a full investigation of
the syntactic aspects of the type theory to future work.

Recursive types are naturally formulated using type variables, and thus
we allow types to contain type variables. Hence our type judgements live 
in contexts $\Gamma$ that can be formed using the rules below
\[
\begin{prooftree}
\justifies
\HODTTcj{()}
\end{prooftree}
\GAP
\begin{prooftree}
\HODTTtj{\Gamma}{-}{\tau}
\justifies
\HODTTcj{(\Gamma, x \co \tau)}
\end{prooftree}
\GAP
\begin{prooftree}
\HODTTcj{\Gamma}
\justifies
\HODTTcj{(\Gamma, X \co \Type)}
\end{prooftree}
\]
Type variables can be introduced as types using the rule
\[
\begin{prooftree}
\HODTTcj{\Gamma}
\justifies
\HODTTtj{\Gamma}{}{X}
\using 
X \co \Type \in \Gamma
\end{prooftree}
\]
The exchange rule of dependent type theory should be extended to allow a type variable $X$ to be exchanged with a term variable $x \co \sigma$ if $X$ does not appear in $\sigma$.

Dependent products and sums and subset types are added to the type theory in the usual way~\cite{Jacobs:99}, but we also add a special type constructor called $\later$ which
acts as a functor. The rules are
\[
\begin{prooftree}
\HODTTtj{\Gamma}{\Theta}{\tau}
\justifies
\HODTTtj{\Gamma}{\Theta}{\later\tau}
\end{prooftree}
\GAP
\begin{prooftree}
\HODTTj{\Gamma}{\Theta}{M}{\sigma\to\tau}
\justifies
\HODTTj{\Gamma}{\Theta}{\later(M)}{\later\sigma\to\later\tau}
\end{prooftree}
\]
and the external equality rules include equations expressing 
the functoriality of $\later$. Moreover, we add, for each pair of types $\sigma, \tau$ in the same context, a term of type $\later\sigma \times \later \tau \to \later(\sigma \times \tau)$ plus equations stating that this is inverse to $\pair{\later(\pi_1)}{\later(\pi_2)} \co \later(\sigma \times \tau) \to \later \sigma\times \later \tau$. 

The natural transformation $\next$ is introduced as follows:
\[
\begin{prooftree}
\HODTTtj{\Gamma}{\Theta}{\tau}
\justifies
\HODTTj{\Gamma}{\Theta}{\next_\tau}{\tau\to\later\tau}
\end{prooftree}
\]
plus equality rules stating that $\next_\tau$ is natural in $\tau$
(i.e., $\next_\sigma \circ u = \later(u) \circ \next_\tau$).
We omit term formation rules for fixed point terms. 

We now introduce the notion of {\functorialcontractiveness} which will be used as a condition ensuring well-formedness of recursive types. The definition is a syntactic reformulation of the semantic notion of local contractiveness. 

A type $\tau$ is functorial in $\vec X$ if there is some way to split up the occurences of the variables $\vec X$ in $\tau$ into positive and negative ones, in such a way that $\tau$ becomes a functor expressible in the type theory. Above, and in the exact definition below we use vectors $\vec X$ to denote vectors of type variables and use $\vec x \co \vec \sigma$ to denote vectors of typing assumptions of the form $x_{1}\co \sigma_{1} \dots x_{n}\co \sigma_{n}$. An assumption of the form $\vec f \co \vec X \to \vec Y$ means $f_1 \co X_1 \to Y_1, \dots f_n \co X_n \to Y_n$. 
\begin{defi} \label{def:functorial:type}
Let $\HODTTtj{\Gamma,  \vec X \co \Type}{}{\tau}$ be a valid typing judgement. We say that $\tau$ is \emph{functorial} in $\vec X$ if there exists some other type judgement $\HODTTtj{\Gamma, \vec X \co \Type, \vec Y\co \Type}{}{\tau'}$ and a term 
\[ \HODTTj{\Gamma, \vec X_{-}, \vec Y_{-}, \vec X_{+}, \vec Y_{+}, \vec f \co \vec X_{+} \to \vec X_{-}, \vec g \co \vec Y_{-} \to \vec Y_{+}}{}{\st(\vec f, \vec g)}{\tau'(\vec X_{-}, \vec Y_{-}) \to \tau'(\vec X_{+}, \vec Y_{+})}
\]
(writing $\tau'(\vec X_{-}, \vec Y_{-})$ for $\tau'[\vec X_{-}, \vec Y_{-} / \vec X, \vec Y]$)
such that $\tau'(\vec X, \vec X) = \tau$, and such that $\st$ is functorial in the sense that $\st(\vec \id, \vec \id) = \id$, $\st(\vec f \circ \vec f', \vec g' \circ \vec g) = \st(\vec f', \vec g') \circ \st(\vec f, \vec g)$. 
\end{defi}

The definition of  $\tau$ being {\contractivelyfunctorial} in $\vec X$ is similar, except that the strength $\st(\vec f, \vec g)$ must be defined for $\vec f \co \later(\vec X_{+} \to \vec X_{-}), \vec g \co \later(\vec Y_{-} \to \vec Y_{+})$. To make sense of functoriality write $f' \circ f$ for the composite 
\begin{diagram}
\later(X \to Y) \times \later(Y \to Z) & \rTo^\iso & \later((X \to Y) \times (Y \to Z)) & \rTo^{\later(\comp)} & \later(X \to Z) 
\end{diagram}
applied to $f'$ and $f$. 

\begin{defi} \label{def:functorially:contractive}
Let $\HODTTtj{\Gamma,  \vec X \co \Type}{}{\tau}$ be a valid typing judgement. We say that $\tau$ is \emph{\contractivelyfunctorial} in $\vec X$ if there exists some other type judgement $\HODTTtj{\Gamma, \vec X \co \Type, \vec Y\co \Type}{}{\tau'}$ and a term 
\[ \HODTTj{\Gamma, \vec X_{-}, \vec Y_{-}, \vec X_{+}, \vec Y_{+}, \vec f \co \later(\vec X_{+} \to \vec X_{-}), \vec g \co \later(\vec Y_{-} \to \vec Y_{+})}{}{\st(\vec f, \vec g)}{\tau'(\vec X_{-}, \vec Y_{-}) \to \tau'(\vec X_{+}, \vec Y_{+})}
\]
such that $\tau'(\vec X, \vec X) = \tau$, and such that $\st$ is functorial in the sense that $\st(\vec \id, \vec \id) = \id$, $\st(\vec f \circ \vec f', \vec g' \circ \vec g) = \st(\vec f', \vec g') \circ \st(\vec f, \vec g)$. 
\end{defi}

\begin{lem}
If $\tau$ is {\contractivelyfunctorial} in $\vec X$ then it is also functorial in $\vec X$.
\end{lem}
 
We now give the introduction rule for recursive types
\[
\begin{prooftree}
\HODTTtj{\Gamma, X \co \Type}{}{\tau}
\justifies
\HODTTtj{\Gamma}{}{\rect{X}{\tau}}
\using 
\text{$\tau$ {\contractivelyfunctorial} in $X$}
\end{prooftree}
\]
As usual, there are associated term constructors $\fold{M}$ and $\unfold{M}$ 
that mediate between the recursive type and its unfolding together with
equations expressing that $\mlfont{fold}$ and $\mlfont{unfold}$ are each others inverses.

There is a rich supply of types {\contractivelyfunctorial} in $\vec X$ as can be seen from the following proposition. Proposition~\ref{prop:functorial:types} is stated compactly, and some of the items in fact cover two statements. For example, item~(\ref{item:dep:prod:functorial}) states that if $\sigma$ is functorial, so are $\Prod_{i\co I} \sigma$ and $\Sum_{i\co I} \sigma$ and if $\sigma$ is {\contractivelyfunctorial} so are $\Prod_{i\co I} \sigma$ and $\Sum_{i\co I} \sigma$. 

\begin{prop} \label{prop:functorial:types}
Let $\vec X$ be type variables and let $\sigma, \tau$ be types
\begin{enumerate}[\em(1)]
\item any type variable $X$ is functorial in $\vec X$
\item if $\vec X$ do not appear in $\sigma$ then $\sigma$ is {\contractivelyfunctorial} in $\vec X$
\item\label{item:to:functorial} if $\sigma$ and $\tau$ are both (contractively) functorial in $\vec X$ so are $\sigma \to \tau$ and $\sigma \times \tau$
\item\label{item:dep:prod:functorial} if $\sigma$ is (contractively) functorial in $\vec X$ and $\vec X$ do not appear in $I$ then $\Prod_{i\co I} \sigma$ and $\Sum_{i\co I} \sigma$  are both (contractively) functorial in $\vec X$
\item\label{item:subset:functorial} If $\sigma$ is (contractively) functorial in $\vec X$ (witnessed by $\sigma'$ and $\st_\sigma$) and $\phi$  is a predicate on $\sigma'$ such that 
 \[\phi_{\vec X_{-}, \vec Y_{-}}(x) 
 \imp \phi_{\vec X_{+}, \vec Y_{+}}(\st(\vec f, \vec g)(x)) \]
then $\{ x \co \sigma \mid \phi[\vec X/\vec Y](x) \}$ is (contractively) functorial in $\vec X$. 
\item\label{item:later:functorial} If $\sigma$ is functorial in $\vec X$, then $\later \sigma$ is {\contractivelyfunctorial} in $\vec X$.
\end{enumerate}
\end{prop}
Item~(\ref{item:subset:functorial}) uses the notation $\phi_{\vec X_{-}, \vec Y_{-}}$ for $\phi[\vec X_{-}, \vec Y_{-} / \vec X, \vec Y]$. 

\proof
The proof is a standard construction of functors from type expressions, and we just show a few examples. For~(\ref{item:to:functorial}) if $\sigma'$ and $\tau'$ along with $\st_\sigma$ and $\st_\tau$ witness that $\sigma$ and $\tau$ are functorial, then $\sigma'(\vec Y, \vec X) \to \tau'$ along with $\st_{\sigma\to\tau}(\vec f, \vec g)$ defined as
\[\lambda h \co \sigma'(\vec Y_{-}, \vec X_{-}) \to \tau'(\vec X_{-}, \vec Y_{-}) \ld \st_{\tau}(\vec f, \vec g) \circ h \circ \st_{\sigma}(\vec g, \vec f)
\]
witness that $\sigma\to\tau$ is functorial.

For~(\ref{item:dep:prod:functorial}) the assumption gives us a type $\sigma'$ plus a term
\[
\HODTTj{\Gamma, i\co I, \vec X_{-}, \vec Y_{-}, \vec X_{+}, \vec Y_{+}, \vec f \co \vec X_{+} \to \vec X_{-}, \vec g \co \vec Y_{-} \to \vec Y_{+}}{}{\st_{\sigma}(\vec f, \vec g)}{\sigma'(\vec X_{-}, \vec Y_{-})
\to \sigma'(\vec X_{+}, \vec Y_{+})}
\]
and we can define $\st_{\prod_{i \co I}\sigma}(\vec f, \vec g)$ as
\[
\lambda x \co \Prod_{i\co I} \sigma'(\vec X_{-}, \vec Y_{-}) 
\ld \lambda i \co I \ld \st_{\sigma} (\vec f, \vec g) (x(i))
\]
(This uses the exchange rule mentioned earlier.)

For item \ref{item:subset:functorial} the assumption is exactly the condition needed to show that $\st_{\sigma} (\vec f, \vec g)$ restricts to a term of the type 
\[
\{ x\co \sigma'(\vec X_{-}, \vec Y_{-}) \mid \phi_{\vec X_{-}, \vec Y_{-}}(x) \} \to 
\{ x\co \sigma'(\vec X_{+}, \vec Y_{+}) \mid \phi_{\vec X_{+}, \vec
  Y_{+}}(x) \}\eqno{\qEd}
\]

To allow for \emph{nested} recursive types, one needs to prove that if $\sigma$ is functorial in $\vec X$ and contractively functorial in $Y$, then $\rect{Y}{\sigma}$ is functorial in $\vec X$. In the type theory sketched above this is not provable because in general $\st_{\rect{Y}{\sigma}}(\vec f, \vec g)$ is not definable, but as we shall see when we sketch the interpretation of the type theory, it is safe to add $\st_{\rect{Y}{\sigma}}$ as a constant, together with appropriate equations, such that nested recursive types can in fact be defined.

\begin{rem}
The rules for well-definedness of recursive types are complicated because of the subset types, which require explicit mention of the syntactic strength $\st$. Alternatively, one could give a simple grammar for well-defined recursive types not including subset types, but including nested recursive types not mentioning $\st$, and then show how to interpret these by inductively constructing the contractive strength in the model. We chose the above approach because it is more expressive and because the subset types are needed in applications as illustrated in Section~\ref{subsec:application:worlds}.
\end{rem}

\subsection{Interpreting the type theory} 
\label{sec:interp:dtt} 

\newcommand{\dttenvproj}[2]{p_{#1, #2}}

The interpretation of an open type $\HODTTtj{\Gamma}{}{\sigma}$ is defined modulo an environment mapping the type variables in $\Gamma$ to semantic types, i.e., objects in slice categories. Precisely, if $\Gamma$ is of the form $\Gamma', X\co \Type, \Gamma''$ then $\rho$ should map $X$ to an object of $\SIslice{\den{\Gamma'}_{\rho'}}$ where $\rho'$ is the restriction of $\rho$ to the type variables of $\Gamma'$. The interpretation of open types is defined by induction and most of the cases are exactly as in the usual interpretation of dependent type theory~\cite{Jacobs:99}, and we just mention the new cases. The interpretation of a type variable introduction is defined as $\den{\HODTTtj{\Gamma', X\co \Type, \Gamma''}{}{X}} = \dttenvproj{\Gamma}{\Gamma'}^*(\rho(X))$, where $\dttenvproj{\Gamma}{\Gamma'}$ denotes the projection $\den{\Gamma}_\rho \to \den{\Gamma'}_\rho$. The interpretation of $\later$ is defined as $\den{\HODTTtj{\Gamma}{}{\later \sigma}} = \laterSlice{\den{\Gamma}_{\rho}}(\den{\HODTTtj{\Gamma}{}{\later \sigma}})$. 

For the interpretation of recursive types, note that for every type $\HODTTtj{\Gamma, \vec X}{}{\sigma}$ functorial in $\vec X$ and every environment $\rho$ mapping the free type variables in $\Gamma$ to semantic types, one can define a strong functor of the type
\[
\den{\sigma}_\rho \co (\opcat{\SIslice{\den{\Gamma}_\rho}} \times \SIslice{\den{\Gamma}_\rho})^{\cardinality{\vec X}} \to \SIslice{\den{\Gamma}_\rho} 
\]
as follows. Assuming that the functoriality of $\sigma$ is witnessed by $\sigma'$ and $\st$ as in Definition~\ref{def:functorial:type}, the action of $\den{\sigma}_\rho$ on objects is defined by the interpretation of $\sigma'$. Given objects $\vec A_-, \vec A_+, \vec B_-, \vec B_+$ of $\SIslice{\den{\Gamma}_\rho}$ the interpretation of $\st$ is a morphism in $\SIslice{\den{\Gamma}_\rho}$ of the type 
\[
A_{-,1}^{A_{+,1}} \times \dots \times A_{-,n}^{A_{+,n}} \times B_{+,1}^{B_{-,1}} \times \dots \times B_{+,n}^{B_{-,n}}  \to \den{\sigma}_\rho(\vec A_+, \vec B_+)^{\den{\sigma}_\rho(\vec A_-, \vec B_-)}
\]
where the products and exponentials are those of the slice $\SIslice{\den{\Gamma}_\rho}$. The interpretation of $\st$ defines the strength of $\den{\sigma}_\rho$, from which the  action of $\den{\sigma}_\rho$ on morphisms can be derived in the usual way. 

Similarly, if $\sigma$ is functorial in the $n$ first type variables and {\contractivelyfunctorial} in the last one then the interpretation of the witness $\st$ defines a strong functor which is locally contractive in the last variable and so 
we can define $\den{\rect{X}{\tau}}_\rho = \Fix(\den{\tau}_\rho)$ using the fixed point given by Theorem~\ref{exist:unique:rec:dom:eq}. 

There is a question of well-definedness here, since the fixed point of $\den{\sigma}_{\rho}$ a priori could depend on the choice of $\sigma'$ and $\st$. The uniqueness of the fixed point of Theorem~\ref{exist:unique:rec:dom:eq}, however, ensures that even for different such choices, the resulting $\den{\sigma}_{\rho}$ will be isomorphic. Usually, $\sigma$ comes with a canonical choice of $\sigma'$ and $\st$ as given by Proposition~\ref{prop:functorial:types}.

As mentioned earlier, for allowing nested recursive types in the type theory we need to add constants of the form $\st_{\rect{Y}{\sigma}}(\vec f, \vec g)$. Having sketched the interpretation of the type theory we can now see that it is safe to do so: $\st_{\rect{Y}{\sigma}}(\vec f, \vec g)$ can be interpreted using the strength of $\Fix\den{\sigma}_\rho$ which exists by Theorem~\ref{exist:unique:rec:dom:eq}. 

\paragraph{On Coherence}

Above, we have worked in the codomain fibration and ignored coherence
issues, i.e., the fact that the codomain fibration and the associated
fibred functors needed for the interpretation of the type theory are
not split.  One further advantage of the concrete representation of
slices $\stopos/I$ as presheaves over $\elements{I}$ is that the
latter gives rise to a split model. The idea is to form a split
indexed category $P : \stopos \to \opcat{\catcat}$, with fibre over $I$
given by $P(I) = \BigPsh{\elements{I}}$, and reindexing $P(u : I\to
J)$ given by $P(u)(X)(n,i) = X(u_n(i))$.  By forming the Grothendieck
construction~\cite{Jacobs:99} on $P$ one obtains a split fibration
$\Fam{\stopos}\to\stopos$ which is equivalent to the codomain
fibration. Then one uses this fibration to interpret the types and terms without
free type variables, and uses split fibred functors 
\[
(\opcat{\Fam{\stopos}_{\den{\Gamma}}} \times \Fam{\stopos}_{\den{\Gamma}})^{\cardinality{\Theta}} 
\to \Fam{\stopos}_{\den{\Gamma}}
\]
to interpret open types $\HODTTtj{\Gamma}{\Theta}{\tau}$.
Finally, one checks that the fibred constructs (e.g., right adjoints
to reindexing) used to interpret the
dependent type theory are split, and that $\later$ and the construction of recursive types 
is also split.  The latter essentially boils down to observing that 
the actual construction of initial algebras in Section~8 is done
fibrewise and thus preserved on-the-nose by reindexing. We omit
further details.

\section{Relation to metric spaces}
\label{sec:cbult}

Let $\CBUlt$ be the category of complete bounded ultrametric spaces
and non-expansive maps. In~\cite{BirkedalL:essence,SchwinghammerEtAl:2010:A-Semantic,Schwinghammer:Birkedal:Stovring:11,BirkedalST:10,Birkedal:Schwinghammer:Stovring:10:Nakano}
only those spaces that were also bisected were used:
a metric space is \emph{bisected} if all non-zero distances are of the
form $2^{-n}$ for some natural number $n \geq 0$.  Let $\BiCBUlt$ be the
full subcategory of $\CBUlt$ of bisected spaces, and let $\BiUlt$ be the
category of all bisected ultrametric spaces (necessarily bounded).

Let $t\stopos$ be the full subcategory of $\stopos$ on the total objects. 
\begin{prop}
  There is an adjunction between $BiUlt$ and $\stopos$, which restricts
  to an equivalence between $t\stopos$ and $\BiCBUlt$, as in the diagram:
\begin{diagram}[size=3em]
t\stopos                 & \pile{\lTo\\ \top \\ \rTo}  & \stopos \\
\uTo \iso \dTo      &                               & \uTo^F \dashv \dTo \\
\BiCBUlt                 & \pile{\lTo \\ \bot \\ \rTo}  & \BiUlt \\
\end{diagram}
\end{prop}
\begin{proof}[Proof sketch]
  The functor $F: \BiUlt \to \stopos$ is defined as follows. A space
  $(X,d) \in \BiUlt$ gives rise to an indexed family of equivalence
  relations by $x \mathrel{=_n} x' \iff d(x,x') \leq 2^{-n}$, which
  can then be viewed as a presheaf: at index $n$, it is the quotient
  $X/(=_n)$, see, e.g.~\cite{DiGianantonio:Miculan:04}.
  One can check that $F$ in fact maps into $t\stopos$ and that $F$
  has a right adjoint that maps into $\BiCBUlt$.
  The right adjoint maps a variable set into a metric space on the limit
  of the family of variable sets; the metric expresses up to what level 
  elements in the limit agree.  The left
  adjoint  from $BiUlt$ to $\BiCBUlt$ is 
  given by the Cauchy-completion. 
\end{proof}

\begin{prop}
  A morphism in $\BiCBUlt$ is contractive in the metric sense iff it is
  contractive in the internal sense of~$\stopos$.
\end{prop}

The later operator on $\stopos$ corresponds to multiplying by $\frac{1}{2}$ in
ultra-metric spaces, except on the empty space. Specifically, $F(\frac{1}{2} X)$ is isomorphic
to $\later(F X)$, for all non-empty $X$.  For ultra-metric spaces, the 
formulation of existence of solutions to guarded recursive domain equations
has to consider the empty space as a special case. Here, 
in $\stopos$, we do not have to do so, since $\later$ behaves better than
$\frac{1}{2}$ on the empty set.

\section{General models of guarded recursive terms}
\label{sec:axioms}
\newcommand{\DTTj}[3]{{#1} \vdash {#2} : {#3}} 
\newcommand{\DTTej}[4]{{#1} \vdash {#2} = {#3} : {#4}} 

\newcommand{\ladj}[2]{#1 \dashv #2}
\newcommand{\arrowcat}[1]{#1^{\to}}

Having presented the specific model $\SItopos$ we now turn to general models of 
guarded recursion. We give an axiomatic definition of what models 
of guarded recursion are, and in Section~\ref{sec:class:of:models} we show that $\SItopos$ 
is just one in a large class of models. 

We start by defining a notion of model of guarded recursive terms, and
showing that the class of such models is closed under taking slices.
This result is not only of interest in its own right, but also needed
for showing that the general models of
Section~\ref{sec:class:of:models} model guarded recursive dependent
types.
 
\begin{defi} \label{def:model:guarded:rec:term} A model of
  guarded recursive terms is a category $\etopos$ with finite products
  together with an endofunctor $\later \co \etopos \to \etopos$ and a
  natural transformation $\next \co \id \to \later$ such that
\begin{iteMize}{$\bullet$}
\item for every morphism $f \co \later X \to X$ there exists a unique
  morphism $h \co 1 \to X$ such that $f \circ \next \circ h =
  h$. 
\item $\later$ preserves finite limits
\end{iteMize}
\end{defi}

\begin{lem}
If $\etopos$ models guarded recursive terms then $\later$ is strong.
\end{lem}

\proof
Using $\next$ one can define a strength for $\later$ as the composite
\[ \iso \circ \, \next\times \id \co X \times \later Y \to \later X \times \later Y \to \later (X \times Y) \, .\eqno{\qEd}
\]

The notion of contractive morphism as well as Lemma~\ref{lem:contractive:comp} and Theorem~\ref{thm:fp:op} generalises directly to the current setting.

\begin{thm} \label{theorem:slices} If $\etopos$ is a locally
  cartesian closed model of guarded recursive terms, then so is every slice
  of $\etopos$.
\end{thm}

To prove Theorem~\ref{theorem:slices} we must first show how to
generalise $\later$ to slices.
We do this by taking the pullback diagram of
Proposition~\ref{prop:later:pb} as a definition of $\laterSlice{I}X$.
In other words we define $\laterSlice{I}$ as the composite
\begin{diagram}[LaTeXeqno] \label{eq:later:composite}
\eslice{I} & \rTo^{\later} & \eslice{\later I} & \rTo^{\next^*} & \eslice{I}
\end{diagram}
where the first functor maps $\p X \co X \to I$ to $\later(\p X) \co
\later X \to \later I$ and the second is given by pullback along
$\next$. Recall that $\next^*$ has a left adjoint $\Coprod_{\next}$
mapping $\p Y \co Y \to I$ to $\next \circ \p Y$ and so preserves
limits. It is easy to see that also the first functor of
(\ref{eq:later:composite}) preserves finite limits because $\later$ does,
and thus we have the following:

\begin{lem} \label{lem:later:slices:limits} The functor
  $\laterSlice{I} \co \eslice{I} \to \eslice{I}$ preserves finite
  limits.
\end{lem}

We define $\next_I: \p{Y}\to \p{\laterSlice{I} Y}$ in the slice over $I$ 
as indicated in the diagram below 
\begin{diagram}
Y \\
& \rdTo \rdTo(4,2)^{\next} \rdTo(2,4)_{\p{Y}}   \\
&& \laterSlice{I} Y \SEpbk & \rTo & \later Y \\
&& \dTo_{\p{\laterSlice{I} Y}} && \dTo_{\later \p{Y}} \\
&& I & \rTo^{\next} & \later I
\end{diagram}
It is easy to show that $\next_I$ is a natural transformation.

The following proposition states that $\later$ defines a fibred functor and
hence can serve as a type constructor in the dependent type theory of $\etopos$. 

\begin{prop} \label{prop:later:fibred}
For every $u \co J \to I$ in $\etopos$ the following diagram commutes up to isomorphism
\begin{diagram}
\eslice{I} & \rTo^{\laterSlice{I}} & \eslice{I} \\
\dTo^{u^*} &  & \dTo_{u^*} \\
\eslice{J} & \rTo^{\later_J} & \eslice{J}  \, .
\end{diagram}
As a consequence, the collection of functors $(\laterSlice{I})_{I \in
  \etopos}$ define a fibred endofunctor on the codomain fibration.
\end{prop}

\begin{proof}
We can write the diagram as a composite as below. 
\begin{diagram}
\eslice{I} & \rTo^{\later} & \eslice{\later I} & \rTo^{\next^{*}} &  \eslice{I}\\
\dTo^{u^*} &  & \dTo_{(\later u)^*} &  & \dTo_{u^*} \\
\eslice{J} & \rTo^{\later} & \eslice{\later J}  & \rTo^{\next^{*}} &  \eslice{J} \, .
\end{diagram}
The square on the left commutes because $\later$ preserves pullbacks, the one on the right follows from the naturality square for $\next$. 
\end{proof}

\begin{prop} \label{prop:next:fibred}
  The collection of $\next$ morphisms defines a fibred natural
  transformation from the fibred identity on the codomain fibration to
  $\later$:
\begin{diagram}
\arrowcat{\etopos} && \pile{\rTo^{\id} \\ \Downarrow \next \\ \rTo_{\later}} && \arrowcat{\etopos} \\
& \rdTo_{\mathrm{cod}} && \ldTo_{\mathrm{cod}} \\
&& \etopos
\end{diagram}
\end{prop}

\begin{proof}
  A fibred natural transformation between fibred functors is a natural
  transformation with vertical components. The components of $\next$
  are clearly vertical, but we must show that $\next$ defines a natural
  transformation between the two functors on the total category
  $\arrowcat{\etopos}$. So consider a morphism in $\arrowcat{\etopos}$
  from $Y \to I$ to $X \to J$, and write it as a composition
\begin{diagram}
Y & \rTo^{g} & f^{*}X \SEpbk & \rTo^{\bar f} & X \\
& \rdTo & \dTo && \dTo \\
&& I & \rTo^{f} & J 
\end{diagram}
of a vertical morphism $g$ and a cartesian morphism $\bar f$. We must
verify naturality diagrams for $\next$ with respect to $\bar f$ and
$g$. Naturality wrt.\ $g$ is just naturality of $\next$ as a functor
$\eslice{I} \to \eslice{I}$,
and naturality wrt.\ $\bar f$ can be verified by a diagram chase that we omit.
\end{proof}

It remains to prove the existence (and uniqueness) of fixed points in slices. We do that by reducing those to global fixed points. In the next lemma we use internal language notation, writing $\Prod_{i\co I} X_i$ for the functor 
\begin{diagram}
\eslice I & \rTo^{\Prod_{\bang \co I \to 1}} & \eslice 1 & \rTo^{\iso} & \etopos
\end{diagram}
applied to an object $\p X \co X \to I$, where $\Prod_{\bang \co I \to 1}$ is the right adjoint to $\bang^*$, and using similar notation for the result of applying the same functor to morphisms.

\begin{lem} \label{lem:locally:contractive:slice} Suppose that $f
  \co \p{X} \to \p{Y}$ is a contractive morphism in slice $\eslice I$. Then $\Prod_{i\co I}f_i \co \Prod_{i\co
    I}X_i \to \Prod_{i\co I} Y_i$ is a contractive morphism in
  $\etopos$. As a consequence any contractive endomorphism in $\eslice I$ has a unique fixed point.
\end{lem}

\begin{proof}{}
The assumption gives us a $g$ such that $f = g \circ \next$ and from that we can derive a factorisation of $\Prod_{i\co I}f_i$ as 
\begin{diagram}
\Prod_{i\co I} X_i & \rTo^{\Prod_{i\co I}\next} & \Prod_{i\co I} \later X_i & \rTo^{\Prod_{i\co I}g_i} & \Prod_{i\co I} Y_i
\end{diagram}
To show $\Prod_{i\co I}f_i$ contractive, it suffices to show commutativity of the triangle
\begin{diagram}[LaTeXeqno] \label{eq:Prod:next:contract}
\Prod_{i\co I} X_i && \rTo^{\Prod_{i\co I} \next} && \Prod_{i\co I} \later X_i \\
& \rdTo_{\next} && \ruTo\\
&& \later \Prod_{i\co I} X_i 
\end{diagram}
Writing $\pi_i$ for the term $\DTTj{i \co I}{\lambda x\co \Prod_{i\co I} X_i \ld x_i}{X_i}$ the adjoint correspondent of (\ref{eq:Prod:next:contract}) can be expressed in the internal language of $\etopos$ as
\[
\DTTej{i\co I, x\co \Prod_{i\co I} {X_{i}}}{\later(\pi_i)\circ \next(x)}{\next \circ \pi_i(x)}{\later (X_{i})}
\]
which is simply naturality of $\next$.
This sketch in the internal language can be turned into a formal diagrammatic
argument.

Now, it is easy to see that if $f$ is an endomorphism then there is a bijective correspondence between fixed points of $\Prod_{i\co I}f_i$ in the global sense, and fixed points of $f$ in the slice. 
\end{proof}

\begin{proofof}{Theorem~\ref{theorem:slices}}
We have seen how every slice of $\etopos$ has an endofunctor $\laterSlice{I}$ and a natural transformation $\next \co \id \to \laterSlice{I}$, and we have seen that $\laterSlice{I}$ preserves finite limits (Lemma~\ref{lem:later:slices:limits}). Lemma~\ref{lem:locally:contractive:slice} gives existence of the needed fixed points. 
\end{proofof}

\subsection{A left adjoint to $\later$} \label{sec:left:adjoint}

In our model $\SItopos$, the functor $\later$ has a left adjoint $\earlier$ mapping the presheaf 
\[X(1) \from X(2) \from X(3) \from \ldots \]
to the presheaf
\[X(2) \from X(3) \from X(4) \from \ldots \, .\] Moreover, $\earlier$
preserves limits and so $\ladj{\earlier}{\later}$ defines a geometric
morphism from $\SItopos$ to itself, in fact it is an embedding. 
Hence $\laterSlice{I}$, as defined in (\ref{eq:later:composite}),
has a left adjoint $\earlierSlice{I}$ because
$\next^{*}$ has a left
adjoint $\Sum_{\next}$ and also $\later \co \eslice{I} \to
\eslice{\later I}$ has a left adjoint defined by mapping $\p X \co X
\to \later I$ to its adjoint correspondent $\earlier X \to I$.

Even though $\earlier$ preserves limits, $\earlierSlice{I}$ does
not. The simplest counter example is that of the terminal object
$\id_{I}$ of $\eslice{I}$ which is mapped to the adjoint correpondent
$\previous \co \earlier I \to I$ of $\next \co I \to \later I$. So, in
particular, $\ladj{\earlierSlice{I}}{\laterSlice{I}}$ does not define
a geometric morphism.

We choose not to take $\earlier$ as part of the basic structure of a
model of guarded recursion because $\earlier$ in $\SItopos$ does not
define a fibred functor, and so it cannot be used in an internal
language based on dependent type theory. To see why, observe that if
$f \co J \to I$ then $\earlierSlice{J} f^{*}(\id_{I}) \iso
\earlierSlice{J} (\id_{J}) = \previous_{J}$ and $f^{*}
\earlierSlice{I} (\id_{I}) = f^{*}\previous_{I}$, and these two are in
general not isomorphic. 

Observe also that $\later$ does not preserve dependent products, i.e., the diagram 
\begin{diagram}
\eslice{J} & \rTo^{\laterSlice{J}} & \eslice{J} \\
\dTo^{\Prod_u} &  & \dTo_{\Prod_u} \\
\eslice{I} & \rTo^{\later_I} & \eslice{I}  \, .
\end{diagram}
does not in general commute. The reason is that the diagram obtained
by taking left adjoints to all functors above is the diagram stating
that $\earlier$ is a fibred functor, which we have just established
does not commute.

\subsection{An operation on predicates}

We now assume that $\etopos$ is a topos modelling guarded recursion
and we shall see how to obtain the principle of L{\"o}b induction in
$\etopos$.

As we have seen, $\laterSlice{X}$ preserves limits, hence monos,
and thus defines a map $\laterPred\co\Sub{X} \to \Sub{X}$ for all $X$,
which is easily seen to be order preserving. The term $\next_X$
verifies that $m \leq \laterPred{m}$. As a consequence of
Proposition~\ref{prop:later:fibred} this family is natural in $X$ and
thus, by the usual Yoneda argument, it 
corresponds to an operation on propositions $\laterPred \co \Omega
\to \Omega$.   We now embark on proving the following theorem.

\begin{thm}[L{\"o}b induction] \label{thm:lob:induction} The
  reasoning principle $\forall p \co \Prop\ld (\laterPred p \imp p)
  \imp p$ is valid in $\etopos$.
\end{thm}

To prove the theorem, we need a few lemmas. The first describes the action of
$\laterPred \co \Sub{X} \to \Sub{X}$ as an action on characteristic
maps.

\begin{lem} \label{lem:later:char:map} Let $m \co M \to X$ be a mono
  and let $\charmap{m} \co X \to \Omega$ be its characteristic
  map. Then $\successor \circ \later \charmap{m} \circ \next$ is the
  characteristic map of $\laterPred (m)$, where $\successor\co \later
  \Omega \to \Omega$ is the characteristic map of the mono $\later
  \top \co \later 1 \to \later \Omega$.
\end{lem}

\begin{proof}
Consider the diagram
\begin{diagram}
\laterPred{m} \SEpbk & \rTo & \later M \SEpbk & \rTo & \later 1 \SEpbk & \rTo & 1 \\
\dTo && \dTo^{\later m} & & \dTo^{\later \top} && \dTo^{\top} \\
X & \rTo^{\next} & \later X & \rTo^{\later \charmap{m}}  & \later \Omega & \rTo^{\successor} & \Omega \, .
\end{diagram}
All the squares are pullbacks, and so also the outer square is a pullback, which proves the lemma.
\end{proof}

Subobjects of $X$ correspond to morphisms $X \to \Omega$ which in turn
correspond to global elements of $\Omega^{X}$. As a consequence of
Lemma~\ref{lem:later:char:map}, the operation $\laterPred$ on
subobjects corresponds to composing the global elements with the
morphism $\Omega^{X} \to \Omega^{X}$ mapping $\charmap{m}$ to
$\successor \circ \later \charmap{m} \circ \next$. Since this morphism
is contractive, it has a unique fixed point.

\begin{cor} \label{cor:unique:fixed:pt:later:prop} Let $m$ be a
  subobject of $X$. If $\laterPred(m) \leq m$ then $m$ is the maximal
  subobject.
\end{cor}

\begin{proofof}{Theorem~\ref{thm:lob:induction}}
  The principle is proved using Joyal-Kripke semantics,
  see~\cite[Thm~8.4]{Lambek:Scott:86}. Using items (7) and (6) of the
  referenced theorem, it suffices to show that for any $X$ and any $f
  \co X \to \Omega$, if the map $\lam{x}{X}{\laterPred f(x) \imp
    f(x)}$ factors through $\top \co 1 \to \Omega$, then so does
  $f$. Expressing this using subobjects rather than representable
  maps, we must show that, for any subobject $m$ of $X$, if
  $\laterPred m \imp m$ is the maximal subobject, then so is $m$. But
  $\laterPred m \imp m$ is maximal iff $\laterPred m \leq m$, and so
  the principle follows from
  Corollary~\ref{cor:unique:fixed:pt:later:prop}.
\end{proofof}

\section{General models of guarded recursive types}
\label{axioms:rec:types}
\newcommand{\ecatpush}[2]{\,_{#1}{#2}}

In this section we formulate the most general existence theorem for recursive types in models of guarded recursion. 
Moreover, we reduce the problem of solving general  recursive domain equations to that of
solving covariant domain equations using the uniqueness of fixed points in combination with
Freyd's theory of algebraic compactness
~\cite{Freyd:90, Freyd:90a,
  Freyd:91}. 

Note first that Definition~\ref{def:locally:contractive} of locally
contractive functor on our concrete model $\stopos$, carries over
verbatim to general cartesian closed models $\etopos$ of guarded recursive terms.

\begin{defi}
  A \emph{model of guarded recursive types} is a cartesian closed model of
  guarded recursive terms (in the sense of
  Definition~\ref{def:model:guarded:rec:term}) $\etopos$ such that
  every locally contractive functor $F \co \etopos \to
  \etopos$ has a fixed point (up to isomorphism). A model of guarded
  recursive \emph{dependent} types is a locally
  cartesian closed category whose slices all are models of guarded
  recursive types.
\end{defi}

As a justification of the above definition we shall prove that fixed points for locally contractive covariant
functors give fixed points of general (locally contractive) mixed
variance functors. 
In fact, we state and prove this not only for functors on
$\etopos$, but, more generally, for functors on $\etopos$-enriched
categories.  This is in line with classical work on recursive types in
$O$-categories~\cite{Smyth-Plotkin:SIAMJoC1982} (categories enriched
in complete partial orders) and more recent work on recursive types in
$M$-categories~\cite{BirkedalL:metric-enriched-journal} (categories
enriched in complete bounded ultrametric spaces).

Recall that an $\etopos$-enriched category $\ccat$ is a
collection of objects together with for each pair of objects $X,Y$ of
$\ccat$ an $\etopos$-object $\Hom{\ccat}{X}{Y}$ together with
composition morphisms $\Hom{\ccat}{X}{Y} \times \Hom{\ccat}{Y}{Z} \to
\Hom{\ccat}{X}{Z}$ and morphisms $\nameof{\id_X} \co 1 \to
\Hom{\ccat}{X}{X}$ satisfying commutative diagrams corresponding to
the rules for morphism composition in category theory~\cite{Kelly:enriched-cats}.
To each enriched category $\ccat$ we can associate a
category in the usual sense with the same objects as $\ccat$ and set
of morphisms from $X$ to $Y$ all $\etopos$-morphisms from $1$ to
$\Hom{\ccat}{X}{Y}$. This category is called the
\emph{externalisation} of $\ccat$. Given a category $\ccat$ in the
usual sense, we say that it is $\etopos$-enriched if there exists an
$\etopos$-enriched category whose externalisation is $\ccat$. Any
cartesian closed category $\ccat$ is self-enriched: one can take
$\Hom{\ccat}{X}{Y}$ to be the exponent $Y^X$.

The notion of locally contractive functor readily generalises to
$\etopos$-enriched categories: if $\ccat$ is $\etopos$-enriched
consider the $\etopos$-enriched category $\ecatpush{\later}{\ccat}$
with the same objects as $\ccat$, hom-objects
$\Hom{\ecatpush{\later}{\ccat}}{X}{Y} = \later\Hom{\ccat}{X}{Y}$, 
composition given as the composite
\begin{diagram}
\later\Hom{\ccat}{\PA}{\PB} \times \later\Hom{\ccat}{\PB}{\PC}\iso \later (\Hom{\ccat}{\PA}{\PB} \times \Hom{\ccat}{\PB}{\PC}) & \rTo^{{\later (\comp)}} &  \later\Hom{\ccat}{\PA}{\PC} \\
\end{diagram}
and identity as 
$\next \circ \nameof{\id} \co 1 \to \later\Hom{\ccat}{\PA}{\PA}$. 
Note that $\ecatpush{\later}{(\ccat \times \dcat)} \iso
\ecatpush{\later}{\ccat} \times \ecatpush{\later}{\dcat}$  and 
$\ecatpush{\later}{(\opcat{\ccat})}  \iso
\opcat{(\ecatpush{\later}{\ccat})}$. 
The natural
transformation $\next$ defines an enriched functor~\cite{Kelly:enriched-cats} $\ccat \to
\ecatpush{\later}{\ccat}$ whose action on objects is the identity and
whose action on morphisms is given by $\next \co \Hom{\ccat}{X}{Y} \to
\later\Hom{\ccat}{X}{Y}$.

\begin{defi} \label{def:gen:locally:contractive} An enriched
  functor $F \co \dcat \to \ccat$ is \emph{locally contractive} if it
  factors as a composition of enriched functors
\begin{diagram}
\dcat & \rTo^{\next} & \ecatpush{\later}{\dcat} & \rTo & \ccat
\end{diagram}
\end{defi}
Specialising Definition~\ref{def:gen:locally:contractive} to the case
of $\SItopos$ as self-enriched gives
Definition~\ref{def:locally:contractive}.

\begin{lem}\label{lem:loc:contractive:comp}
\mbox{}
\begin{enumerate}[\em(1)]
\item \label{item:loc:contractive:comp} If $F \co \bcat \to \ccat$ and $G\co \ccat \to \dcat$ are enriched functors and either $F$ or $G$ is locally contractive also $GF$ is locally contractive.
\item \label{item:loc:contractive:prod} If $F\co \ccat \to \dcat$ and $G \co \ccat' \to \dcat'$ are locally contractive, so is $F \times G \co \ccat\times \ccat' \to \dcat\times \dcat'$. 
\item \label{item:loc:contractive:curry} Let $H \co \bcat\times \ccat \to \dcat$ be enriched and suppose the enriched functor category $\dcat^{\ccat}$ exists. Then $H$ is locally contractive in the first variable iff  $\hat H \co \bcat \to \dcat^{\ccat}$ is locally contractive. 
\end{enumerate}
\end{lem}

\begin{defi}
  An $\etopos$-enriched category $\ccat$ is \emph{\complete} if any locally
  contractive functor $F\co\ccat \to \ccat$ has a fixed point, i.e.,
  an object $X$ such that $FX\iso X$.
\end{defi}

The isomorphism $FX\iso X$ is an isomorphism in the externalisation of
$\ccat$. Similarly, the notation $f \co X \to Y$ always refers to
morphisms in the external version of $\ccat$.

We can now state the main theorem. It uses the symmetrization of $\symmetric{G}$
of a mixed variance functor $G$ defined in Section~\ref{sec:rec:dep:types}. The proof
follows after a brief series of lemmas. 

\begin{thm} \label{thm:rec:types:from:init:alg} Let $\etopos$ be a
  model of guarded recursive terms, $\ccat$ be $\etopos$-enriched and
  contractively complete, and
  let $F\co (\opcat{\ccat} \times \ccat)^{n+1} \to \ccat$ be locally
  contractive in the $(n+1)th$ variable pair. Then there exists a
  unique (up to isomorphism) $\Fix F \co (\opcat{\ccat} \times
  \ccat)^{n} \to \ccat$ such that $F \circ \pair{\id}{\widesymmetric{\Fix
      F}} \iso \Fix F$.
Moreover, if $F$ is locally contractive in all variables, so is $\Fix F$. In particular, the above statement holds for $\ccat$= $\etopos$ if $\etopos$ is a model of guarded recursive types. 
\end{thm}

\begin{lem} \label{lem:initial:algebras} Let $\ccat$ be 
  $\etopos$-enriched and let $F \co \ccat \to \ccat$ be a locally contractive
  functor. If $X \iso F(X)$, then the two directions of the
  isomorphism give an initial algebra structure and a final coalgebra
  structure for $F$ on $X$.  In particular, if $F(X) \iso X$ and $F(Y)
  \iso Y$, then $X\iso Y$.
\end{lem}

\begin{proof}
  Given an isomorphism $f \co FX \to X$ and some other algebra $g \co
  FZ \to Z$, $h\co X \to Z$ is an algebra homomorphism iff the diagram
\begin{diagram}
 FX & \rTo^{Fh} & FZ \\
 \uTo^{\inv f} && \dTo_{g} \\
 X & \rTo^{h} & Z
\end{diagram}
commutes, i.e., iff $h$ is a fixed point of the map $h \mapsto g\circ
F(h) \circ \inv f$, which is a contractive endomorphism on
$\Hom{\ccat}{X}{Z}$ (as $F$ is locally contractive). Since this map
has exactly one fixed point, we conclude that there is exactly one
algebra homomorphism from $f$ to $g$. The argument for final
coalgebras is similar.
\end{proof}

There is also a morphism in $\etopos$ computing the unique mediating
homomorphism from the initial algebra.

\begin{lem} \label{lem:initial:algebra:computable} Let $\ccat$ and
  $F$ be as in Lemma~\ref{lem:initial:algebras}, and let $f \co FX \to
  X$ be an isomorphism. For any $Z$ there exists a morphism $k \co
  \Hom{\ccat}{FZ}{Z} \to \Hom{\ccat}{X}{Z}$ such that $\forall g \co
  \Hom{\ccat}{FZ}{Z} \ld k(g) \circ f = g \circ F(k(g))$ holds in the
  internal language of $\etopos$.
\end{lem}

\begin{proof}
  Define $k$ to be the fixed point of the map $\Hom{\ccat}{FZ}{Z}
  \times \Hom{\ccat}{X}{Z} \to \Hom{\ccat}{X}{Z}$ mapping $g, h$ to
  $g\circ Fh \circ \inv f$.
\end{proof}

\begin{lem} \label{lem:initial:algebra:functor}
Let $\ccat, \dcat$ be $\etopos$-enriched categories and let $F\co \dcat \times \ccat \to \ccat$ be enriched and locally contractive in the second variable. If the functor $F(X, -) \co \ccat \to \ccat$ has an initial algebra for all $X$ in $\dcat$, then there is an $\etopos$-enriched functor $\mu F \co \dcat \to \ccat$ mapping $X$ to the carrier of the initial algebra. If, moreover, $F$ is locally contractive in the first variable, then $\mu F$ is locally contractive. 
\end{lem}

\begin{proof}
The functor $\mu F$ is defined (as is standard) to map $f \co X \to Y$ to the unique $\mu F(f)$ making the diagram 
\begin{diagram}[LaTeXeqno] \label{eq:def:init:alg:functor}
F(X, \mu F(X)) & & \rTo && \mu F(X) \\
\dTo^{F(X, \mu F(f))} &&&& \dTo_{\mu F(f)} \\
F(X, \mu F(Y)) & \rTo^{F(f, \id)} & F(Y, \mu F(Y)) & \rTo & \mu F(Y)
\end{diagram}
commute. Now, the enrichment of $\mu F$ is obtained by composing the morphism $\Hom{\dcat}{X}{Y} \to \Hom{\ccat}{F(X, \mu F(Y))}{\mu F(Y)}$ mapping $f$ to the composite in the bottom line of (\ref{eq:def:init:alg:functor}) with the morphism of Lemma~\ref{lem:initial:algebra:computable}. In the case of $F$ being locally contractive in both variables, the first stage of this composite morphism is contractive and so $\mu F$ becomes locally contractive. 
\end{proof}

Recall that an \emph{initial dialgebra} for $G \co \opcat{\ccat} \times \ccat \to \ccat$ 
is an initial algebra of $\symmetric{G}$~\cite{Freyd:90,Freyd:90a,Freyd:91}.

\begin{lem} \label{lem:initial:dialgebra} Let $\ccat$ be
  $\etopos$-enriched and $G \co \opcat{\ccat} \times \ccat \to \ccat$
  be a locally contractive functor. If $G(X,Y) \iso Y$ and $G(Y,X)
  \iso X$ then the pair $(X,Y)$ together with the isomorphisms
  constitute an initial dialgebra for $G$. In particular $(X,Y)$ is
  unique up to isomorphism with this property. Moreover $X \iso Y$.
\end{lem}

\begin{proof}
  If $G$ is locally contractive, so is $\symmetric G$. Thence
  Lemma~\ref{lem:initial:algebras} proves that $(X,Y)$ is an
  initial dialgebra. To show $X \iso Y$ note that the hypothesis of
  the lemma is symmetric in $X$ and $Y$, so we may apply what we have
  just proved to conclude that $(Y,X)$ is an initial dialgebra. By
  uniqueness of initial dialgebras $(X, Y) \iso (Y, X)$.
\end{proof}

We can now give the promised proofs of the main theorem and proposition 
in this section.

\begin{proofof}{Theorem~\ref{thm:rec:types:from:init:alg}}
  Consider first the case of \mbox{$n = 0$}. Recall the functor $\mu F
  \co \opcat{\ccat} \to \ccat$ from
  Lemma~\ref{lem:initial:algebra:functor} mapping $X$ to the unique
  fixed point of $F(X, -)$. Define $Z$ to be the unique fixed point of
  the functor $X \mapsto F(\mu F(X), X)$ and define $W = \mu
  F(Z)$. Then $F(W, Z) = F(\mu F(Z), Z) \iso Z$ and $F(Z, W) = F(Z,
  \mu F(Z)) \iso \mu F(Z) = W$, and so
  Lemma~\ref{lem:initial:dialgebra} applies giving the unique solution
  to $F$ and proving that $W \iso Z$.

In the general case of $n \neq 0$, Lemma~\ref{lem:initial:algebra:functor} applies to give the functor $\Fix F$. 
\end{proofof}

The statement and proof of Proposition~\ref{prop:rec:types:subst} carries over verbatim from the case of $\SItopos$ to the general case of 
$\etopos$ a model of guarded recursive dependent types.

\section{A class of models of guarded recursion}
\label{sec:class:of:models}
\newcommand{\sheaf}{\mathbf{a}}
\newcommand{\nextpresheaf}{\mathrm{next}^{pre}}
\newcommand{\predinv}{\pred^*}

\newcommand{\lub}{\bigvee}
\newcommand{\Idl}[1]{\textit{Idl}(#1)} 

\newcommand{\pred}{p}

\newcommand{\Sh}[1]{\textit{Sh}(#1)}
\newcommand{\base}{K}
\newcommand{\bel}{k} 

\newcommand{\sol}[1]{X_{#1}}
\newcommand{\solIso}[1]{\phi_{#1}}
\newcommand{\solProj}[2]{\pi_{#1,#2}}
\newcommand{\mainOrd}{\ordMap{\lub A}}
\newcommand{\ordA}{\alpha}
\newcommand{\ordB}{\beta}
\newcommand{\ordC}{\beta'}
\newcommand{\ordD}{\gamma}
\newcommand{\card}[1]{|#1|}
\newcommand{\ordMap}[1]{\textit{Ord}(#1)}
\newcommand{\downset}[1]{\downarrow #1}

The aim of this section is to establish a large class of models of
guarded recursive dependent types including our main example, the topos
$\SItopos$. This involves showing existence of fixed points for
locally contractive functors. The special case of $\SItopos$, together
with the results of Section~\ref{axioms:rec:types}, prove 
Theorem~\ref{exist:unique:rec:dom:eq}. 

The class of models we consider are sheaves over  a complete Heyting
algebra with a well-founded basis.  In this section we assume some
familiarity with the basics of complete Heyting algebras and sheaves
over such~\cite{MacLane:Moerdijk:92}.

\begin{defi}
  A partial order $A$ is \emph{well-founded} if there are no infinite
  descending sequences $a_0 > a_1 > a_2 > \dots$
\end{defi}
Here $a>a'$ means $a \geq a'$ and $a \neq a'$ as usual. 
Note that any forest is well-founded.

\begin{defi}
  Let $A$ be a partial order and let $\base\subseteq A$. Then $\base$
  is a \emph{basis} for $A$ if each $a \in A$ is a least upper bound
  of all the base elements below it, i.e. $a = \lub \{ \bel \in \base
  \mid \bel \leq a\}$.
\end{defi}

\begin{exa} \label{ex:ideal:completion} If $\base$ is a
  well-founded partial order then the ideal completion $\Idl{\base}$
  consisting of down-closed subsets of $\base$ is a complete Heyting
  algebra and the set $\{\downset \bel \mid \bel
  \in \base \}$, where $\downset \bel = \{ \bel' \in \base \mid \bel'
  \leq \bel\}$ is a well-founded basis.
\end{exa}

In the following we reserve $a$'s and $b$'s for elements of $A$ and
$\bel$'s for elements in $\base$.  A sieve $B$ on $a$ in $A$ is just a
downward closed subset of $\{b \in A \mid b \leq a\}$ and it is
covering if $\lub B = a$. If $A$ is a complete Heyting algebra then this defines a Grothendieck topology, and
the corresponding category $\Sh{A}$ of sheaves is the full subcategory
of presheaves $\PA$ such that $(\PA(\lub B) \to \PA(b))_{b \in B}$ is
a limiting cone for all $B \subseteq A$.
We recall the following well-known fact.

\begin{prop} \label{prop:idl:compl}
 If $A$ is a partial order then $\Sh{\Idl{A}} \simeq \Psh{A}$.
\end{prop}

\begin{proof}
  The equivalence maps $\PA$ in $\Psh{A}$ to $\lambda B \ld \lim_{b \in
    B} \PA(b)$ (we shall write $\bar \PA$ for this sheaf) and $\PB$ in
  $\Sh{\Idl{A}}$ to $\lambda a \ld \PB(\downset a)$.
\end{proof}

Collectively Proposition~\ref{prop:idl:compl} and
Example~\ref{ex:ideal:completion} state that the general class of
models we consider include all toposes of the form $\Psh{A}$ for $A$ a
well-founded partial order, in particular all slices of $\SItopos$.

\begin{thm} \label{thm:Sh:model:g:rec:types} Let $A$ be a complete
  Heyting algebra with a well-founded base. Then $\Sh{A}$ is a model
  of guarded recursive dependent types. In particular $\SItopos$ and
  indeed any topos of the form $\Psh{A}$ for $A$ a well-founded
  partial order is a model of guarded recursive dependent types.
\end{thm}

Di Gianantonio and Miculan~\cite{DiGianantonio:Miculan:04} essentially
prove that $\Sh{A}$ is a model of guarded recursive \emph{terms} if $A$ is
the set of opens of a topological space with a well-founded basis; here we extend
their results to guarded recursive \emph{types} and, moreover, consider more general
models (not necessarily arising from topological spaces). 

\begin{thm} \label{thm:Sh:model:g:rec:types:enr} Let $A$ be a
  complete Heyting algebra with a well-founded basis and let $\ccat$ be
  a $\Sh{A}$-enriched category. If $\ccat$ is complete (precisely, the
  externalisation of $\ccat$ is complete in the usual sense) then it
  is {\complete}.
\end{thm}
Note that the notion of completeness assumed for $\ccat$ above is the usual one
(rather than the enriched notion of completeness).

In the remainder of this section we prove
Theorems~\ref{thm:Sh:model:g:rec:types:enr}
and~\ref{thm:Sh:model:g:rec:types}. We start by showing that $\Sh{A}$
models guarded recursive terms.

\subsection{Modelling recursive terms}

Following~\cite{DiGianantonio:Miculan:04} we give the following definition.

\begin{defi}
  Define the \emph{predecessor} map $\pred \co A \to A$ by
  \[\pred(a) = \lub \{ \bel \in \base \mid \bel < a\} \, .
  \]
\end{defi}

The predecessor map induces an endofuntor on the category of presheaves
on $A$; following standard notation, we write $\predinv \co \Psh{A} \to
\Psh{A}$ for this functor, defined by $\predinv(\PA) = \PA \circ\pred$.
We define $\later \co \Sh{A} \to \Sh{A}$ by $\later \PA =
\sheaf(\predinv \PA)$,  where $\sheaf$ is the associated sheaf
functor. Define $\nextpresheaf \co \PA \to \predinv\PA$ by
$\nextpresheaf_a (x \in \PA(a)) = \restrict{x}{\pred(a)}$ and define
$\next = \sheaf(\nextpresheaf) \co \PA \to \later \PA$ for all sheaves
$\PA$.

Note that 
\begin{equation}
\next = \eta \circ \nextpresheaf \label{eq:next}
\end{equation}
where $\eta$ is the unit of the adjunction
$\ladj{\sheaf}{I}$, with $I\co \Sh{A}\to\Psh{A}$ the inclusion of
sheaves into presheaves.  This can be seen by applying
$\sheaf$ to both sides of the equation since $\sheaf$ fixes maps
between sheaves and because $\sheaf(\eta)$ is the identity.

\begin{rem}
  The use of the associated sheaf functor $\sheaf$ in the definition
  of $\later$ is necessary, because $\predinv\PA$ needs not be a
  sheaf. Consider, for example, the situation where $A$ is the
  powerset of a 2-element set $\{a,b\}$. Then a sheaf is a presheaf
  $\PA$ such that $\PA(\emptyset) = 1$ and $\PA(\{a,b\}) = \PA(\{a\})
  \times \PA(\{b\})$. The map $\pred$ is
  \begin{align*}
    \pred(\emptyset) & = \emptyset &  \pred(\{a\}) & = \emptyset \\
    \pred(\{b\}) & = \emptyset & \pred(\{a, b\}) & = \{a, b\}
  \end{align*}
  So $\predinv\PA(\{a, b\}) = \PA(\{a, b\})$, but $\predinv\PA(\{a\})
  = \predinv\PA(\{b\}) = 1$, in particular $\predinv\PA$ is in general not a
  sheaf. On the other hand $\later \PA = 1$.
\end{rem}

\begin{lem}
The functor $\later$ preserves finite limits.
\end{lem}

We will now show that the above definition of $\later$
generalises the definition of $\later$ from Section~\ref{sec:dtt:later:slice}
on slices of $\SItopos$, see Proposition~\ref{prop:later:psh:trees}
below. For that we first need a lemma:

\begin{lem}
  Let $A$ be a partial order. The composite
  \begin{diagram}
    \BigPsh{\Idl{A}} & \rTo^{\sheaf} & \Sh{\Idl{A}} & \rTo^{\simeq} &
    \BigPsh{A}
  \end{diagram}
  maps $P$ to $\lambda b \ld P(\downset b)$. In other words $\sheaf
  P(\downset b) = P(\downset b)$.
\end{lem}

\begin{proof}
  Since $\sheaf$ is left adjoint to the inclusion, the composite
  sought for is left adjoint to the functor $P \mapsto \bar P$, and it
  is easy to check that the functor of the lemma satisfies this
  condition.
\end{proof}

\begin{prop} \label{prop:later:psh:trees} Let $I$ be an object of $\SItopos$. The composite
\begin{diagram}
\textstyle{\Psh{\elements I} \simeq \Sh{\Idl{\int I}}} & \rTo^{\later} &
  \textstyle{\Sh{\Idl{\elements I}} \simeq \Psh{\elements I}}
\end{diagram}
  which we shall also call $\later$ agrees with $\laterSlice{I}$ as defined in Section~\ref{sec:dtt:later:slice}
\end{prop}

\begin{proof}
  We compute
  \begin{align*}
    \later P(n,i) & = \later \bar P(\downset (n,i)) \\
    & = (\sheaf \,p^* \bar P)(\downset (n,i)) \\
    & = (p^* \bar P)(\downset (n,i)) \\
    & = \bar P(p(\downset (n,i)))
  \end{align*}
  Now, it is easy to see that if $n=1$ then $p(\downset (n,i)) =
  \emptyset$ so that $\later P(1,i) = \bar P(\emptyset) = 1$ and 
   and otherwise  
\begin{align*}
\bar P(p(\downset (n,i))) & \, = \bar P( \downset (n-1, \restrict{i}{n-1})) \\
& \, = P(n-1, \restrict{i}{n-1}) \,
\end{align*}
which implies the result.
\end{proof}

Using the well-founded basis we can reason by well-founded induction
over $A$ as the following easy lemma shows.

\begin{lem} \label{lem:induction} 
  Let $\phi(a)$ be a predicate on $A$. If
  \begin{align*}
    \forall a \in A \ld (\forall \bel\co \base \ld \bel < a \imp
    \phi(\bel)) \imp \phi(a)
  \end{align*}
  then $\phi(a)$ holds for all $a$ in $A$.
\end{lem}

\begin{proof}
  First use well-founded induction to conclude that $\phi(\bel)$
  holds for all $\bel\in \base$, then use the condition again to
  conclude that $\phi(a)$ holds for all $a$.
\end{proof}

We now aim to show that any morphism $f \co \later \PA \to \PA$ has a
unique fixed point. Since the associated sheaf functor is left adjoint
to the inclusion of sheaves into presheaves such morphisms correspond
bijectively to morphisms of presheaves $\hat f \co \predinv\PA \to
\PA$ (where $\hat f = f \circ \eta$), and we shall start by constructing fixed points of morphisms of
the latter form.

\begin{lem} \label{lem:fix} 
  Let $\PA$ be a sheaf and let $f \co \predinv\PA \to \PA$ and $a
  \in A$. Then there exists a unique family $(x_b)_{b \leq a}$ such
  that
\begin{enumerate}[\em(1)]
\item $\restrict{x_a}{b} = x_{b}$  for all $b \leq a$ \label{item:compatibility}
\item $f_{b}(x_{\pred b}) = x_{b}$ for all $b \leq a$ \label{item:fp}
\end{enumerate}
\end{lem}

\begin{proof}
  The proof is by well-founded induction on $a$ using
  Lemma~\ref{lem:induction}. Thus suppose the lemma holds for all $\bel <
  a$, i.e., for any $\bel < a$ there exists a unique family
  $(x_{\bel,b})_{b \leq \bel}$ satisfying the requirements. Note that
  by uniqueness, if $b \leq \bel' \leq \bel$ then $x_{\bel,b} =
  x_{\bel',b}$, so for any $b < a$ we can define $x_b$ to be the unique
  amalgamation of the family $(x_{\bel,\bel})_{\bel \leq b}$. This
  gives us a compatible family $(x_b)_{b <a}$, i.e., $x_{b'} =
  \restrict{x_b}{b'}$ if $b' < b$. To see that this family also
  satisfies (\ref{item:fp}), for all $b < a$, note that it suffices to
  show that $\restrict{f_{b}(x_{\pred b})}{\bel} = x_{\bel}$, for all
  $\bel \leq b$. But
  \begin{align*}
    \restrict{f_{b}(x_{\pred b})}{\bel} \, & = f_{\bel}(x_{\pred \bel}) \\
    & = x_{\bel} \\
  \end{align*}
  since the family $(x_{\bel,b})_{b \leq \bel}$ satisfied (\ref{item:fp}).

  It only remains to extend this family with a component $x_a$.  By
  the sheaf condition there is a unique $y$ in $\PA(\pred(a))$ such
  that $\restrict{y}{b} = x_b$. Define $x_a = f_a(y)$. We must check
  that the extended family $(x_b)_{b \leq a}$ satisfies the
  conditions, and all that remains to prove is the case of $b = a$.

  For (\ref{item:compatibility}) we must show that
  $\restrict{x_a}{b} = x_b$ for all $b < a$.
  \begin{align*}
  \restrict{x_a}{b} & = \restrict{f_a(y)}{b} \\
   & =  f_b (\restrict{y}{\pred b}) \\
  & =  f_b (x_{\pred b}) \\
  & = x_b
  \end{align*}

  For (\ref{item:fp}) we branch on
  whether $a = \pred a$ or not (using classical reasoning).
  If $\pred a < a$ then $y= x_{\pred a}$, and we are done. If
  $a = \pred a$ then, by the sheaf condition, it suffices to prove
  that $\restrict{f_a(x_a)}{b} = x_b$ for all $b < a$. But
  \begin{align*}
    \restrict{f_a(x_a)}{b} & = f_b (\restrict{(f_a(y))}{\pred(b)}) \\
    & = f_b (f_{\pred (b)}(\restrict{y}{\pred\pred(b)})) \\
    & = f_b (f_{\pred (b)}(x_{\pred\pred(b)})) \\
    & = f_b (x_{\pred b}) \\
    & = x_b
  \end{align*}
  For the proof of uniqueness, we must show that $x_a$ as defined
  above gives the unique extension of $(x_b)_{b < a}$ satisfying the
  conditions. Again we branch on $\pred a = a$ or $\pred a < a$. In
  the first case, (\ref{item:compatibility}) together with the
  sheaf condition gives uniqueness and in the second it is 
  (\ref{item:fp}) that gives uniqueness.
\end{proof}

\begin{thm} \label{thm:Sh:model:g:rec:terms} If $A$ is a complete
  Heyting algebra with a well-founded basis then every slice of $\Sh{A}$ is a model
  of guarded recursive terms.
\end{thm}

\begin{proof}
  By Theorem~\ref{theorem:slices} it suffices to show that $\Sh{A}$
  is a model of guarded recursive terms, and for this it remains to show 
  that if $f \co \later \PA \to \PA$, then there
  exists a unique $\fix (f) \co 1 \to \PA$ such that $f \circ \next
  \circ \fix (f) = \fix (f)$

  Consider first $f\circ \eta \co \predinv\PA \to \PA$.  The family
  $(x_b)_{b \leq \lub A}$ given by Lemma~\ref{lem:fix} defines a map
  $\fix(f) \co 1 \to \PA$: the naturality condition needed to have a map in $\Psh{A}$
  is (\ref{item:compatibility}) and
  (\ref{item:fp}) states
  \begin{equation}
    f\circ \eta \circ \nextpresheaf \circ \fix (f) = \fix (f) \label{eq:presheaf:fix:pt}
  \end{equation}
  which by (\ref{eq:next}) is equivalent to $f \circ
  \next \circ \fix (f) = \fix (f)$.
  In fact we see that to give a map $\fix (f) \co 1 \to \PA$
  satisfying the (\ref{eq:presheaf:fix:pt}) is the same as giving a
  family $(x_b)_{b \leq \lub A}$ and so the uniqueness statement of
  Lemma~\ref{lem:fix} shows that $\fix (f) \co 1 \to \PA$ is the
  unique such map.
\end{proof}

\subsection{Recursive types in sheaf models}

Having proved that $\Sh{A}$ models guarded recursive terms, we now
show that it models guarded recursive dependent types. We first prove
Theorem~\ref{thm:Sh:model:g:rec:types:enr} and then show how
Theorem~\ref{thm:Sh:model:g:rec:types} follows from it. So in the
following, let $\ccat$ be a complete $\Sh{A}$-enriched
category. 

In the technical development it is simpler to work with presheaves and
$\predinv$ than it is to work with sheaves and $\later$, so we first
reformulate the definition of local contractiveness in terms of
$\predinv$. Note that we can define $\ecatpush{\predinv}{\ccat}$
in the same way as we defined $\ecatpush{\later}{\ccat}$, using
$\predinv$ rather than $\later$. This gives us an $\Psh{A}$-enriched
category rather than a $\Sh{A}$-enriched one. Any $\Sh{A}$-enriched
category is also $\Psh{A}$-enriched and so in particular, $\ccat$ and
$\ecatpush{\later}{\ccat}$ are $\Psh{A}$-enriched.

There is a commutative diagram of $\Psh{A}$-enriched functors
\begin{diagram}
\ccat & \rTo^{\nextpresheaf} & \ecatpush{\predinv}{\ccat} \\
& \rdTo^{\next} & \dTo_{\eta} \\
&& \ecatpush{\later}{\ccat}
\end{diagram}
and the following lemma tells us that we can proceed to 
work with $\predinv$ and presheaves rather than $\later$ and sheaves.

\begin{lem} \label{lem:local:contractive:alt:form} An enriched
  functor $F\co \ccat \to \ccat$ is \emph{locally contractive} iff
  there exist a $\Psh{A}$-enriched functor $H \co
  \ecatpush{\predinv}{\ccat} \to \ccat$ such that $H \circ
  \nextpresheaf = F$.
\end{lem}

\begin{proof}
  If $F$ is locally contractive and $G$ is a witness of this, we can
  construct $H$ by precomposing $G$ with $\eta$.  On the other hand,
  given $H$ as above we can construct $G$ by applying $\sheaf$ to
  each hom-action of $H$.
\end{proof}

Now suppose $F \co \ccat \to \ccat$ is locally contractive. We will
construct a fixed point for $F$ by a sufficiently large induction. To determine the height of the induction we
start by assigning to each element $a$ of $A$ an ordinal by well-founded
induction on $a$. We use ordinals (rather than just the elements of $A$)
to get a linear diagram to take limits over when constructing the fixed
point for $F$.

\begin{defi} \label{def:ord:map} Define for each $a \in A$ the
  ordinal $\ordMap{a} = \sup \{\ordMap{\bel} + 1 \mid \bel < a \meet
  \bel \in \base\}$.
\end{defi}

\begin{lem} \label{lem:ordMap:order:pres}
  Definition~\ref{def:ord:map} defines an order preserving map $\ordMap{-} \co A \to
  \mainOrd$. If $\bel < a$ and $\bel \in \base$ then $\ordMap{\bel} <
  \ordMap{a}$.
\end{lem}

We shall use $\pred \co \mainOrd \to \mainOrd$ defined as
$\pred(\ordA) = \lub \{ \ordB \mid \ordB < \ordA \}$.

In the following we distinguish notationally between ordinals and
elements of $A$ by using Greek letters for the former and latin
letters for the latter.

Next we generalise the notion of $n$-isomorphism of
Lemma~\ref{lem:lcf:n:iso}. Recall that a morphism $f \co X \to
Y$ in $\ccat$ is the same as a morphism $1 \to \Hom{\ccat}{X}{Y}$ in
$\Sh{A}$, which is the same as a family $(f_{a})_{a \in A}$ with
$f_{a} \in \Hom{\ccat}{X}{Y}_{a}$ such that $\restrict{f_{a}}{b} =
f_{b}$ for all $a$ and $b\leq a$. We say that $f_{a}$ is an
isomorphism if there exists $g_{a} \in \Hom{\ccat}{X}{Y}_{a}$ such
that $\comp_{a}(f_{a}, g_{a}) = \id_{a}$ and $\comp_{a}(g_{a}, f_{a})
= \id_{a}$. In the following we shall simply write $f_{a} \circ g_{a}$
for $\comp_{a}(g_{a}, f_{a})$.

\begin{defi}
  Let $f \co \PA \to \PB$ be a morphism in $\ccat$, let $a \in A$ and
  let $\ordA$ be an ordinal. We say that $f$ is an
  \emph{$a$-isomorphism} if for all $b \leq a$ the component $f_{b}$
  is an isomorphism. We say that $f$ is an \emph{$\ordA$-isomorphism} if
  it is a $b$-isomorphism, for all $b$ such that $\ordMap{b} \leq
  \ordA$.
\end{defi}

\begin{lem} \label{lem:b:iso} Let $F \co \ccat \to \ccat$ be locally
  contractive, and suppose $f \co \PA \to \PB$ is a $b$-isomorphism
  for all $b < a$. Then $Ff$ is an $a$-isomorphism. As a consequence,
  $Ff$ is an $\ordA$-isomorphism if $f$ is a $\ordB$-isomorphism, for
  all $\ordB < \ordA$, or, equivalently, if $f$ is a $\pred(\ordA)$
  isomorphism.
\end{lem}

\begin{proof}
Formulating the assumption of local contractiveness using the equivalent condition of Lemma~\ref{lem:local:contractive:alt:form} we get maps $H_{\PA, \PB} \co \predinv\Hom{\ccat}{\PA}{\PB} \to \Hom{\ccat}{F\PA}{F\PB}$ such that 
\[(Ff)_{b} = H_{b}(f_{\pred(b)})
\]
The functoriality conditions on $H$ are commutative diagrams in $\Psh{A}$. These amount to the following equations required to hold for each $b$ in $A$
\begin{align}
H_{b}(f_{\pred(b)} \circ g_{\pred(b)}) & = H_{b}(f_{\pred (b)}) \circ H_{b}(g_{\pred (b)}) \label{eq:pres:comp} \\
H_{b}(\id_{\pred(b)}) & = \id_{b} \label{eq:pres:id} 
\end{align}
Now, suppose $f \co \PA \to \PB$ is a $b$-isomorphism, for all $b <
a$. Define $\inv{f}_{\pred(a)}$ to be the unique amalgamation of $(\inv{f}_{b})_{{b<a}}$. Then $\inv{f_{\pred(a)}}$
is an inverse to $f_{\pred(a)}$: to show $\inv{f}_{\pred(a)} \circ f_{\pred(a)} = \id_{\pred(a)}$ it suffices to show $\restrict{(\inv{f}_{\pred(a)} \circ f_{\pred(a)})}{b} = \id_{b}$ for all $b<a$, which is clear since composition commutes with restriction.

So $f_{b}$ has an inverse
$\inv{f}_{b}$ for all $b \leq \pred(a)$, in particular $f_{\pred(b)}$
has an inverse, for all $b \leq a$. Equations (\ref{eq:pres:comp}) and
(\ref{eq:pres:id}) then say that $H_{b}(\inv {f}_{\pred(b)})$ is an
inverse of $F(f)_{b}$, for all $b \leq a$.

For the last statement, suppose $f$ is a $\ordB$-isomorphism for all
$\ordB < \ordA$, and suppose $\ordMap{a} \leq \ordA$. We must show that
$Ff$ is an $a$-isomorphism. By what we have just proved, it suffices
to show that $f$ is a $b$-isomorphism, for all $b < a$, and for this,
by the sheaf property, it suffices to show that $f$ is a
$\bel$-isomorphism, for all $\bel < a$, $\bel \in \base$. But this is
true because $\ordMap{\bel} < \ordMap{a} \leq \ordA$. 
\end{proof}

\begin{rem}
  The strengthening of the definition of locally contractive functor
  compared to the definition used in the conference version of this
  paper~\cite{BirkedalL:sgdt} was introduced in order to make
  Lemma~\ref{lem:b:iso} true, also with the weaker notion of
  $a$-isomorphism used here. Without the requirement of functoriality
  of $H$, equation (\ref{eq:pres:comp}) only holds for families
  $(f_{b})_{b \leq \pred (a)}$, $(g_{b})_{b \leq \pred (a)}$ in the
  image of $\next$, i.e., families that extend to families $(f_{b})_{b
    \leq a}$, $(g_{b})_{b \leq a}$
\end{rem}
 
We construct, by well-founded induction, for every $\ordA \leq
\ordMap{\lub A}$ a $\ccat$-object $\sol{\ordA}$ and maps
\begin{displaymath}
  \solIso{\ordA}  \co F(\sol{\ordA}) \to \sol{\ordA} \qquad \text{and}\qquad
  \solProj{\ordA}{\ordB} \co \sol{\ordA} \to \sol{\ordB}, \quad\text{for}\; \ordB < \ordA
\end{displaymath}
by 
\begin{align*}
\sol{\ordA} & = \lim_{\ordB < \ordA} F(\sol{\ordB}) 
\end{align*}
and 
\begin{diagram}
\solProj{\ordA}{\ordB}: \, &  \lim_{\ordB' < \ordA} F(\sol{\ordB'})  & \rTo^{\pi_{\ordB}} & F(\sol{\ordB})  & \rTo^{\solIso{\ordB}} & \sol{\ordB} \\
\solIso{\ordA} : \, &  F(\lim_{\ordB < \ordA} F(\sol{\ordB})) & \rTo^{F(\lim_{\ordB < \ordA} \solIso{\ordB})} & F(\lim_{\ordB < \ordA} (\sol{\ordB})) & \rTo & \lim_{\ordB < \ordA} F(\sol{\ordB}) 
\end{diagram}
Precisely, each $\ordA$ is an ordered set and so can be considered a category. We define 
$\sol{\ordA}$ as the limit of a diagram indexed over
$\ordA$ mapping an inequality $\ordC \leq \ordB < \ordA$ to
$F(\solProj{\ordB}{\ordC}) \co F(\sol{\ordB}) \to F(\sol{\ordC})$.

\begin{thm} \label{thm:gen:sol} Each $\solProj{\ordA}{\ordB}$ is a
  $\ordB$-isomorphism and each $\solIso{\ordA}$ is an
  $\ordA$-isomorphism. In particular, $\solIso{\ordMap{\lub A}} \co
  F(\sol{\ordMap{\lub A}}) \to \sol{\ordMap{\lub A}}$ is an
  isomorphism.
\end{thm}

Before we give the proof we record the following simple lemma.

\begin{lem} \label{lem:smaller:limit} Let $\ordA > \ordB$ and let
  $(Y_{\ordC})_{\ordC < \ordA}$ be a diagram over $\ordA$ considered a
  category. The morphism $\lim_{\ordC < \ordA} Y_{\ordC} \to
  \lim_{\ordB \leq \ordC < \ordA} Y_{\ordC}$ given by diagram
  inclusion is an isomorphism.
\end{lem}

\begin{lem} \label{lem:lub}
$\ordA \leq \lub \{\ordD \mid \pred \ordD < \ordA \}$  
\end{lem}

\begin{proof}
Recall that for ordinals $\ordB < \ordD$ is equivalent to $\ordB \in \ordD$, and so $\pred(\ordD) = \bigcup_{\ordB \in \ordD} \ordB$. For the lemma we must show that if $x \in \ordA$ also $x \in \bigcup_{\pred \ordD \in \alpha}  \ordD$, i.e., there exists a $\ordD$ such that $x \in \ordD$ and $(\bigcup_{\ordB \in \ordD} \ordB) \in \alpha$. Take $\ordD = \{ \ordB \mid \ordB \leq x\}$.
\end{proof}

\begin{proofof}{Theorem~\ref{thm:gen:sol}}
  The theorem is proved by induction on $\ordA$, but the induction
  hypothesis must be strengthened with the following two statements.
\begin{enumerate}[(1)]
\item \label{item:proj:basic} For all $\ordB < \ordA$, the
  projection \[\pi_{\ordB} \co \lim_{\ordC < \ordA} \sol{\ordC} \to
  \sol{\ordB}\] is a $\ordB$-isomorphism.
\item \label{item:proj:F} For all $\ordB < \ordA$ and all $\ordD$ such
  that $\pred \ordD \leq \ordB$, the projection
\[\pi_{\ordB} \co \lim_{\ordC < \ordA} F(\sol{\ordC}) \to F(\sol{\ordB})\] 
is a $\ordD$-isomorphism. In particular, each $\pi_{\ordB}$ is a $\ordB$-isomorphism.
\end{enumerate}

We now give the induction steps of the inductive proof, proving each part of the induction hypothesis in turn. 

For (\ref{item:proj:basic}) note first that by Lemma~\ref{lem:smaller:limit} we may replace the limit $\lim_{\ordC < \ordA} \sol{\ordC}$ by $\lim_{\ordB \leq \ordC < \ordA} \sol{\ordC}$. By the induction hypothesis, all morphisms of the form 
$\solProj{\ordC}{\ordC'} \co \sol{\ordC} \to \sol{\ordC'}$ 
for $\ordB \leq \ordC' < \ordC < \ordA$ are $\ordB$-isomorphisms. Therefore the limit $\lim_{\ordB \leq \ordC < \ordA} \sol{\ordC}$ is a limit of a diagram of $\ordB$-isomorphisms. Since limits are computed pointwise, the projections are $\ordB$-isomorphisms. 

For (\ref{item:proj:F}) we reason similarly and conclude by
Lemma~\ref{lem:b:iso} that each $F(\solProj{\ordC}{\ordC'})$ is a
$\ordD$-isomorphism. So in this case the limit $\lim_{\ordB \leq \ordC
  < \ordA} F(\sol{\ordC})$ is a limit of a diagram of
$\ordD$-isomorphisms and each projection $\pi_{\ordB}$ is a $\ordD$-isomorphisms.

Now consider the case of $\solProj{\ordA}{\ordB} = \solIso{\ordB} \circ \pi_\ordB$. By (\ref{item:proj:F}) above and the induction hypothesis, this is a $\ordB$-isomorphism. 

We will now show that $\solIso{\ordA}$ is an $\ordA$-isomorphism. Consider the following commutative diagram
\begin{diagram}
\lim_{\ordC < \ordA} F(\sol{\ordC}) & \rTo^{\lim_{\ordC < \ordA} \solIso{\ordC}} & \lim_{\ordC < \ordA} (\sol{\ordC}) \\
\dTo_{\pi_\ordB} & & \dTo_{\pi_\ordB} \\
F(\sol{\ordB}) & \rTo_{\solIso{\ordB}} & \sol{\ordB}
\end{diagram}
Since (\ref{item:proj:basic}) and (\ref{item:proj:F}) state that both projections $\pi_\ordB$ are $\ordB$-isomorphisms and by induction hypothesis $\solIso{\ordB}$ is a $\ordB$-isomorphism, also $\lim_{\ordC < \ordA} \solIso{\ordC}$ must be a $\ordB$-isomorphism. Since this holds for all $\ordB < \ordA$, by Lemma~\ref{lem:b:iso} also $F(\lim_{\ordB < \ordA} \solIso{\ordB})$ must be an $\ordA$-isomorphism.

Now, consider the diagram
\begin{diagram}
F(\lim_{\ordC < \ordA} \sol{\ordC}) & \rTo & \lim_{\ordC < \ordA} F(\sol{\ordC}) \\
& \rdTo_{F(\pi_\ordB)} & \dTo_{\pi_\ordB} \\
&& F(\sol{\ordB})
\end{diagram}
It only remains to show that the vertical map is an $\ordA$-isomorphism. 
By induction hypothesis (\ref{item:proj:F}) the maps $F(\pi_\ordB)$ and $\pi_\ordB$ are $\ordD$-isomorphisms for any $\ordD$ such that $\pred \ordD \leq \ordB$. Since this holds for all $\ordB$, the vertical map is a $\lub \{\ordD \mid \pred \ordD < \ordA \}$-isomorphism, and we conclude by Lemma~\ref{lem:lub}. 
\end{proofof}

\begin{proofof}{Theorem~\ref{thm:Sh:model:g:rec:types:enr}}
We must show that any locally contractive endofunctor $F\co \ccat \to \ccat$ has a fixed point, but Theorem~\ref{thm:gen:sol} gives such a fixed point.
\end{proofof}

For Theorem~\ref{thm:Sh:model:g:rec:types} it remains to show that any
slice of $\Sh{A}$ is a model of guarded recursive types. We do that by
reducing to Theorem~\ref{thm:Sh:model:g:rec:types:enr}, using the
fact that slices of $\Sh{A}$ are all $\Sh{A}$-enriched. Indeed
this holds for any locally cartesian closed category $\etopos$,
because one can take as homobject from $\p X$ to $\p Y$ the object $\Prod_{i\co I} {Y_{i}}^{X_{i}}$ 
(using internal language notation as in Lemma~\ref{lem:locally:contractive:slice}). 
Since each slice $\eslice{I}$ is also self-enriched,
this gives us two possible notions of local contractiveness. The next
lemma states a relation between the two.

\begin{lem} \label{lem:local:contractive:slices}
Let $\etopos$ be a locally cartesian closed model of guarded recursive terms, and let $F\co \eslice{I} \to \eslice{I}$ be a functor. If $F$ is locally contractive in the $\eslice{I}$-enriched sense then it is also locally contractive in the $\etopos$-enriched sense.
\end{lem}

\begin{proof}
The assumption gives an $\eslice{I}$-enrichment of $F$ as a composite 
\begin{diagram}
{\p Y}^{\p X} & \rTo^{\next} & \laterSlice{I} ({\p Y}^{\p X}) & \rTo^{\strengthof{G}{\p X}{\p Y}} & {\p {FY}}^{\p {FX}}
\end{diagram}
Lemma~\ref{lem:locally:contractive:slice} then tells us that each
$\Prod_{i \co I} F_{X_i, Y_i}$ is contractive in the
$\etopos$-enriched sense. To show that $F$ is locally contractive in
the $\etopos$-enriched sense one must check that the derived witness
of contractiveness commutes with composition and identity, but this
follows from naturality of the morphism $\later \Prod_{i\co I} X_i \to
\Prod_{i\co I} \later X_i$ used in
Lemma~\ref{lem:locally:contractive:slice}.

\end{proof}

\begin{proofof}{Theorem~\ref{thm:Sh:model:g:rec:types}}
  We have already shown (Theorem~\ref{thm:Sh:model:g:rec:terms}) that
  every slice of $\Sh{A}$ is a model of guarded recursive terms.
  It remains to show that any functor $F\co \Sh{A}/{I} \to \Sh{A}/{I}$, which
  is locally contractive in the $\Sh{A}/{I}$-enriched sense, has a
  fixed point. Since $\Sh{A}$ is complete~\cite[Prop.~III.4.4]{MacLane:Moerdijk:92},
  its slices $\Sh{A}/{I}$ are also complete and thus the required
  follows from Theorem~\ref{thm:Sh:model:g:rec:types:enr} and
  Lemma~\ref{lem:local:contractive:slices}.
\end{proofof}

\section{Conclusion and Future Work}
\label{sec:future-work}

We have shown that the topos of sheaves over a complete Heyting algebra with a well-founded basis,
in particular $\stopos$, the topos of trees, provides a model for an extension
of higher-order logic over dependent type theory with guarded
recursive types and terms. Moreover, we have argued that this logic
provides the right setting for the synthetic construction of
step-indexed models of programming languages and program logics, by
constructing a model of the programming language \Fmuref\ in the
logic.

In this paper we have focused solely on guarded recursion.  As future
work, it would be interesting to study further the connections
between guarded and unguarded recursion in $\stopos$. For example, it
might be possible to show the existence of recursive types in which
only negative occurrences of the recursion variable were guarded.

We plan to make a tool for formalized reasoning in the internal logic
of $\stopos$.  We have conducted some initial experiments by adding
axioms to Coq and used it to formalize some of the proofs
from~\cite{Schwinghammer:Birkedal:Stovring:11} involving recursively
defined relations on recursively defined types.  These experiments
suggest that it will be important to have special support for the
manipulation of the isomorphisms involved in recursive type equations,
such as the coercions and canonical structures
of~\cite{gonthier:structure}. 
An alternative approach, inspired by the conference version of the
present paper, has recently been proposed by Jaber
et.\ al.~\cite{Jaber:lics12}, who show how to internalize the
construction of the topos of trees in Coq and thus model guarded
recursive types.  Future work includes
investigating how easy or difficult it is in practice to develop and work with 
step-indexed models using that approach.

Future work also includes studying further applications of guarded
recursion in connection with step-indexed models. In particular, we
plan to give a synthetic account of a recent step-indexed model by the
first and third author for a language with countable
non-determinism~\cite{Schwinghammer:nondet}. That model uses
step-indexing over $\omega_1$, the first uncountable ordinal, so would
naturally live in sheaves over $\omega_1$. Indeed, this was part of the
motivation for generalizing the study of models of guarded recursion from
$\stopos$ to general sheaf categories $\Sh{A}$.

It could also be interesting to study predicative models of guarded
recursive dependent type theory, thus extending the work of Moerdijk
and Palmgren~\cite{palmgren-moerdijk,moerdijk-palmgren} on ``predicative
toposes''.

\paragraph{Acknowledgments}
We thank Andy Pitts and Paul Blain Levy for encouraging discussions.
This work was supported in part by grants from the Danish research council
(Birkedal and M{\o}gelberg) and from the Carlsberg Foundation (St{\o}vring).

{

}

\iflongversion
\clearpage
\appendix
\section{More details on the application to step-indexing}
\label{app:step-indexing}

Here are some more details on the application in
Section~\ref{sec:application}.  Everything is this appendix should be
understood within the logic of $\stopos$.

\subsection{Language}
\label{app:language}

The full language considered in the application is shown in
Figure~\ref{fig:language}.

\begin{figure*}
  \centering
  \begin{eqnarray*}
    \text{Types:} \qquad \tp & ::= & 1 \mid \tp_1
    \times \tp_2 \mid 0 \mid \tp_1 + \tp_2 \mid \mu \alpha. \tp
    \mid \forall \alpha. \tp \mid \alpha \mid \arrowtp{\tp_1}{\tp_2} \mid \reftp{\tp}
  \end{eqnarray*}
  \begin{eqnarray*}
    \text{Terms:} \qquad \tm & ::= & x \mid l \mid \unitcst \mid
    \pair{\tm_1}{\tm_2} 
    \mid \fst\tm \mid \snd\tm
\mid \void\tm \mid \inl\tm \mid \inr\tm
\\&&{} \mid
    \case{\tm_0}{x_1}{\tm_1}{x_2}{\tm_2} \mid \fold\tm \mid \unfold\tm
    \\&&{} \mid
    \tplam \alpha \tm \mid \tpapp \tm \tp \mid
    \lam x \tm \mid \tm_1 \ap \tm_2
    \mid \fixtm{f}{x}{\tm}
    \mid \reftm\tm \mid \lookup\tm \mid \assign{\tm_1}{\tm_2}
  \end{eqnarray*}
Typing rules:\hfill\strut\\[.3cm]

\newcommand{\ruleskip}{.7cm}

\begin{tabular}{cc}
\AxiomC{} \UnaryInfC{$\Xi \mid \Gamma \vdash x :
\tp$} \DisplayProof ($\Xi \vdash \Gamma$, $\Gamma(x) = \tp$) 
&
\AxiomC{} \UnaryInfC{$\Xi \mid \Gamma \vdash \unitcst : 1$} \DisplayProof ($\Xi \vdash
\Gamma$) 
\\[\ruleskip]
\AxiomC{$\Xi \mid \Gamma \vdash \tm_1 : \tp_1$}
\AxiomC{$\Xi \mid \Gamma \vdash \tm_2 : \tp_2$}
\BinaryInfC{$\Xi \mid \Gamma \vdash \pair{\tm_1}{\tm_2} :
\tp_1 \times \tp_2$} \DisplayProof 
&
\AxiomC{$\Xi \mid \Gamma \vdash \tm : 0$}
\UnaryInfC{$\Xi \mid \Gamma \vdash
\void \tm : \tp$}
\DisplayProof ($\Xi \vdash \tp$) 
\\[\ruleskip]
\AxiomC{$\Xi \mid \Gamma \vdash \tm : \tp_1 \times
\tp_2$} \UnaryInfC{$\Xi \mid \Gamma \vdash
\fst\tm : \tp_1$} \DisplayProof
&
\AxiomC{$\Xi \mid
\Gamma \vdash \tm : \tp_1 \times \tp_2$}
\UnaryInfC{$\Xi \mid \Gamma \vdash \snd\tm :
\tp_2$} \DisplayProof 
\\[\ruleskip]
\AxiomC{$\Xi \mid \Gamma \vdash \tm : \tp_1$}
\UnaryInfC{$\Xi \mid \Gamma \vdash
\inl \tm : \tp_1 + \tp_2$}
\DisplayProof ($\Xi \vdash \tp_2$) 
&
\AxiomC{$\Xi \mid \Gamma \vdash \tm : \tp_2$}
\UnaryInfC{$\Xi \mid \Gamma \vdash
\inr \tm : \tp_1 + \tp_2$}
\DisplayProof ($\Xi \vdash \tp_1$) 
\\[\ruleskip]
\multicolumn{2}{c}{%
\AxiomC{$\Xi \mid \Gamma \vdash \tm_0 : \tp_1 + \tp_2$}
\AxiomC{$\Xi \mid \Gamma, x_i :\tp_i \vdash \tm_i : \tp \quad (i = 1,2)$}
\BinaryInfC{$\Xi \mid \Gamma \vdash
  \case{\tm_0}{x_1}{\tm_1}{x_2}{\tm_2} : \tp$} \DisplayProof}
\\[\ruleskip]
\AxiomC{$\Xi \mid \Gamma \vdash \tm : \tp[\mu
\alpha.\tp / \alpha]$} \UnaryInfC{$\Xi \mid \Gamma
\vdash \fold \tm : \mu \alpha.\tp$}
\DisplayProof 
&
\AxiomC{$\Xi \mid \Gamma \vdash \tm : \mu \alpha.\tp$}
\UnaryInfC{$\Xi \mid \Gamma \vdash
\unfold \tm : \tp[\mu \alpha.\tp /
\alpha]$} \DisplayProof 
\\[\ruleskip]
\AxiomC{$\Xi,\alpha \mid \Gamma \vdash \tm : \tp$}
\UnaryInfC{$\Xi \mid \Gamma \vdash \tplam \alpha \tm :
\forall \alpha.\tp$} \DisplayProof ($\Xi
\vdash \Gamma$) 
&
\AxiomC{$\Xi \mid \Gamma \vdash \tm : \forall
\alpha.\tp_0$} \UnaryInfC{$\Xi \mid \Gamma \vdash
\tpapp{\tm}{\tp_1} : \tp_0[\tp_1 / \alpha]$} \DisplayProof ($\Xi \vdash
\tp_1$) 
\\[\ruleskip]
\AxiomC{$\Xi \mid \Gamma,
x:\tp_0 \vdash \tm : \comptp{\tp_1}$} \UnaryInfC{$\Xi \mid \Gamma
\vdash \lam{x}{\tm} : \arrowtp{\tp_0}{\tp_1}$} \DisplayProof 
&
\AxiomC{$\Xi \mid \Gamma \vdash \tm_1: \arrowtp{\tp}{\tp'}$}
\AxiomC{$\Xi \mid \Gamma \vdash \tm_2: \tp$}
\BinaryInfC{$\Xi \mid \Gamma \vdash \tm_1 \ap \tm_2 : \comptp{\tp'}$}
\DisplayProof
\\[\ruleskip]
\AxiomC{$\Xi \mid \Gamma, f: \arrowtp{\tp_0}{\tp_1},
x:\tp_0 \vdash \tm : \comptp{\tp_1}$} \UnaryInfC{$\Xi \mid \Gamma
\vdash \fixtm{f}{x}{\tm} : \arrowtp{\tp_0}{\tp_1}$} \DisplayProof 
&
\AxiomC{$\Xi \mid \Gamma \vdash \tm : \tp$}
\UnaryInfC{$\Xi \mid \Gamma \vdash \reftm\tm : \reftp\tp$}
\DisplayProof\\[\ruleskip]
\AxiomC{$\Xi \mid \Gamma \vdash \tm : \reftp{\tp}$}
\UnaryInfC{$\Xi \mid \Gamma \vdash \lookup\tm : \comptp\tp$}
\DisplayProof
&
\AxiomC{$\Xi \mid \Gamma \vdash \tm_1 : \reftp\tp$}
\AxiomC{$\Xi \mid \Gamma \vdash \tm_2 : \tp$}
\BinaryInfC{$\Xi \mid \Gamma \vdash \assign{\tm_1}{\tm_2} :
\comptp{1}$} \DisplayProof
\end{tabular}
\vskip.2cm

  \caption{Programming language}
  \label{fig:language}
\end{figure*}

\subsection{Interpretation of types}
\label{app:interpretation-types}

Recall that we have
\begin{align*}
  \W &= N \tofin \ST\\
  \STalt &= \W \tomon \Powset{\Value}\\
  \STalt^{\mathrm{c}} &= \W \to \Powset{\Term}
\end{align*}
and
\[
\App : \ST \to \STalt, \qquad \Lam : \STalt \to \ST
\]
with $\App \circ \Lam = \laterPred : \STalt \to \STalt$.

Let $\mathrm{TVar}$ be the set of type variables, and for $\tp \in
\mathrm{OType}$, let $\mathrm{TEnv}(\tp) = \{\,\varphi \in
\mathrm{TVar} \tofin \STalt \mid \fv{\tp} \subseteq
\dom{\varphi}\,\}$.  The interpretation of programming-language types
is defined by induction, as a function
\[
\den{\cdot} : \prod_{\tp \in \mathrm{OType}} \mathrm{TEnv(\tp)} \to
\STalt \, .
\]
  \begin{align*}
    \mngTp{\alpha}{\Xi}{\varphi} &= \varphi(\alpha)\\
    \mngTp{1}{\Xi}{\varphi} &= \lambda w.\, \{()\}\\
    \mngTp{0}{\Xi}{\varphi} &= \lambda w.\, \emptyset\\
    \mngTp{\tp_1 \times \tp_2}{\Xi}{\varphi} &= \lambda w.\,
\{\, \pairtm{v_1}{v_2} \mid
\begin{arraytl}v_1 \in \mngTp{\tp_1}{\Xi}{\varphi}(w) \mathrel{\land} 
v_2 \in \mngTp{\tp_2}{\Xi}{\varphi}(w) \, \}\end{arraytl}\\
    \mngTp{\tp_1 + \tp_2}{\Xi}{\varphi} &=
\lambda w.\, \begin{arraytl}\{\, \inl \ap v_1 \mid v_1 \in \mngTp{\tp_1}{\Xi}{\varphi}(w) \, \} \mathrel{\cup} 
\{\, \inr \ap v_2 \mid v_2 \in \mngTp{\tp_2}{\Xi}{\varphi}(w) \, \}\end{arraytl}\\
     \mngTp{\reftp \tp}{\Xi}{\varphi} &=
\lambda w.\, \{\, l \mid\begin{arraytl} l \in \dom{w} \mathrel{\land}
\forall w_1 \geq w .\,
\App(w(l)) (w_1) =
\laterPred(\mngTp{\tp}{\Xi}{\varphi})(w_1) \, \}\end{arraytl}\\
    \mngTp{\forall \alpha . \tp}{\Xi}{\varphi} &= \lambda w.\,\{\, \tplam{\alpha}{t}
    \mid\begin{arraytl} \forall \nu \in \STalt.\, \forall w_1 \geq w.\, 
    t \in \compOp(\mngTp{\tp}{\Xi,\alpha}{\varphi[\alpha \mapsto
      \nu]})(w_1) \, \}\end{arraytl}\\
    \mngTp{\mu \alpha . \tp}{\Xi}{\varphi} &= \mathit{fix}(\begin{arraytl}\lambda
      \nu. \, \lambda w. \, 
    \{\, \fold{v} \mid \laterPred(v \in
    \mngTp{\tp}{\Xi,\alpha}{\varphi[\alpha \mapsto \nu]} \ap (w))\,
    \})\end{arraytl}\\
    \mngTp{\arrowtp{\tp_1}{\tp_2}}{\Xi}{\varphi} &= 
\lambda w.\, \{\, \lamtm{x}{t} \mid
\begin{arraytl}\forall w_1 \geq w.\,
\forall v \in \mngTp{\tp_1}{\Xi}{\varphi}(w_1).\, 
t[v/x] \in \compOp(\mngTp{\tp_2}{\Xi}{\varphi})(w_1) \, \}
\end{arraytl}
\end{align*}
Here the operations $\compOp : \STalt \to \STalt^\mathrm{c}$ and
$\statesOp : \W \to \Powset{\Store}$ are given by
\begin{align*}
\compOp(\nu)(w) &= \{\, t \mid \begin{arraytl}\forall s \in \statesOp(w) .\,
\mathit{eval}(t,s,\, 
\lambda (v_1,s_1).\, 
\exists w_1 \geq w.\,\\ 
\quad 
v_1 \in \nu(w_1) ~\land~
s_1 \in \statesOp(w_1))
\,\}\end{arraytl}\\
\statesOp(w) &= \{\, s \mid
\begin{arraytl}\dom{s} = \dom{w} ~\mathrel\land~\\
\forall l \in \dom{w}.\, s(l) \in \App (w(l))(w) \, \}\,.\end{arraytl}
\end{align*}

\subsection{Soundness and the fundamental theorem}
\label{app:step-indexing-soundness}

Given $\Xi$ and $\Gamma$ such that $\Gamma$ is well-formed in $\Xi$, and
given $\varphi \in \STalt^\Xi$, define
\[
  \mngTp{\Gamma}{\Xi}{\varphi}(w) =
      \{\rho : \Value^{\dom{\Gamma}} | \forall (x,\tp) \in \Gamma.\, \rho(x) \in \mngTp{\tp}{\Xi}{\varphi}(w)\}.
\]
Abbreviate $\mngTpC{\tp}{\Xi}{\varphi} = \compOp(\mngTp{\tp}{\Xi}{\varphi})$.

Now we define semantic validity.  The notation
\[
  \semrel{\Xi}{\Gamma}{t}{\tp}
\]
means: For all $w \in W$, all $\varphi \in \STalt^\Xi$, and all $\rho \in \mngTp{\Gamma}{\Xi}{\varphi}(w)$, we have $\rho(t) \in \mngTpC{\tp}{\Xi}{\varphi}(w)$. (Here
$\rho(t)$ is $\rho$ acting by substitution on $t$.)

To show the fundamental theorem, we must show semantic counterparts of
all the typing rules.  First we need some ``monadic'' properties of
the $\compOp$ operator.
For $\nu \in \STalt$ and $\xi \in \STaltC$ and $w \in \W$, let
$\funrel{\nu}{\xi}{w}$ be the set of closed evaluation contexts $E$
that satisfy the following property: for all $w_1 \geq w$ and $v \in
\nu(w_1)$ we have $E[v] \in \xi(w_1)$.

\begin{lem}\hfill
\label{lem:comp-monadic}
  \begin{enumerate}[\em(1)]
  \item If\/ $v \in \nu(w)$, then $v \in \compOp(\nu)(w)$.
  \item If\/ $t \in \compOp(\nu_1)(w)$ and $E \in
    \funrel{\nu_1}{\compOp(\nu_2)}{w}$, then $E[t] \in \compOp(\nu_2)(w)$.
  \end{enumerate}
\end{lem}

\begin{proof}
  The first part follows immediately from the definitions of $\compOp$
  and $\eval$.  As for the second part, let $t \in \compOp(\nu_1)(w)$
  and $E \in \funrel{\nu_1}{\compOp(\nu_2)}{w}$ be given; we must show
  that $E[t] \in \compOp(\nu_2)(w)$.  We unfold the definition of
  $\compOp$.  Let $s \in \statesOp(w)$ be given; we must show that
  $\eval(E[t],s,Q)$ where
\[
Q(v_2,s_2) = \exists w_2 \geq w.\, v_2 \in \nu_2(w_2) ~\land~
s_2 \in \statesOp(w_2)).
\]
By Proposition~\ref{prop:eval-move-ectx}, it suffices to show
\begin{equation}
\label{eq:comp-monadic-proof}
\eval(t,s,\, \lambda (v_1,s_1). \, \eval(E[v_1],s_1, Q)).
\end{equation}

Since $t \in \compOp(\nu_1)(w)$ and $s \in \statesOp(w)$ we know that
\[
\eval(t,s,\, \begin{arraytl}\lambda(v_1,s_1).\, \exists w_1 \geq w.\\
  \quad \, v_1 \in \nu_1(w_1) ~\land~
s_1 \in \statesOp(w_1)).\end{arraytl}
\]
We can therefore use Proposition~\ref{prop:rule-of-consequence} to
show~\eqref{eq:comp-monadic-proof}: it suffices to show that $\exists w_1 \geq w.\, v_1 \in \nu_1(w_1) ~\land~
s_1 \in \statesOp(w_1))$ implies $\eval(E[v_1],s_1, Q)$.  So let $w_1
\geq w$ be given and assume that $v_1 \in \nu_1(w_1)$ and $s_1 \in
\statesOp(w_1)$.  Then, since $E \in \funrel{\nu_1}{\compOp(\nu_2)}{w}$,
we have $E[v_1] \in \compOp(\nu_2)(w_1)$ and hence
\[
\eval(E[v_1], s_1, \, \begin{arraytl}\lambda(v_2,s_2).\, \exists w_2
  \geq w_1.\\ \quad v_2 \in \nu_2(w_2) ~\land~
s_2 \in \statesOp(w_2)).\end{arraytl}
\]
Since $w_1 \geq w$, another use of Proposition~\ref{prop:rule-of-consequence}
gives $\eval(E[v_1],s_1, Q)$, which is what we had to show.
\end{proof}

\begin{proof}[Proof of Proposition~\ref{prop:fundamental} (fundamental theorem)]
We show four key cases.

\paragraph{Case {}``allocation'':}

If $\semrel{\Xi}{\Gamma}t{\tp}$, then $\semrel{\Xi}{\Gamma}{\reftm t}{\reftp{\tp}}.$ 

Let $w\in\W$ and $\varphi\in\STalt^{\Xi}$ and
$\rho\in\mngTp{\Gamma}{\Xi}{\varphi}$ be given; we must show that
$\rho(\reftm t)\in\mngTpC{\reftp{\tp}}{\Xi}{\varphi}(w)$.  Since
$\semrel{\Xi}{\Gamma}t{\tp}$ holds we know that
$\rho(t)\in\mngTpC{\tp}{\Xi}{\varphi}(w)$.  Therefore, by
Lemma~\ref{lem:comp-monadic}, it suffices to show that
$\reftm{\textnormal{\textrm{-}}}\in\funrel{\mngTp{\tp}{\Xi}{\varphi}}{\mngTpC{\reftp{\tp}}{\Xi}{\varphi}}w$.
To that end, let $w_{1}\geq w$ and
$v\in\mngTp{\tp}{\Xi}{\varphi}(w_{1})$ be given. We must show that
$\reftm v\in\mngTpC{\reftp{\tp}}{\Xi}{\varphi}(w_{1})$.

Let $s\in\statesOp(w_{1})$ be given. By definition of $\compOp$
we must show
\begin{displaymath}
\eval(\reftm v,s,\:\lambda(v_{1},s_{1}).\,\exists w_{2}\geq w_{1}.\,
v_{1}\in\mngTp{\reftp{\tp}}{\Xi}{\varphi}(w_{2})\land
s_{1}\in\statesOp(w_{2})).
\end{displaymath}
Let $l$ be the smallest location not in $s$. Then we have $step((\reftm v,s),(l,s_{1}))$
where $s_{1}=s[l\mapsto v]$. Therefore, by definition of $\eval$
and Proposition~\ref{prop:internal-logic-rules}, it suffices to show\[
\exists w_{2}\geq w_{1}.\, l\in\mngTp{\reftp{\tp}}{\Xi}{\varphi}(w_{2})\land s_{1}\in\statesOp(w_{2}).\]
(In fact, we are only required to show $\laterPred$ applied to this
formula, which is weaker by Proposition~\ref{prop:internal-logic-rules}(1).)
To that end, we choose $w_{2}=w_{1}[l\mapsto\Lam(\mngTp{\tp}{\Xi}{\varphi})]$.
It remains to show\begin{gather}
l\in\mngTp{\reftp{\tp}}{\Xi}{\varphi}(w_{2})\label{eq:case-alloc-one}\\
s_{1}\in\statesOp(w_{2}).\label{eq:case-alloc-two}\end{gather}

As for \eqref{eq:case-alloc-one}, we expand the definition of $\den{\reftp{\tp}}$.
Clearly we have $l\in dom(w_{2})$ as required. Now let $w_{3}\geq
w_{2}$ be
given; Lemma~\ref{lem:app-lam} gives
\begin{align*}
\App(w_{2}(l))(w_{3}) & =  \App(\Lam(\mngTp{\tp}{\Xi}{\varphi}))(w_{3})\\
 & =  \laterPred(\mngTp{\tp}{\Xi}{\varphi})(w_{3})
\end{align*}
as required.

As for \eqref{eq:case-alloc-two}, we first have that $\dom{s_{1}}=\dom{w_{2}}$
since $s\in\statesOp(w_{1}).$ Second, we must show that $s_{1}(l')\in\App(w_{2}(l'))(w_{2})$
for all $l'\in\dom{s_{1}}$. For $l'=l$ we have
$\App(w_{2}(l))(w_{2}) = \laterPred(\mngTp{\tp}{\Xi}{\varphi})(w_{2})$
as above. But $s_{1}(l)=v$, and we know that $v\in\mngTp{\tp}{\Xi}{\varphi}(w_{1})$
where\[
\mngTp{\tp}{\Xi}{\varphi}(w_{1})\subseteq\mngTp{\tp}{\Xi}{\varphi}(w_{2})\subseteq\laterPred(\mngTp{\tp}{\Xi}{\varphi})(w_{2})\]
by monotonicity and Proposition~\ref{prop:internal-logic-rules}(1). We conclude
that $s_{1}(l)\in\App(w_{2}(l))(w_{2})$.

For $l'\neq l$ we have $s_{1}(l')=s(l')$. Since $s\in\statesOp(w_{1})$
we know that $s(l')\in\App(w_{1}(l'))(w_{1}).$ But
\begin{align*}
\App(w_{1}(l'))(w_{1}) &= \App(w_{2}(l'))(w_{1})\\ &\subseteq \App(w_{2}(l'))(w_{2})
\end{align*}
by monotonicity. Therefore $s_{1}(l')\in\App(w_{2}(l'))(w_{2})$,
which completes the proof of \eqref{eq:case-alloc-two}.

\paragraph{Case {}``lookup'': }

If $\semrel{\Xi}{\Gamma}t{\reftp{\tp}}$ then $\semrel{\Xi}{\Gamma}{\lookup t}{\tp}$.

Let $w\in\W$ and $\varphi\in\STalt^{\Xi}$ and $\rho\in\mngTp{\Gamma}{\Xi}{\varphi}$
be given; we must show that $\rho(\lookup t)\in\mngTpC{\tp}{\Xi}{\varphi}(w)$.
Since $\semrel{\Xi}{\Gamma}t{\reftp{\tp}}$ holds we know that $\rho(t)\in\mngTpC{\reftp{\tp}}{\Xi}{\varphi}(w)$.
Therefore, by Lemma~\ref{lem:comp-monadic}, it suffices to
show that $\lookup -\in\funrel{\mngTp{\tp}{\Xi}{\varphi}}{\mngTpC{\reftp{\tp}}{\Xi}{\varphi}}w$.
This is essentially what was done in the proof sketch in the main
text, but for completeness we repeat the argument here. 

Let $w_1 \geq w$ and $v \in \mngTp{\reftp \tp}{\Xi}{\varphi}(w_1)$ be given.  We
must show that $\lookup v \in
\compOp(\mngTp{\tp}{\Xi}{\varphi})(w_1)$ We unfold the
definition of $\compOp$.  Let $s \in \statesOp(w_1)$ be given; we must
show
\begin{equation}
\label{eq:fundamental-theorem-lookup}
\mathit{eval}(\lookup v,s,\,\lambda (v_2,s_2).\, 
\exists w_2 \geq w_1.\,
v_2 \in \mngTp{\tp}{\Xi}{\varphi}(w_2) ~\land~ s_2 \in \statesOp(w_2))\,.
\end{equation}
By the assumption that $v \in \mngTp{\reftp
  \tp}{\Xi}{\varphi}(w_1)$, we know that $v = l$ for some
location $l$ such that $l \in \dom{w_1}$ and $\App (w_1(l)) (w_2) =
\laterPred(\mngTp{\tp}{\Xi}{\varphi})(w_2)$ for all $w_2 \geq w_1$.
Since $s \in \statesOp(w_1)$, we know that $l \in \dom{s} = \dom{w_1}$ and
$s(l) \in \App(w_1(l))(w_1)$.  We therefore have $\step((\lookup v, s),
(s(l), s))$.  Hence, by unfolding the definition of $\eval$ in
\eqref{eq:fundamental-theorem-lookup} and using the rules from
Proposition~\ref{prop:internal-logic-rules}, it
remains to show that
\[
\exists w_2 \geq w_1.\,
\laterPred(s(l) \in \mngTp{\tp}{\Xi}{\varphi}(w_2)) ~\land~
\laterPred(s \in \statesOp(w_2)).
\]
To that end, choose $w_2 = w_1$.  First, $s \in \statesOp(w_1)$ and hence
$\laterPred(s \in \statesOp(w_1))$.  Second, 
\[
s(l) \in \App(w_1(l))(w_1) =
\laterPred(\mngTp{\tp}{\Xi}{\varphi})(w_1),
\]
which means exactly that
$\laterPred(s(l) \in \mngTp{\tp}{\Xi}{\varphi}(w_1))$.

\paragraph{Case {}``assignment'':}

If $\semrel{\Xi}{\Gamma}{t_{1}}{\reftp{\tp}}$ and $\semrel{\Xi}{\Gamma}{t_{2}}{\tp}$,
then $\semrel{\Xi}{\Gamma}{\assign{t_{1}}{t_{2}}}1$.

Here we must use Lemma~\ref{lem:comp-monadic} twice. Let $w\in\W$
and $\varphi\in\STalt^{\Xi}$ and $\rho\in\mngTp{\Gamma}{\Xi}{\varphi}$
be given; we must show that \[\rho(\assign{t_{1}}{t_{2}})\in\mngTpC 1{\Xi}{\varphi}(w).\]
Since $\semrel{\Xi}{\Gamma}{t_{1}}{\reftp{\tp}}$ holds we know that
$\rho(t_{1})\in\mngTpC{\reftp{\tp}}{\Xi}{\varphi}(w)$. Therefore,
by Lemma~\ref{lem:comp-monadic}, it suffices to show that \[(\assign -{\rho(t_{2})})\in\funrel{\mngTp{\reftp{\tp}}{\Xi}{\varphi}}{\mngTpC 1{\Xi}{\varphi}}w.\]
So let $w_{1}\geq w$ and $v_{1}\in\mngTp{\reftp{\tp}}{\Xi}{\varphi}(w_{1})$
be given; we must show that $\assign{(v_{1}}{\rho(t_{2}}))\in\mngTpC 1{\Xi}{\varphi}(w_{1})$.
By assumption we have $\rho(t_{2})\in\mngTpC{\tp}{\Xi}{\varphi}(w_{1})$,
so by Lemma~\ref{lem:comp-monadic} again, it suffices to show
that \[(\assign{v_{1}}-)\in\funrel{\mngTp{\tp}{\Xi}{\varphi}}{\mngTpC 1{\Xi}{\varphi}}{w_{1}}.\]
Therefore, let $w_{2}\geq w_{1}$ and $v_{2}\in\mngTp{\tp}{\Xi}{\varphi}(w_{2})$
be given. The final proof obligation is to show that\[
(\assign{v_{1}}{v_{2}})\in\mngTpC 1{\Xi}{\varphi}(w_{2}).\]

We unfold the definition of $\compOp$. Assume that $s\in\statesOp(w_{2})$
is given; we must show
\begin{displaymath}
\eval((\assign{v_{1}}{v_{2}}),s,\:\lambda(v_{3},s_{3}).\,\exists
w_{3}\geq w_{2}.\,
 v_{3}\in\mngTp 1{\Xi}{\varphi}(w_{3}) ~\land~
s_{3}\in\statesOp(w_{3})).
\end{displaymath}

By monotonicity we have $v_{1}\in\mngTp{\reftp{\tp}}{\Xi}{\varphi}(w_{2})$,
and therefore $v_{1}=l$ for some $l\in\dom{w_{2}}$ such that \begin{equation}
\App(w_{2}(l))(w_{3})=\laterPred(\mngTp{\tp}{\Xi}{\varphi})(w_{3})
\quad \text{for all $w_{3}\geq w_{2}$.} \label{eq:ft-case-assignment}\end{equation}
Furthermore, since $s\in\statesOp(w_{2})$
we know that $\dom s=\dom{w_{2}}$ and hence that $l\in\dom s$. Therefore
$\step((\assign{v_{1}}{v_{2}},\, s),\:(\unitcst,\, s[l\mapsto v_{2}]))$
holds. By definition of $\eval$ and Proposition~\ref{prop:internal-logic-rules},
it then suffices to show\[
\exists w_{3}\geq w_{2}.\,\unitcst\in\mngTp 1{\Xi}{\varphi}(w_{3}) ~\land~ s[l\mapsto v_{2}]\in\statesOp(w_{3}).\]

We choose $w_{3}=w_{2}$. Now $\unitcst\in\mngTp 1{\Xi}{\varphi}(w_{2})$
holds trivially, and it remains to show that $s[l\mapsto v_{2}]\in\statesOp(w_{2})$.
For $l'\neq l$ we have \[(s[l\mapsto v_{2}])(l')=s(l')\in\App(w_{2}(l'))(w_{2})\]
since $s\in\statesOp(w_{2})$. Furthermore, \[(s[l\mapsto v_{2}])(l)=v_{2}\in\mngTp{\tp}{\Xi}{\varphi}(w_{2}),\]
and therefore $\laterPred(v_{2}\in\mngTp{\tp}{\Xi}{\varphi}(w_{2}))$
by Proposition~\ref{prop:internal-logic-rules}(1). But this means exactly
that $v_{2}\in\laterPred(\mngTp{\tp}{\Xi}{\varphi})(w_{2})$. We conclude
from~\eqref{eq:ft-case-assignment} that $v_{2}\in\App(w_{2}(l))(w_{2})$
as required.

\paragraph{Case {}``unfold'':}

If $\semrel{\Xi}{\Gamma}t{\mu\alpha.\tp}$, then $\semrel{\Xi}{\Gamma}{\unfold t}{\tp[(\mu\alpha.\tp)/a]}.$

Abbreviate $\tp_{1}=\tp[(\mu\alpha.\tp)/a]$. Let $w\in\W$ and $\varphi\in\STalt^{\Xi}$
and $\rho\in\mngTp{\Gamma}{\Xi}{\varphi}$ be given; we must show
that $\rho(\unfold t)\in\mngTpC{\tp_{1}}{\Xi}{\varphi}(w)$. Since
$\semrel{\Xi}{\Gamma}t{\mu\alpha.\tp}$ holds we know that $\rho(t)\in\mngTpC{\mu\alpha.\tp}{\Xi}{\varphi}(w)$.
Therefore, by Lemma~\ref{lem:comp-monadic}, it suffices to
show that $\unfold -\in\funrel{\mngTp{\mu\alpha.\tp}{\Xi}{\varphi}}{\mngTpC{\tp_{1}}{\Xi}{\varphi}}w$.
To that end, let $w_{1}\geq w$ and $v\in\mngTp{\mu\alpha.\tp}{\Xi}{\varphi}(w_{1})$
be given. We must show that $\unfold v\in\mngTpC{\tp_{1}}{\Xi}{\varphi}(w_{1})$.

Let $s\in\statesOp(w_{1})$ be given. By definition of $\compOp$
we must show
\begin{equation}
\eval(\unfold v,s,\:\lambda(v_{1},s_{1}).\,\exists w_{2}\geq w_{1}.\,
v_{1}\in\mngTp{\tp_{1}}{\Xi}{\varphi}(w_{2}) ~\land~
s_{1}\in\statesOp(w_{2})).\label{eq:fl-case-unfold}
\end{equation}

By definition of $\den{\mu\alpha.\tp}$ we know that $v=\fold{v_{0}}$
for some $v_{0}$ such that $\laterPred(v_{0}\in\mngTp{\tp}{\Xi,\alpha}{\varphi[\alpha\mapsto\mngTp{\mu\alpha.\tp}{\Xi}{\varphi}]}(w_{1}))$
holds. By Proposition~\ref{prop:internal-logic-rules} and a substitution lemma
(shown by an easy induction on types), this means that $\laterPred(v_{0}\in\mngTp{\tp_{1}}{\Xi}{\varphi}(w_{1}))$
holds.

Since $v=\fold{v_{0}}$ we have $\step((\unfold v,s),(v_{0},s))$.
Therefore, by unfolding the definition of $\eval$ in \eqref{eq:fl-case-unfold}
and using Proposition~\ref{prop:internal-logic-rules}, it suffices to show\[
\exists w_{2}\geq w_{1}.\,\laterPred(v_{0}\in\mngTp{\tp_{1}}{\Xi}{\varphi}(w_{2}))\land\laterPred(s\in\statesOp(w_{2})).\]
We choose $w_{2}=w_{1}$. We have already shown that 
$\laterPred(v_{0}\in\mngTp{\tp_{1}}{\Xi}{\varphi}(w_{1}))$ holds, and Proposition~\ref{prop:internal-logic-rules}(1) gives that $s\in\statesOp(w_{1})$
implies $\laterPred(s\in\statesOp(w_{1}))$, as required.
\end{proof}

As an immediate corollary of the fundamental theorem we get a
type-safety result for the ``temporal'' semantics given by the $\eval$
predicate.  This is formulated by means of a trivial post-condition.

\begin{cor}[Type safety]
\label{cor:type_safety}
Assume that\/ $\vdash t : \tp$ holds.  Then
$\eval(t,s_{\mathrm{init}},\top)$ holds where $s_{\mathrm{init}}$ is
the empty store.  
\end{cor}

\begin{proof}
  Follows directly from the fundamental theorem (using the empty world
  $\emptyset \in \W$) and Proposition~\ref{prop:rule-of-consequence}.
\end{proof}

\else
\fi

\end{document}